\documentclass[letterpaper,11pt]{article}

\usepackage[utf8]{inputenc}
\usepackage[margin=1in]{geometry}
\usepackage{hyperref}
\hypersetup{
           breaklinks=true,   
           colorlinks=false,   
 }
\usepackage{amssymb,amsmath, amsthm}
\usepackage{verbatim}
\usepackage{dsfont}
\usepackage{algorithm}
\usepackage{algorithmic}
\usepackage{thmtools,thm-restate}

\usepackage{thmtools}
\usepackage{mathtools}
\usepackage{enumerate}
\usepackage{paralist}
\usepackage{hhline}

\usepackage{lineno}

\newcommand{\tdeg}{\mathcal{T}_{\deg}} 
 
\newcommand{\target}{\textsc{target}}

\usepackage{enumitem} 
\newlist{condenum}{enumerate}{1} 
\setlist[condenum]{label=(F\arabic*)., 
ref=(F\arabic*), wide}

\usepackage{booktabs} 
\usepackage{subcaption}
\usepackage{tikz}
\usetikzlibrary{decorations.pathmorphing}
\usetikzlibrary{decorations.markings}
\usetikzlibrary{shapes.geometric}
\usetikzlibrary{shapes.gates.logic.US,trees,positioning,arrows}
\usetikzlibrary{calc}

\newcommand{\spp}{\text{Sparsity}} 

\newcommand{\lca}{\textsc{LCA}}

\newcommand{\Q}{\widehat{Q}}

\newcommand{\m}{\widehat{m}}

\newcommand{\CT}{\text{COST}}
\newcommand{\wCT}{\text{WCOST}}

\newcommand{\lv}{\textsc{leaves}}
\newcommand{\est}{O\left(\frac{1}{\varphi^2}\right)\cdot  \wCT_{\wt{H}}(\wt{T}) + \textsc{TotalClustersCost}(G)+ O\left(\frac{\xi m n k^2 }{\varphi^2} \right) }
\newcommand{\contrib}{\textsc{contrib}}

\newcommand{\eqdef}{:=}

\newcommand{\rsm}{\textsc{RandomSampleModel}}

\newcommand{\estg}{\textsc{EST}}
\newcommand{\estgen}{\frac{1}{a} \cdot \wCT_{\wt{H}}(\wt{T}) + \frac{b}{a} mn + \textsc{TotalClustersCost}(G)}

\newcommand{\zext}{z_{\text{ext}}}

\newcommand{\runningtimeiv}{ O^*\left(n^{1/2+O(\e/\varphi^2)}\cdot \left(\frac{1}{\xi}\right)^{O(1)}\right)}

\newcommand{\rank}{\mathrm{rank}}

\newcommand{\cC}{\mathcal{C}}
\newcommand{\cD}{\mathcal{D}}

\newcommand{\cL}{\mathcal{L}}

\newcommand{\cT}{\mathcal{T}}

\newcommand{\G}{\mathcal{G}}

\newcommand{\E}{\mathbb{E}}

\newcommand{\R}{\mathbb{R}}

\newcommand{\inner}[2]{\langle #1, #2 \rangle}

\theoremstyle{plain}
\newtheorem{theorem}{Theorem}[section]

\newtheorem{lemma}[theorem]{Lemma}

\theoremstyle{definition}
\newtheorem{definition}[theorem]{Definition}

\theoremstyle{remark}
\newtheorem{remark}[theorem]{Remark}

\newtheorem{lemmma}{Lemma}
\newtheorem{mydef}{Definition}
\newtheorem{mycor}{Corollary}
\newtheorem{observation}{Observation}
\newtheorem{claim}{Claim}
\newtheorem{fact}{Fact}

\newcommand{\e}{\epsilon}

\newcommand{\poly}{\text{poly}}
\newcommand{\vol}{\text{vol}}

\newcounter{todocounter}

\newcommand{\adp}[1]
  {\ensuremath{\left\langle #1 \right\rangle_{\scriptscriptstyle apx}}}
\newcommand{\rdp}[1]
  {\ensuremath{\left\langle #1 \right\rangle}}

\newcommand{\wt}[1]{\widetilde{#1}}
\newcommand{\return}{\textbf{return }}

\newcommand{\mycolor}[1]{\textcolor{black}{#1}}

\DeclareMathOperator*{\argmin}{arg\,min}
  
\def\polylog{\operatorname{polylog}}

\usepackage[utf8]{inputenc} 
\usepackage[T1]{fontenc}    
\usepackage{hyperref}       
\usepackage{url}            
\usepackage{booktabs}       
\usepackage{amsfonts}       
\usepackage{nicefrac}       
\usepackage{microtype}      
\usepackage{xcolor}         

\title{Approximating Dasgupta Cost in Sublinear Time\\ from a Few Random Seeds}

\author{%
Michael Kapralov\\EPFL \and
Akash Kumar\\IIT Bombay \and 
Silvio Lattanzi\\Google Research \and
Aida Mousavifar \\Google \and 
Weronika Wrzos-Kaminska \\EPFL}

\date{}
\begin{document}

\maketitle

\begin{abstract}
 Testing graph cluster structure has been a central object of study in property testing since the foundational work of Goldreich and Ron [STOC'96] on expansion testing, i.e. the problem of distinguishing between a single cluster (an expander) and a graph that is far from a single cluster. 
 More generally, a $(k, \e)$-clusterable graph $G$ is a graph whose vertex set admits a partition into $k$ induced expanders, each with outer conductance bounded by $\e$.  A recent line of work initiated by Czumaj, Peng and Sohler [STOC'15] has shown how to test whether a graph is close to $(k, 
\e)$-clusterable, and to locally determine which cluster a given vertex belongs to with misclassification rate $\approx \e$, but no sublinear time algorithms for learning the structure of inter-cluster connections are known. As a simple example, can one locally distinguish between the `cluster graph' forming a line and a clique? 

 In this paper, we consider the problem of testing the {\em hierarchical} cluster structure of $(k, \e)$-clusterable graphs in sublinear time. Our measure of hierarchical clusterability is the well-established Dasgupta cost, and our main result is an algorithm that approximates Dasgupta cost of a $(k, \e)$-clusterable graph in sublinear time, using a small number of randomly chosen {\em seed vertices} for which cluster labels are known. Our main result is an $O(\sqrt{\log k})$ approximation to Dasgupta cost of $G$ in  $\approx n^{1/2+O(\e)}$ time using $\approx n^{1/3}$ seeds, effectively giving a sublinear time simulation of the algorithm of Charikar and Chatziafratis [SODA'17] on clusterable graphs. To the best of our knowledge, ours is the first result on approximating the hierarchical clustering properties of such graphs in sublinear time.
 
\end{abstract}


 \vspace{-10pt}
\section{Introduction}
\label{sec:intro}

Graph clustering is a central problem in data analysis, with applications in a wide variety of scientific disciplines from data mining to social science, statistics and more. The overall objective in these problems is to partition
the vertex set of the graph into disjoint ``well connected'' subgraphs which are sparsely connected to each other.  It is quite common in the practice of graph clustering that besides the graph itself one is given a list of vertices with correct cluster labels for them, and one must extend this limited amount of cleanly labeled data to a clustering of the entire graph. This corresponds to the widely used {\em seeded} model (see, e.g.,~\cite{basu2002semi} and numerous follow up works, e.g.,~\cite{demiriz1999semi, kulis2009semi, sinkkonen2002clustering, AshtianiB15}). The central question that we consider in this paper is 

\vspace{0.05in}

\fbox{
\parbox{0.9\textwidth}{
\begin{center}
	What can be learned about the cluster structure of the input graph from a few seed nodes in sublinear time?
\end{center}
}
}
\vspace{0.05in}

Formally, we work with the classical model for well-clusterable graphs~\cite{CzumajPS15}, where the input graph $G = (V, E)$ is assumed to admit a partitioning into a disjoint union of $k$ induced expanders $C_1,\ldots, C_k$ with outer conductance bounded by $\epsilon\ll 1$ and inner conductance being $\Omega(1)$. We refer to such instances as \textit{$(k, \Omega(1), \e)$-clusterable}  graphs, or $(k, \e)$-clusterable  graphs for short. Such graphs have been the focus of significant attention in the property testing literature~\cite{CzumajS07,KaleS08,NachmiasS10}, starting from the seminal work of \cite{GoldreichR11}. A recent line of work has shown how to design nearly optimal sublinear time clustering oracles for such graphs, i.e., algorithms that can consistently answer clustering queries on such a graph from a local exploration only. However, existing works do not show how to learn the structure of connections between the clusters. In particular, to the best of our knowledge, no approach in existing literature can resolve the following simple question:

\vspace{0.05in}

\fbox{
\parbox{0.9\textwidth}{
\begin{center}
Distinguish between the clusters being arranged in a line 
and the clusters forming an (appropriately subsampled) clique (See Fig.~\ref{fig:linevsclique})
. 
\end{center}
}
}

\vspace{0.05in}
\begin{figure}[H]
	\centering
		    \begin{tikzpicture}[scale=1, every node/.style={circle, draw, minimum size=1cm}]
        \node[fill=gray!30] (A) at (0,2)  {\(C_1\)};
        \node[fill=gray!30] (B) at (2,2) {\(C_2\)};
        \node[fill=gray!30] (C) at (4,2) {\(C_3\)};
        \node[fill=gray!30] (D) at (6,2) {\(C_4\)};
        \node[fill= gray!30] (E) at (8,2) {\(C_k\)};
  
        \draw[bend left=15] (A) to (B);
                \draw[bend left=25] (A) to (B);
        \draw[bend left=5] (A) to (B);
        \draw[bend right=5] (A) to (B);
        \draw[bend right=15] (A) to (B);
        \draw[bend right=25] (A) to (B);
        \draw[bend left=25] (B) to (C);
        \draw[bend left=15] (B) to (C);
        \draw[bend left=5] (B) to (C);
        \draw[bend right=5] (B) to (C);
        \draw[bend right=15] (B) to (C);
        \draw[bend right=25] (B) to (C);

             \draw[bend left=25] (C) to (D);
        \draw[bend left=15] (C) to (D);
        \draw[bend left=5] (C) to (D);
        \draw[bend right=5] (C) to (D);
        \draw[bend right=15] (C) to (D);
             \draw[bend right=25] (C) to (D);
                   \draw[bend left=25] (D) to (E);
      \draw[bend left=15] (D) to (E);
        \draw[bend left=5] (D) to (E);
        \draw[bend right=5] (D) to (E);
        \draw[bend right=15] (D) to (E);
                \draw[bend right=25] (D) to (E);
        
        \node[fill=gray!30] (X) at ({2*cos(90)+12}, {2*sin(90)+2}) {\(C_1\)};
        \node[fill=gray!30] (Y) at ({2*cos(162)+12}, {2*sin(162)+2}) {\(C_2\)};
        \node[fill=gray!30] (Z) at ({2*cos(234)+12}, {2*sin(234)+2}) {\(C_3\)};
        \node[fill=gray!30] (W) at ({2*cos(306)+12}, {2*sin(306)+2}) {\(C_4\)};;
        \node[fill=gray!30] (V) at ({2*cos(378)+12}, {2*sin(378)+2}) {\(C_k\)};;
        
        \foreach \i/\j in {X/Y, Y/Z, Z/W, W/V, V/X, X/Z, X/W, Y/W, Y/V, Z/V} {
            \draw[bend left = 10] (\i) to (\j);
             \draw (\i) to (\j);
            
            \draw[bend right = 10] (\i) to (\j);
        }
    \end{tikzpicture}
	  \caption{Clusters arranged in a line (Left);  Clusters forming a clique (Right)}
	\label{fig:linevsclique}
\end{figure}

 More generally, we would like to design a sublinear time algorithm that approximates the {\em hierarchical clustering} properties of $k$-clusterable graphs.  Hierarchical clustering is a useful primitive in machine learning and data science with essential applications in information retrieval \cite{ManningRS, Berkhin06}, social networks \cite{GilbertSZJB} and phylogenetics \cite{EisenSBB}. Informally, in hierarchical clustering the objective is to construct a hierarchy of partitions that explain the cluster structure of the graph at different scales -- note that such a partitioning looks very different in the two cases (line vs. clique) above.  Formally, the quality of such a hierarchy of partitions is often evaluated using Dasgupta cost~\cite{dasgupta2016cost}, and the main question studied in our paper is

\vspace{0.05in}

\fbox{
\parbox{0.9\textwidth}{
\begin{center}
	Is it possible to approximate the Dasgupta cost of a $(k, \Omega(1), \e)$-clusterable graph using few queries to the input graph and a few correctly clustered seed vertices?
\end{center}
}
}

\vspace{0.05in}
In practice, an algorithm operating in the seeded model~\cite{basu2002semi} most often does not have full control over the seeds, but rather is given a list generated by some external process. To model this, we assume that the seed vertices are sampled independently from the input graph, with probability proportional to their degrees: we refer to this model as the {\em random sample model}. 

\vspace{0.05in}

\paragraph{The case $k=1$, i.e., approximating the Dasgupta cost of an expander.}  
When $k=1$, our input is a single expander, i.e., a single cluster, we approximate its Dasgupta cost in sublinear time using degree queries on the seeds. At first glance one might think that Dasgupta cost of an expander can be approximated well simply as a function of its number of vertices and average degree, but this is only the case for {\em regular} expanders. The irregular case is nontrivial, a $\poly(1/\varphi)$ approximation was recently given  by~\cite{MSun21}.  As our first result, we give an algorithm approximating Dasgupta cost of an (irregular) $\varphi$-expander using $\approx n^{1/3}$ seed vertices (and degree queries on these vertices). This, somewhat surprisingly, turns out to be a tight bound. Specifically, we show

\begin{theorem}[Approximating Dasgupta cost of an expander]\label{thm:iwehgiwehgh}
Dasgupta cost of a $\varphi$-expander can be approximated to within a $\text{poly}(1/\varphi)$ factor using degree queries on $\approx n^{1/3}$ seed vertices. Furthermore, the bound of $\approx n^{1/3}$ is tight up to polylogarithmic factors.
\end{theorem}

\paragraph{The case $k>1$}  For $ k >1$  we leverage recent results on clustering oracles to decompose the problem of approximating the Dasgupta cost into two: {\bf (1)} approximating Dasgupta cost of individual clusters and {\bf (2)} approximating Dasgupta cost of the contracted graph, in which each cluster is contracted into a supernode. 
Such a decomposition is only possible for bounded degree graphs, see Example 4.2 in {\cite{MSun21}}, so this is the setting we work in. We show that access to a few seed vertices is sufficient to obtain oracle access to the cut function (and, more generally, quadratic form of the Laplacian) of the contracted graph in time $\approx n^{1/2+O(\e)}$.   Our main result is Theorem~\ref{thm:infprepros-tree} below:Theorem~\ref{thm:infprepros-tree} below:

\begin{restatable}[Informal version of Theorem~\ref{thm:estdcost}]{theorem}{treepreproc}\label{thm:infprepros-tree}
There exists an algorithm that for every  $(k, \Omega(1), \e)$-clusterable bounded degree graph $G = (V, E)$ estimates the Dasgupta cost of $G$ up to $O(\sqrt{\log k})$ factor in the random sample model in time $\approx n^{1/2+O(\epsilon)}\cdot (d_{\max})^{O(1)}$. 
\end{restatable}

\begin{remark}
We remark that our algorithm for estimating Dasgupta cost from Theorem~\ref{thm:estdcost} can be made to provide an oracle access to a low cost hierarchical clustering tree.
\end{remark}

\begin{remark}
One can verify by adapting the lower bound of $\Omega(n^{1/2})$ on expansion testing due to Goldreich and Ron~\cite{GR02}  that at least $\Omega(\sqrt{n/k})$ queries are needed for a $o(k/\log k)$ approximation for constant $k$ in this model. The proof is a rather direct adaptation of the classical result of Goldreich and Ron, and we therefore do not present it.    
\end{remark}

\begin{remark}
Recall that in our {\em random sample} model for seed vertices the seeds are sampled independently with probability proportional to their degrees.  This model matches quite closely what happens in practice in the sense that the algorithm does not always have full control over the seeds~\cite{basu2002semi}. One can also consider the stronger model in which the algorithm can ask for correct label of {\em any} vertex of its choosing. This model is significantly stronger, and in particular, one can design an algorithm for obtaining the same approximation of Dasgupta cost as our Theorem~\ref{thm:infprepros-tree} above, but with time complexity polynomial in $d$, $\log n$ and $1/\epsilon$.
\end{remark}

We note that the currently best known approximation to the Dasgupta cost on $n$-vertex graphs is $O(\sqrt{\log n})$, achieved by the recursive sparsest cut algorithm of \cite{charikar2017approximate}. Our approximation is $O(\sqrt{\log k})$, matching what the Charikar and Chatziafratis algorithm achieves on $k$-node graphs. In fact, our main technical contribution is an efficient way of simulating this algorithm in sublinear time on $k$-clusterable graphs.

{\bf Related work on $(k, \Omega(1), \e)$-clusterable graphs.} Such graphs have been extensively studied in the property testing framework as well as local computation models. Its testing version, where one essentially wants to determine $k$, the number of clusters in $G$, in sublinear time, generalizes the well-studied problem of testing graph
expansion, where one wants to distinguish between an expander (i.e. a good single cluster) and a graph with a sparse cut (i.e., at least two clusters).  \cite{GoldreichR11} showed that expansion testing requires $\Omega(n^{1/2})$ queries, then \cite{CzumajS07,KaleS08,NachmiasS10} developed algorithms to distinguish an expander from a graph that is far from a graph with conductance $\epsilon$ in time $\approx n^{1/2+O(\epsilon)}$, which the recent work of~\cite{chiplunkar2018testing} showed to be tight.
The setting of $k>2$ has seen a lot of attention recently~\cite{CzumajPS15, chiplunkar2018testing,Peng20,GluchKLMS21}, with close to information theoretically optimal clustering oracles, i.e., small space data structures that provide quick access to an approximate clustering, obtained in~\cite{GluchKLMS21}. More recently, \cite{MSun21} studied hierarchical clustering of $k$-clusterable graphs and developed a nearly 
linear time algorithm that approximates the Dasgupta cost of the graph up to a constant factor. However, their algorithm to work requires significantly stronger assumptions on the input data i.e., $\epsilon \ll 1/k^{O(1)}$, and their algorithm does not run in {\em sublinear time}. Note that the problem of estimating the Dasgupta cost becomes non-trivial when $\epsilon\gg \frac{1}{k}$, i.e., when the Dasgupta cost of the graph is dominated by the outgoing 
edges between different clusters\footnote{\label{foot-cost} For instance, in a $d$-regular, $(k, \varphi, \e)$-clusterable graph, one can easily show that the Dasgupta cost is at least $\Omega(\frac{\varphi\cdot d\cdot n^2}{k})$, simply because of the contribution of the $k$ induced $\varphi$-expanders. On the other hand, the total number of edges running between the clusters is bounded by $\e\cdot d\cdot n$, and therefore their total contribution to the Dasgupta cost is $O(\e \cdot d\cdot n^2)$. Thus, the problem becomes non-trivial when $\epsilon\gg \frac{1}{k}$.}.

The most closely related work on our setting is~\cite{KKLM23} where the authors provide a sublinear algorithm for hierarchical clustering. However, their algorithm works under significantly stronger assumptions on their input instance. They introduce the notion of hierarchically clusterable graphs, which assumes a planted hierarchical clustering structure not only at the bottom level of the hierarchy but at \emph{every level}. Their result relies on several properties of such graphs. In contrast, we only assume that the input graph is $k$-clusterable. For this reason we cannot use the techniques developed in~\cite{KKLM23},  and we need to develop a completely new approach.

{Very recently, \cite{AKLP, ACLMW22} considered the problem of hierarchical clustering under Dasgupta 
objective in the streaming model. Both papers give a one pass $\widetilde{O}(n)$ memory streaming 
algorithm which finds a tree with Dasgupta cost within an $O(\sqrt{\log n})$ factor of the optimum 
in polynomial time. Additionally, \cite{AKLP} also considers this problem in the query model 
and presents an $O(\sqrt{\log n})$  approximate hierarchical clustering using $\widetilde{O}(n)$ 
queries without making any clusterability assumptions of the input graph. On the other hand, our 
algorithms assume the graph is $k$-clusterable and approximate the Dasgupta cost within an $O(\sqrt{\log k})$ 
in sublinear time.}

{\bf Related work on hierarchical clustering.} 
We briefly review developments in the area of algorithms for hierarchical clustering
since the introduction of Dasgupta's objective function.
Dasgupta designed an algorithm based on recursive sparsest-cut that provides $O(\log ^{3/2} n)$ 
approximation for his objective function. This was improved by Charikar and Chatizafratis
who showed that the recursive sparsest-cut algorithm already returns a tree with approximation 
guarantee $O(\sqrt{\log n})$ \cite{charikar2017approximate}. Furthermore, they showed that it's 
impossible to approximate the Dasgupta cost within a constant factor in general graphs under 
the Small-Set Expansion hypothesis.
More recently, \cite{cohen2018hierarchical} studied this problem in a regime in which the input graph is sampled from 
a Hierarchical Stochastic Block Model \cite{cohen2018hierarchical}. They construct a tree in nearly linear time that approximates Dasgupta cost of the graph up to a constant factor. \cite{cohen2018hierarchical}  uses a type of hierarchical stochastic block model, which generates close to regular expanders with high probability, and their analysis crucially relies on having dense clusters and large degrees. Our model allows for arbitrary expanders as opposed to dense random graphs and is more expressive in this sense.

{\bf Related work in semi-supervised active clustering.}
We note that our model is also related to the semi-supervised active clustering framework (SSAC) introduced in~\cite{AshtianiKB16}. In this model we are given a set $X$ of $n$ points and an oracle answering to same-cluster queries of the form ``are these two points in the same cluster?''. Thanks to its elegance and applications to crowdsourcing, the model received a lot of attention and has been extensively studied both in theory~\cite{ailon2018approximate, ailon2017approximate, BCLP20, BCLP21-density, huleihel2019same, mazumdar2017semisupervised, mazumdar2017clustering, NIPS2017_7054, saha2019correlation, vitale2019flattening} and in practice~\cite{firmani2018robust, gruenheid2015fault, verroios2015entity, verroios2017waldo} --- see also \cite{emamjomeh2018adaptive} for other types of queries. 

\vspace{-5pt}

\subsection{Basic definitions}

\begin{mydef}[Inner and outer conductance] \label{def:inner-outer-conductance}
Let $G = (V, E) $ be a graph. For a set $C\subseteq V$ and a set $S\subseteq C$, let $E(S,C\setminus S)$ be the set of edges with one endpoint in $S$ and the other in $C\setminus S$. The \textit{conductance of $S$ within $C$}, is $ \phi^G_C(S)=\frac{|E(S, C\setminus S)|}{\vol(S)} $. The \textit{outer conductance} of $C$ is defined to be 
$\phi^G_{\text{out}}(C)=\phi_V^G(C)=\frac{|E(C,V\setminus C)|}{\vol(C)}\text{.}$ The \textit{inner conductance} of $C\subseteq V$ is defined to be  $\phi_{\text{in}}^G(C)=\min_{S\subseteq C\text{,} 0<|S|\leq \frac{\vol(C)}{2}}\phi^G_C(S)$ if $|C|>1$ and one otherwise. 
\end{mydef}
For Theorem {\ref{thm:infprepros-tree}}, we assume that the degree of every vertex is maximal by adding self-loops, and use the notion of conductance corresponding to the graph with the added self-loops. We define $k$-\textit{clusterable} graphs as a class of instances that can be partitioned into $k$ expanders with small outer conductance:

\begin{restatable}[$(k,\varphi,\epsilon)$-clustering]{mydef}{clusterable}\label{def:clusterable}
Let $G=(V,E)$ be a graph. A $(k,\varphi,\epsilon)$-clustering of $G$ is a partition of vertices $V$ into disjoint subsets $C_1, \ldots, C_k$ such that for all $i\in [k]$, $\phi_{\text{in}}^G(C_i)\geq\varphi$, $\phi_{\text{out}}^G(C_i)\leq\epsilon$ and for all $i,j \in [k]$ one has $\eta \coloneqq \frac{|C_i|}{|C_j|} \in O(1).$ 
A graph $G$ is called $(k,\varphi,\epsilon)$-clusterable if there exists a $(k,\varphi,\epsilon)$-clustering for $G$.
\end{restatable}

\noindent {\bf Dasgupta cost.} Hierarchical clustering is the task of partitioning vertices of a graph into nested clusters. 
The nested partitions can be represented by a rooted tree whose leaves correspond to the vertices 
of graph, and whose internal nodes represent the clusters of vertices. Dasgupta 
introduced a natural optimization framework for formulating hierarchical clustering tasks as an 
optimization problem \cite{dasgupta2016cost}. We recall this framework now. Let $T$ be any rooted tree whose leaves are vertices of the graph. For any node $x$ of $T$, let $T[x]$ be the subtree rooted at $x$, and let $\lv(T[x]) \subseteq V$ denote the leaves of this subtree. For leaves $x,y \in V$ , let $\lca(x,y)$ denote the lowest common ancestor of $x$ and $y$ in $T$. In other words, $T[\lca(x,y)]$ is the smallest subtree whose leaves contain both $x$ and $y$. 

\begin{restatable}[Dasgupta cost \cite{dasgupta2016cost}]{mydef}{defdasgupta} 
\label{def:dasgupta-cost}
The Dasgupta cost of the tree $T$ for the graph $G = (V, E)$ is defined to be
$\CT_G(T) = \sum_{\{x,y\}\in E} |\lv(T[\lca(x, y)])|\text{.}$
\end{restatable}

\noindent {\bf The random sample model for seed vertices.} We consider a  {\em random sample} model for seed vertices, in which the algorithm is given a (multi)set $S$ of {\em seed} vertices, which are sampled independently with probability proportional to their degrees, together with their cluster label.

\section{Technical overview}

In this section, we give an overview of our main algorithmic result, stated below as Theorem \ref{thm:estdcost} (formal version of Theorem \ref{thm:infprepros-tree}). 
It postulates a sublinear time algorithm for estimating the Dasgupta cost of $k$-clusterable graphs. Here, we use $O^*$-notation to suppress $\poly(k)$, $\poly(1/\varphi)$, $\poly(1/\e)$  and $\polylog n$-factors. 

\begin{restatable}{theorem}{thmestdcost} \label{thm:estdcost}
	Let $k \geq 2$, $\varphi \in (0,1)$ and $\frac{\epsilon}{\varphi^2}$ be a sufficiently small constant. Let $G=(V,E)$ be a bounded degree graph that admits a $(k,\varphi,\epsilon)$-clustering $C_1, \ldots , C_k$. Let $|V|=n$. 
 
	There exists an algorithm ($\textsc{EstimatedCost}(G)$; Algorithm 
	\ref{alg:tree:mn}) that
	w.h.p. estimates the optimum Dasgupta cost of $G$ within an $O\left(\frac{\sqrt{\log k}}{\varphi^{O(1)}}\right)$ factor in time 
	$O^*\left(n^{1/2+O(\e/\varphi^2)} \cdot  (d_{\max})^{O(1)}\right)$ using $O^*\left(n^{O(\e/\varphi^2)}\cdot (d_{\max})^{O(1)}\right)$ seed queries.
\end{restatable}

Our algorithm consists of two main parts: First, we estimate the contribution from the inter-cluster edges to the Dasgupta cost. A natural approach is to contract the clusters $C_1, \ldots, C_k$ into supernodes, and use the Dasgupta cost of the contracted graph (defined below) as a proxy.
\begin{restatable}[Contracted graph]{mydef}{defcontract}
\label{def:contracted}
\label{def:weighetdcontracted}
Let $G=(V,E)$ be a graph and let $\cC = (C_1, \ldots , C_k)$ denote a partition of $V$ into disjoint subsets.  We say that the weighted graph $H=\left([k], {[k] \choose 2}, W, w\right)$ is a contraction of $G$ with respect to the partition $\cC$ if for every $i,j\in [k]$ we have $W(i,j)=|E(C_i, C_j)|$, and for every $i\in[k]$ we have $w(i)=|C_i|$. 
We denote the contraction of $G$ with respect to the partition $\cC $ by  $H = G/\cC$.
\end{restatable}
The problem is of course that it is not clear how to get access to this contracted graph in sublinear time, and our main contribution is a way of doing so. 
Our approach  amounts to first obtaining access to the quadratic form of the Laplacian of the contracted graph $H$, and then using the hierarchical clustering algorithm of \cite{charikar2017approximate} 
on the corresponding approximation to the contracted graph. Thus, we essentially show how to simulate the algorithm of~\cite{charikar2017approximate} in sublinear time on $(k, \varphi, \e)$-clusterable graphs.

The procedure \textsc{TotalClustersCost} approximates the contribution from the internal cluster edges to the Dasgupta cost. 

Algorithm \ref{alg:tree:mn} below presents our estimator for the Dasgupta cost of the graph.

\begin{algorithm}[H]
\caption{$\textsc{EstimatedCost}(G)$ {\qquad  \qquad\qquad\qquad\qquad \qquad \qquad \qquad  time  $\approx n^{1/2+O(\epsilon)}$}
}
\label{alg:tree:mn}
\begin{algorithmic}[1]
\STATE $\xi \gets \left(\frac{\varphi}{k\cdot d_{\max}}\right)^{O(1)}$
\STATE $\mathcal{D} \gets \textsc{InitializeWeightedDotProductOracle($G,\xi$)}$ \textit{\qquad\qquad\ \qquad \quad  \# See Algorithm \ref{alg:LearnEmbedding-tt}}
\STATE $\widetilde{H} \gets\textsc{ApproxContractedGraph}(G, \xi,\mathcal{D})$ \textit{\quad \ \#The Laplacian $\widetilde{\mathcal{L}}$ of $\widetilde{H}$ satisfies Equation \eqref{eq:l-lh-spec}}
\STATE $\widetilde{T} \gets \textsc{WeightedRecursiveSparsestCut}(\widetilde H)$ \textit{\qquad  \#Weighted version of Algorithm of~\cite{charikar2017approximate}}
\STATE $\estg \gets \est$
\STATE \return \estg
\end{algorithmic}
\end{algorithm}
Our algorithm uses a weighted definition of Dasgupta cost (Definition~\ref{def:weightedcost}), which we denote $\wCT$, to relate the cost of $G$ and the contracted graph $H$. 
Then, our estimate EST in Algorithm \ref{alg:tree:mn} simply sums the contribution from the weighted Dasgupta cost of the tree $\widetilde{T}$ on the contracted graph $\widetilde{H}$, with the contribution from the clusters. We want to ensure that the estimate always provides an upper bound on the optimal Dasgupta cost of $G$. To this end, we scale the weighted Dasgupta cost $\wCT_{\widetilde{H}}(\widetilde{T})$ up by a factor of $O\left(\frac{1}{\varphi^2} \right)$ (to account for the multiplicative error), and add a term on the order of $\frac{\xi mn k^2}{\varphi^2}$ (to account for the additive error). That way we obtain an estimate EST such that

\[\CT(G) \leq \text{EST} \leq O\left(\frac{\sqrt{ \log k}}{\varphi^{O(1)}} \right)\CT(G), \]
where $\CT(G) $ denotes the optimum Dasgupta cost of $G$. 

We outline the main ideas behind accessing the contracted graph in Section~\ref{sec:overview}, and present the complete analysis in Section~\ref{sec:estimate-cost}. The\textsc{
TotalClustersCost} procedure simply outputs a fixed value that depends on $n$, $d$, and $k$. We provide more details on this in Section~\ref{sec:expandercost-overview} and present to full analysis in Section \ref{sec:clustercost}. 

We remark that with a little post-processing, our algorithms for estimating Dasgupta cost can be adapted to recover a low-cost hierarchical-clustering tree. To construct such a tree we first construct a tree $\widetilde{T}$ with $k$ leaves on the contracted graph. The algorithm constructs $\widetilde{T}$ in sublinear time $\approx n^{1/2+O(\epsilon)}$.
Then, for every cluster $C_i$ one can construct a particular tree $\mathcal{T}^i_{\deg}$ using (Algorithm $1$ of \cite{MSun21}) on the vertices of $C_i$. Finally, we can extend the leaf $i$ of the tree $\widetilde{T}$ by adding trees $\mathcal{T}^i$ as its direct child. Note that constructing $\mathcal{T}^i_{\deg}$ explicitly takes time $O(|C_i|)$, however, this step is only required if one intends to output the full hierarchical-clustering tree of $G$. Otherwise, for only estimating $\CT(G)$, we can estimate $\CT(\mathcal{T}^i_{\deg})$ as a function of the cluster size and the degree without explicitly constructing $\mathcal{T}^i_{\deg}$.

\vspace{-5pt}

\subsection{Estimating Dasgupta cost of an expander using seed queries}\label{sec:expandercost-overview}
In this section, we design an algorithm for estimating the Dasgupta cost of an irregular $\varphi$-expander up to $\poly(1/\varphi)$ factor using $\approx n^{1/3}$ seed queries. We also prove that this is optimal (Theorem \ref{thm:lw-Gk}) in subsection \ref{sec:tightrunningtime}. Later in the paper (in Section \ref{sec:clustercost}), we approximate the contribution of the clusters to the Dasgupta cost of a $d$-regular $(k,\varphi, \epsilon)$-clusterable graph. There, a more basic approach suffices. In this section, we focus on a single but irregular $\varphi$-expander.

\begin{restatable}{theorem}{singleclustercost}
	\label{thm:singleclustercost}
	Let $G=(V,E)$ be a $\varphi$-expander (possibly with self-loops).  Let $T^*$ denote the tree
	with optimum Dasgupta cost for $G$. 
	Then procedure \textsc{ClusterCost} (Algorithm \ref{alg:clustercontr}), uses $O^*\left(n^{1/3}\right)$ seed queries and with probability $1-n^{-101}$ returns a value such that:
	\[ \CT(T^*) \leq \textsc{ClusterCost}(G) \leq O\left(\frac{1}{\varphi^{5}}\right)\cdot  \CT(T^*) \text{.} \]
\end{restatable} 

We now outline the proof of Theorem \ref{thm:singleclustercost}.  

Let $G$ be a $\varphi$-expander, i.e., $\phi_{\text{in}}(G)\geq \varphi$. To estimate the Dasgupta cost of $G$, we use Theorem \ref{thm:MSun21:thm3} from \cite{MSun21}. This result shows that there is a specific tree called $\mathcal{T}_{\deg}$ on $G$ that approximates the Dasgupta cost of $G$ up to $O\left(\frac{1}{\varphi^4}\right)$.  For completeness, we include the algorithm (Algorithm \ref{alg:HCdeg}) for computing $\mathcal{T}_{\deg}$ from \cite{MSun21}. Note that Algorithm \ref{alg:HCdeg} from \cite{MSun21} runs in time $O(m+n\log{n})$, however, we don't need to explicitly construct $\mathcal{T}_{\deg}$. Instead,
we design an algorithm that estimates the cost of $\mathcal{T}_{\deg}$ in time $n^{1/3}$. 

\begin{algorithm}[H]
	\caption{\textsc{HCwithDegrees($G\{V\}$)} \cite{MSun21}}
	\label{alg:HCdeg}
	\begin{algorithmic}[1]
		\STATE \textbf{Input}: $G=(V, E, w)$ with the ordered vertices such that $d_{v_1}\geq \ldots \geq d_{v_{|V|}}$  
		\STATE \textbf{Output}: An HC tree $\mathcal{T}_{\deg}(G)$
		\IF{$|V|=1$}
		\STATE \textbf{return} the single vertex $V$ as the tree
		\ELSE
		\STATE $i_{\max} \coloneqq \lfloor \log_2{|V|-1} \rfloor$; $r \coloneqq 2^{i_{\max}}$; $A \coloneqq \{v_1,\ldots,v_r\}$; $B \coloneqq V\textbackslash A$
		\STATE Let $\mathcal{T}_1 \coloneqq $ \textsc{HCwithDegrees($G\{A\}$)}; $\mathcal{T}_2 \coloneqq $ \textsc{HCwithDegrees($G\{B\}$)}
		\STATE \textbf{return} $\mathcal{T}_{\deg}$ with $\mathcal{T}_1$ and $\mathcal{T}_2$ as the two children
		\ENDIF
	\end{algorithmic}
\end{algorithm}

\begin{restatable}[Theorem 3 in \cite{MSun21}]{theorem}{mstedg}\label{thm:MSun21:thm3}
	Given any graph $G=(V, E, w)$ with inner-conductance $\varphi$ as input, Algorithm \ref{alg:HCdeg} runs in $O(m+n\log{n})$ time, and returns an $HC$ tree 
	$\mathcal{T}_{\deg}(G)$ that satisfies $\CT_G(\mathcal{T}_{\deg}(G))=O(1/\varphi^4) \cdot {OPT}_G$.  
\end{restatable}
Our procedure for estimating the Dasgupta cost of the tree returned by Algorithm~\ref{alg:HCdeg} is based on a simple expression for the (approximate) cost of this tree that we derive (and later show how to approximate by sampling). 

Let $G = (V,E)$ be an arbitrary expander with vertices $x_1, x_2, \ldots x_n$ ordered such that $d_1 \geq d_2 \geq \ldots \geq d_n$, where $d_i = \text{deg}(x_i)$. We denote
by $\tdeg$ the Dasgupta Tree returned by Algorithm 1 of \cite{MSun21}. Specifically, we show that the cost  $\CT_G(\tdeg)$ is to within an $O(1/\varphi)$ factor approximated by
\begin{equation}\label{eq:iewh9u23rfeqwrf}
	\sum_{i=1}^{n} i \cdot d_i=\sum_{x\in V} \rank(x) \cdot \deg(x),
\end{equation}
where $\deg(x)$ is the degree of $x$ and $\rank(x)$ is the rank of $x$ in the ordering of vertices of $V$ in non-increasing order of degrees. The proof is rather direct, and is presented in the appendix (Lemmas~\ref{lem:dcost:lb} and~\ref{lem:dcost:ub}). Our task therefore reduces to approximating~\eqref{eq:iewh9u23rfeqwrf} in sublinear time. To achieve this, we partition the vertices into buckets according to their degree:  For every $d$ between $1$ and $n/\varphi$ that is a power of $2$, let $B_d := \{x \in V : d \leq \deg(x) <2d\}.$  We will refer to $B_d$ as the \emph{degree class} of $d$. Let $n_d:=|B_d|$ denote the size of the degree class, and let $r_d$ denote the highest rank in $B_d$.  Note that $r_d$ is the number of vertices in $G$ that have degree at least $d$, so we have $r_d= \sum_{t \geq d} n_t$. 

The vertices in $B_d$ have ranks $r_d, r_{d}-1,\dots, r_d -n_d+1$ and degrees in $[d,2d]$, which gives the bounds
\begin{equation}\label{eqn:rdnd:tech}
    \frac{d}{2} \cdot n_d \cdot r_d \leq \sum_{i = r_d -n_d+1}^{r_d} i \cdot d \leq \sum_{x \in B_d} \rank(x) \cdot \deg(x) \leq \sum_{i = r_d -n_d+1}^{r_d} i \cdot 2d \leq 2 d \cdot n_d \cdot r_d,
\end{equation}
so our task is further reduced to estimating the quantity 
\begin{equation}\label{eqn:estquantity}
    \sum_{d} d \cdot n_d \cdot r_d. 
\end{equation}
We do so by sampling: simply sample $\approx n^{1/3}$ vertices, and approximate the number of vertices $n_d$ and the highest rank $r_d$ of each degree class. This is summarized in Algorithm~\ref{alg:clustercontr} below. 
\begin{algorithm}[H]
	\caption{\textsc{ClusterCost$(G, S,\hat{m})$}
		\newline \textit{\# $S$ is a (multi)set of size $s$ of vertices in $G=(V, E)$ }
		\newline \textit{\#	$\hat{m}$ is a constant factor estimate of $|E|$} }
	\label{alg:clustercontr}
	\begin{algorithmic}[1]
            \FOR{every $d$ between $1$ and $n/\varphi$ that is a power of $2$}
            \STATE $\hat{n}_d \gets \frac{2\hat{m}}{s}|\{ v \in S: d \leq \deg(v) < 2d\}|$ \qquad \quad \textit{\# Estimate the number of vertices by sampling}
            \STATE $\hat{r}_d \gets \frac{2\hat{m}}{s}|\{ v \in S: d \leq \deg(v) \}|$ \qquad \qquad \quad  \textit{ \# Estimate the rank by sampling}
		\ENDFOR
		\STATE \return  $\sum_d \hat{n}_d \cdot \hat{r}_d$
	\end{algorithmic}
\end{algorithm}
While the algorithm is simple, the analysis is quite interesting, and the bound of $n^{1/3}$ on the number of seeds is tight! We now outline the main ideas behind the analysis of the algorithm.

Ideally, we would like to estimate the number of vertices $n_d$ and the highest rank $r_d$ of every degree class $d$. However, this is hard to achieve, as some degree classes may be small. The crux of the analysis is showing that with $\approx n^{1/3}$ samples, we can approximate $n_d$ and $r_d$ for any degree class that contributes at least a $(1/\log n)$-fraction of the Dasgupta cost. 

 Recalling that our model assumes degree proportional sampling, the expected number of samples from any degree class $B_t$ is
\[\frac{s}{2m} \sum_{x \in B_t} \deg (x)  \approx \frac{s}{m} n_t \cdot t, \] where $s$ is the total number of samples. 
Thus, we can estimate $n_t$ and $r_t$ whenever
$n_t \cdot t \geq \Omega^*\left(\frac{m}{s}\right).$

Now, consider the degree class $d$ with the highest degree mass. Since there are at most $\log n/\varphi$ different degree classes,  we have $n_d \cdot d \geq \Omega^*(m).$
Thus, we can estimate the contribution to the Dasgupta cost of any degree class $B_t$  which satisfies 

\begin{equation*}\label{eq:doscoverclass}
 \frac{n_d \cdot d}{n_t \cdot t} \leq O^*(s). 
\end{equation*}

Using the degree class $d$ as a reference, we show that any degree class $t$ that has a significant contribution to the Dasgupta cost, must have a sufficiently large degree mass compared to $d$. 

Specifically, if $B_t$ is a degree class that contributes at least  a $(1/\log n)$-fraction of the Dasgupta cost, i.e. 
\[ \sum_{x \in B_t} \rank(x) \deg(x) \geq \frac{1}{\log n} \sum_{x \in V} \rank(x) \deg(x), \]
then, by Equation \eqref{eqn:rdnd:tech}, we have 
\[2t \cdot r_t \cdot n_t \geq \sum_{x \in B_t} \rank(x) \deg(x) \geq  \frac{1}{\log n} \sum_{x \in V} \rank(x) \deg(x) \geq  \frac{1}{\log n} \sum_{i=1}^{r_t} i \cdot t \geq  \frac{1}{\log n} \cdot \frac{r_t^2}{2} \cdot t.\]
From this, we conclude that $n_t \gtrapprox r_t$, allowing us to use the quantity $n_t ^2 \cdot t$ as a further proxy for the contribution of $B_t$ to the Dasgupta cost.

Furthermore, we show that if $d$ is our high-degree-mass reference class and $t$ is any degree class that contributes at least a $1/\log n $ fraction of the Dasgupta cost, then $n_t^2 \cdot t \geq n_d^2 \cdot d$. Intuitively, this is because the contribution from $B_t$ is no smaller than the contribution from $B_d$. 

Therefore, the following optimization problem provides an upper bound on the sufficient number of samples. 
\begin{align*}
      \max_{t,n_t,d,n_d} \frac{n_d \cdot d}{n_t \cdot t}  & \quad  \\
      \text{ such that}&\\
       n_t^2 \cdot t & \geq  n_d^2\cdot  d \qquad \qquad \#\text{$B_t$  has large contribution to the Dasgupta Cost} \\
      &   n_t, n_d \leq n \qquad \quad \ \#\text{at most $n$ vertices} \\
    &   n_t,t, d \geq 1 \qquad \quad   \#\text{$B_t$ is non-empty and degrees are non-zero} \\
      &   n_d\geq 0 . 
\end{align*}

However, the above optimization problem is too weak. For example, setting $n_d = n^{1/2}$, $d = n$, $n_t = n$, $t=1$ gives a feasible solution with value $n^{1/2}$. But this solution would correspond to having $n^{1/2}$ vertices of degree $n$ and $n$ vertices of degree $1$, which is impossible in an actual graph. 
We remedy this by adding an additional constraint that encodes that  $t, n_t, d, n_d $ arise from a valid graph. 
\begin{align*}
      \max_{t,n_t,d,n_d} \frac{n_d \cdot d}{n_t \cdot t}  & \quad  \\
      \text{ such that}&\\
       n_t^2 \cdot t & \geq  n_d^2\cdot  d  \qquad \qquad \#\text{$B_t$  has large contribution to the Dasgupta cost} \\
            & d \leq n_d  \qquad \qquad \ \   \#\text{$B_d$ does not have too many edges to $V \setminus B_d$}\\
       &   n_t, n_d \leq n \qquad \quad  \ \#\text{at most $n$ vertices} \\
    &   n_t,t,d \geq 1 \qquad \quad  \#\text{$B_t$ is non-empty and degrees are non-zero} \\
      &   n_d\geq 0 . 
\end{align*}

A priori, there is no reason why the constraint $d \leq n_d$ should be satisfied by our reference class $B_d$. However, we show that for any graph, it is possible to find a reference class $B_d$ which satisfies $d \leq n_d$ and contributes a large fraction of the degree mass. Intuitively, this is because if all the high-degree-mass classes had $d> n_d$, then they would require too many edges to be routed outside of their degree class, eventually exhausting the available vertices. See proof of Lemma \ref{lemma:heavy_and_good} in Section \ref{sec:rankdeg} for the details.

Finally, we prove that the refined optimization problem has optimal value $\approx n^{1/3}$. Therefore $\approx n^{1/3}$ samples suffice to discover any degree class $t$ with a non-trivial contribution to the Dasgupta cost. The full analysis is presented in Appendix \ref{sec:exp-cost}. 

We also show that $\Omega(n^{1/3})$ seeds are necessary to approximate $\sum_{x\in V} \rank(x)\cdot \deg(x)$ to within any constant factor:

\begin{restatable}{theorem}{unionclusterableruntimetight}
	\label{thm:lw-Gk}
	For every positive constant $\alpha>1$ and $n$ sufficiently large, there exists a pair of expanders $G$ and $G'$ such that $\sum_{i=1}^n i \cdot d_i  \leq  n^2$, $\sum_{i=1}^n i \cdot d'_i \geq \alpha n^2$ and at least $\Omega(n^{1/3})$ vertices need to be queried in order to have probability above $2/3$ of distinguishing between them (where $d_1  \geq ... \geq d_n \geq 1$ is the degree sequence in $G$ and $d'_1 \geq ...\geq  d'_n \geq 1$ is the degree sequence in $G'$). 
\end{restatable}
Figure \ref{fig:lb} below illustrates the graphs $G$ and $G'$ from Theorem \ref{thm:lw-Gk}. Graph $G$ has of a set $A$ of $n^{2/3}$ vertices of degree $n^{2/3}$, and the remaining vertices have degree $1$. Graph $G’$ has a set $A’$ of $n^{2/3}$ vertices of degree $n^{2/3}$, but the remaining vertices have degree $\alpha$. In order to distinguish the two graphs, we need to query a vertex outside of $A$ or $A'$, but this requires  $\Omega(n^{1/3})$ queries in expectation. The proof of Theorem \ref{thm:lw-Gk} is straightforward, and is included in Section \ref{sec:tightrunningtime}. 

\begin{figure}[H]
	\centering
	\begin{subfigure}[b]{0.45\textwidth}
		\centering
		\begin{tikzpicture}
    \def\n{8} 
    \def\radius{1cm} 
    \def\outerradius{2cm} 

    \foreach \i in {1,...,\n}
    {
        \node[draw, circle, inner sep=2pt, fill=black] (C\i) at ({360/\n * (\i - 1)}:\radius) {};
    }

    \foreach \i in {1,...,\n}
    {
        \foreach \j in {1,...,\n}
        {
            \ifnum\i<\j
                \draw (C\i) -- (C\j);
            \fi
        }
    }

    \foreach \i in {1,...,\n}
    {
        \node[draw, circle, inner sep=2pt, fill=black] (O\i a) at ({360/\n * (\i - 1) + 15}:\outerradius) {};
        \node[draw, circle, inner sep=2pt, fill=black] (O\i b) at ({360/\n * (\i - 1) - 15}:\outerradius) {};
        \draw (C\i) -- (O\i a);
        \draw (C\i) -- (O\i b);
    }
        \node[font=\Large] at (0, -\outerradius-1cm) {$G$};
\end{tikzpicture}
		\label{fig:lb_G}
	\end{subfigure}
	\hfill
	\begin{subfigure}[b]{0.45\textwidth}
		\centering
		\begin{tikzpicture}
    \def\n{8} 
    \def\radius{1cm} 
    \def\outerradius{2cm} 

    \foreach \i in {1,...,\n}
    {
        \node[draw, circle, inner sep=2pt, fill=black] (C\i) at ({360/\n * (\i - 1)}:\radius) {};
    }

    \foreach \i in {1,...,\n}
    {
        \pgfmathtruncatemacro{\next}{mod(\i,\n)+1} 

        \foreach \j in {\next,...,\n}
        {
            \draw (C\i) -- (C\j);
        }
    }

    \foreach \i in {1,3,...,\n}
    {
        \pgfmathtruncatemacro{\prev}{mod(\i-2+\n,\n)+1} 
        \draw[white, line width=3pt] (C\i) -- (C\prev); 
    }

    \foreach \i in {1,...,\n}
    {
        \node[draw, circle, inner sep=2pt, fill=black] (O\i a) at ({360/\n * (\i - 1) + 15}:\outerradius) {};
        \node[draw, circle, inner sep=2pt, fill=black] (O\i b) at ({360/\n * (\i - 1) - 15}:\outerradius) {};
        
        \pgfmathtruncatemacro{\prev}{mod(\i-2+\n,\n)+1}
        \pgfmathtruncatemacro{\next}{mod(\i,\n)+1}
        
        \draw (C\i) -- (O\i a);
        \draw (C\prev) -- (O\i a);

        \draw (C\i) -- (O\i b);
        \draw (C\next) -- (O\i b);
    }

    \node[font=\Large] at (0, -\outerradius-1cm) {$G'$};

\end{tikzpicture}
		\label{fig:lb_G_prime}
	\end{subfigure}
	\caption{Illustration of the two instances in Theorem \ref{thm:lw-Gk}.}
	\label{fig:lb}
\end{figure}

Finally, we describe our procedure \textsc{TotalClustersCost} for approximating the total contribution of the clusters to the Dasgupta cost of a $d$-regular $(k,\varphi, \epsilon)$-clusterable graph. To approximate the contribution of a single cluster $C_i$, it suffices to apply the formula $ \sum_{x \in C_i} \rank(x) d = \frac{|C_i|(|C_i|+1)}{2}d \approx |C_i|^2d.$ The procedure \textsc{TotalClustersCost} simply sums up these contributions from all the clusters. See Section \ref{sec:clustercost} for more details. 

\subsection{Sublinear time access to the contracted graph} \label{sec:overview}
In this section, we consider a $(k,\varphi, \epsilon)$-clusterable graph with bounded maximum degree. For simplicity of presentation, we assume without loss of generality that the graph is $d$-regular. This is because, by a standard reduction, we can convert a degree $d$-bounded graph into a $d$-regular graph by adding self-loops to each vertex.

We denote the Laplacian of $G$ by $\cL_G$ and the normalized Laplacian of $G$ by $L_G$. 
We will use the following notation for additive-multiplicative approximation.
\begin{mydef}
For $x,y \in \mathbb{R}$, write 
$$x \approx_{a,b} y \text{\qquad \qquad if }a^{-1}\cdot y - b \leq x \leq a \cdot y+b.$$
For matrices $X,Y \in \mathbb{R}^{k \times k}$, write
$$X \approx_{a,b} Y \text{\qquad \qquad if }a^{-1}\cdot Y - b\cdot m \cdot I_k \preceq X \preceq a \cdot  Y+b\cdot m \cdot I_k.$$
\end{mydef}
Let $H = G/\cC$ be the contraction of $G$ with respect to the underlying clustering $\cC= (C_1, C_2, \ldots C_k)$ (Definition \ref{def:contracted}). We write $H = ([k], \binom{[k]}{2},W)$ to emphasize that the vertex set of the contracted graph $H$ corresponds 
to the clusters of $G$ and for $i,j \in [k]$, the pair $(i,j)$ is an edge of $H$ with weight
$W(i,j) = |E(C_i, C_j)|$. 
If we were explicitly given the adjacency/Laplacian matrix
of the contracted graph $H$, then finding a good Dasgupta tree for $H$ can be easily done by using
the algorithm of \cite{charikar2017approximate} which gives a $O(\sqrt{\log k})$
approximation to the optimal tree for $H$ (and an $\sqrt{\log k}/\varphi^{O(1)}$ approximation
to the optimal tree for $G$ as shown in Theorem~\ref{thm:estdcost}).

The problem is that we do not have explicit access to the Laplacian of the contracted graph (denoted by $\cL_H$).
However, to get a good approximation to the Dasgupta Cost of $H$,
it suffices to provide explicit access to a Laplacian $\wt{L}$ (which corresponds to a graph $\wt{H}$) 
where cuts in $\wt{H}$ approximate sizes of corresponding cuts in $H$ 
 in the following sense: 
$\exists \alpha > 1, \beta>0$ and such that for all $S \subseteq [k]$, 
$$
 |\widetilde{E}(S,V\setminus S)| \approx_{\alpha, \beta} |E(S,V\setminus S)|. 
$$
Motivated by this observation, we simulate this access approximately  
 by constructing a matrix $\widetilde{\cL}$ which spectrally approximates $\cL_H$ in the sense that
\begin{equation}
\label{eq:l-lh-spec}
\cL_H \approx_{a, \xi} \widetilde{\cL} 
\end{equation}
in time $\approx m^{1/2+O(\e/\varphi^2)} \cdot \text{poly}(1/\xi)$ for some $0<\xi<1<a$ (see Theorem \ref{thm:contracted-graph}).
So, our immediate goal is to spectrally approximate $\cL_H$. We describe this next. \\

\vspace{-5pt}
\noindent {\bf Spectrally approximating $\cL_H$:} The key insight behind our spectral approximation $\widetilde{\cL}$ to $\cL_H$ comes from considering the case where our graph is a collection of $k$ disjoint expanders each on $n/k$ vertices. To understand this better, let $L_G = U \Lambda U^T$ denote the eigendecomposition of the Laplacian and let $M = U \Sigma U^T$ denote the eigendecomposition of the lazy random walk matrix. Letting $U_{[k]} \in \R^{n \times k}$ denote a matrix whose columns are the first $k$ columns of $U$, we will use random sampling to obtain our spectral approximation $\widetilde{\cL}$ to the matrix $(I - U_{[k]} \Sigma_{[k]} U_{[k]}^T)$. Indeed, for the instance consisting of $k$-disjoint equal sized expanders, note that $I - U_{[k]} \Sigma_{[k]} U_{[k]}^T = U_{-[k]} \Sigma_{-[k]} U_{-[k]}^T$ where $U_{[-k]} \in \R^{n \times (n-k)}$ is the matrix whose columns are the last $(n-k)$ columns of $U$. Using the information that $\lambda_n \geq \lambda_{n-1} \geq \cdots \geq \lambda_{k+1} \geq \varphi^2/2$, one can compare quadratic forms on $I - U_{[k]} \Sigma_{[k]} U_{[k]}^T$ and $L_G$ (the normalized Laplacian of $G$) to show
$$ \wt{\cL}\approx_{O(1/\varphi^2), \xi} \cL_H. $$ 

We will now describe this in more detail. First, we will show that the matrix $I - U_{[k]} \Sigma_{[k]} U_{[k]}^T$ approximates the quadratic forms of $L_G$ multiplicatively. Then, we describe how this allows us to approximate the quadratic forms of $\mathcal{L}_H$. Finally, we will outline how to approximate the matrix $I - U_{[k]} \Sigma_{[k]} U_{[k]}^T$ . 

First, we introduce a central definition to this work, which is the notion of spectral embedding. 

\begin{restatable}[$k$-dimensional spectral embedding]{mydef}{defembedding}\label{def:spec-emb}
For every vertex $x$ we let $f_x=U_{[k]}^T  \mathds{1}_x $ be the $k$-dimensional spectral embedding of vertex $x$.
\end{restatable}

The spectral embeddings of vertices in a graph provide rich geometric information which has been
shown to be useful in graph clustering \cite{lee2014multiway, CzumajPS15, chiplunkar2018testing, GluchKLMS21}.
The following remark asserts that the inner products between $f_x$ and
$f_y$ are well-defined even though the choice for these vectors may not be basis free. First, we need the following 
standard result on eigenvalues of $(k, \varphi, \e)$-clusterable graphs \cite{lee2014multiway, chiplunkar2018testing}.

\begin{restatable}[\cite{GluchKLMS21}]{lemmma}{GLKMlemma}
\label{lem:bnd-lambda}
Let $G=(V,E)$ be a $d$-regular graph that admits a $(k,\varphi,\epsilon)$-clustering. 
Then we have $\lambda_k\leq 2\epsilon$ and $\lambda_{k+1}\geq \frac{\varphi^2}{2}$. 
\end{restatable}

\begin{restatable}{remark}{subspacewelldefine}\label{rm:dot-products-well-defined}
	Take a $(k,\varphi,\epsilon)$-clusterable graph $G$ where $\e/\varphi^2$ smaller than a constant. Thus, 
	the space spanned by the bottom $k$ eigenvectors of the normalized Laplacian 
	of $G$ is uniquely defined, i.e. the choice of $U_{[k]}$ is unique up to multiplication 
	by an orthonormal matrix $R\in \R^{k \times k}$ on the right.  Indeed, 
	by Lemma~\ref{lem:bnd-lambda} it holds that $\lambda_k \leq 2 \e$ and $\lambda_{k+1} \geq \varphi^2/2$. 
	Thus, since we assume that $\e/\varphi^2$ is smaller than an 
	absolute constant, we have $2 \e < \varphi^2/2$ and thus, the subspace spanned by the 
	bottom $k$ eigenvectors of the 
	Laplacian, i.e. the space of $U_{[k]}$, is uniquely defined, as required. We note 
	that while the choice of $f_x$ for $x \in V$ is not unique, but  the dot product 
	between the spectral embedding of $x\in V$ and $y\in V$ is well defined, since for every orthonormal 
	$R\in \R^{k\times k}$ one has  
	\[\langle Rf_x, Rf_y\rangle=(Rf_x)^T(Rf_y)=\left(f_x\right)^T (R^TR) \left(f_y\right)=\left(f_x\right)^T\left(f_y\right)\text{.}\]
\end{restatable}

Since $G$ is $(k, \varphi, \e)$-clusterable, by Remark \ref{rm:dot-products-well-defined}, the space spanned by the bottom $k$ eigenvectors of the $M$ is uniquely defined. 
Thus, for any $z \in \R^n$,  $z^T (U_{[k]} \Sigma_{[k]} U_{[k]}^T) z$ is well defined. 

Having observed this, we will now show that quadratic forms of $I - U_{[k]}\Sigma_{[k]} U_{[k]}^T$ approximate quadratic forms of $L_G$ multiplicatively.
\begin{restatable}{lemmma}{quadraticformlemma}
\label{lem:lw-Mkmass}Suppose that $G$ is $d$-regular, and let $L_G$ and $M$ denote the normalized Laplacian and lazy random walk matrix of $G$. Let $M=U\Sigma U^T$ denote the eigendecomposition of $M$. Then for any vector $z\in \R^n$ with $||z||_2=1$ we have
\[ \frac{1}{2}  \cdot  z^T L_G z \leq z^T\left(1-  U_{[k]}\Sigma_{[k]}U_{[k]}^T \right)z \leq  \frac{3}{\varphi^2} \cdot  z^T L_G z.\]
\end{restatable}
\begin{proof}
Recall that $U_{[k]}\in \R^{n\times k}$ is a matrix whose columns are the first $k$ 
columns of $U$, and $\Sigma_{[k]}\in \R^{k\times k}$ is a matrix whose columns are the first $k$ rows and columns of $\Sigma$. 
Let $U_{[-k]}\in \R^{n\times (n-k)}$ be matrix whose columns are the last $n-k$ 
columns of $U$, and $\Sigma_{[-k]}\in \R^{(n-k)\times (n-k)}$ be a matrix whose columns are the last $n-k$ rows and columns of $\Sigma$. Thus, the eigendecomposition of $M$ is $M=U\Sigma U^T= U_{[k]}\Sigma_{[k]} U_{[k]}^T + U_{[-k]}\Sigma_{[-k]} U_{[-k]}^T$.
Note that $M = I - \frac{L_G}{2}$,  thus we have
\begin{equation}
		\label{eq:mass-bottom}
		 z^T (U_{[k]}\Sigma_{[k]} U_{[k]}^T) z + z^T (U_{-[k]}\Sigma_{-[k]} U_{-[k]}^T) z = z^T M z= 1-\frac{z^T L_G z}{2}\text{,}
	\end{equation}
which by rearranging gives
	\begin{equation}
		\label{eq:contracted:lap}
		\frac{z^T L_G z}{2} \leq 1 - z^T (U_{[k]}\Sigma_{[k]} U_{[k]}^T) z = \frac{z^T L_G z}{2} + z^T (U_{-[k]}\Sigma_{-[k]} U_{-[k]}^T) z.
	\end{equation}

	The first inequality gives $1 - z^T (U_{[k]} \Sigma_{[k]} U_{[k]}^T) z \geq \frac{z^T L_G z}{2}$ as desired.
	
	To establish the second inequality above, we will show 
	$z^T (U_{-[k]} \Sigma_{-[k]} U_{-[k]}^T) z \leq \frac{2}{\varphi^2} z^T L_G z$.
	Let $z=\sum_{i=1}^n \alpha_i u_i$ be the eigendecomposition of vector $z$. Note that
	\[ z^T L_G z = \sum_{i=1}^n \lambda_i \alpha_i^2
	\geq \lambda_{k+1} \sum_{i=k+1}^n  \alpha_i^2.  \]
	By Lemma \ref{lem:bnd-lambda} we have $\lambda_{k+1}\geq \frac{\varphi^2}{2}$. This gives
	\begin{equation}
		\label{eq:alpha}
		\sum_{i=k+1}^n  \alpha_i^2 \leq \frac{z^T L_G z }{\lambda_{k+1}} \leq \frac{2}{\varphi^2}\cdot z^T L_G z.
	\end{equation}
Finally, putting \eqref{eq:contracted:lap} and \eqref{eq:alpha} together we get
	\begin{align*}
		1-z^T (U_{[k]}\Sigma_{[k]} U_{[k]}^T) z &= \frac{z^T L_G z}{2} + z^T (U_{-[k]}\Sigma_{-[k]} U_{-[k]}^T) z \leq z^T L_G z\left(\frac{1}{2}+\frac{2}{\varphi^2} \right) \leq \frac{3}{\varphi^2}\cdot  z^T L_G z \text{.}
	\end{align*}
\end{proof}

Now, we apply Lemma~\ref{lem:lw-Mkmass} to estimate the quadratic form of $\mathcal{L}_H$ on a vector $z\in \R^k$. To that effect, for $z \in \R^k$, we define $\zext \in \R^n$ as the natural extension of $z$ to $\R^n$: we let $\zext\in \R^n$ be the vector such that for every $x\in V$,  $\zext(x) = z_i$, where $C_i$ is the cluster that $x$ belongs to.

Note that $z^T \mathcal{L}_H z = \zext^T \mathcal{L}_G \zext = d\cdot \zext^T L_G \zext.$ 
Thus, to estimate $z^T \mathcal{L}_H z$ it suffices to design a good estimate for $\zext^T L_G \zext$, for which we use $\zext^T(I - U_{[k]}\Sigma_{[k]} U_{[k]}^T)\zext$, as per Lemma~\ref{lem:lw-Mkmass}. 

Finally, we briefly discuss how to estimate the quantity $\zext^T(I - U_{[k]}\Sigma_{[k]} U_{[k]}^T)\zext$. We have
$$
\zext^T(I - U_{[k]}\Sigma_{[k]} U_{[k]}^T)\zext=||\zext||_2^2- \zext^TU_{[k]}\Sigma_{[k]} U_{[k]}^T\zext= \sum_{i\in [k]} |C_i| z_i^2- \zext^TU_{[k]}\Sigma_{[k]} U_{[k]}^T\zext.
$$
Since the first term on the RHS can be easily approximated in the random sample model, we concentrate on obtaining a good estimate for the second term. We have
\begin{equation}
	\zext^T U_{[k]} \Sigma_{[k]} U_{[k]}^T \zext = \sum_{i,j \in [k]} z_i z_j \sum_{\substack{x \in C_i\\y \in C_j}} \rdp{f_x, \Sigma_{[k]} f_y},
\end{equation}
and therefore in order to estimate $\zext^T  L_G \zext$, it suffices to use a few random samples to estimate the sum above, as long as one is able to compute high accuracy estimates for  $\rdp{f_x, \Sigma_{[k]} f_y}, x, y\in V,$ with high probability. We refer to such a primitive as a {\em weighted dot product oracle}, since it computes a weighted dot product between the $k$-dimensional spectral embeddings $f_x$ 
 and $f_y$ for $x,y\in V$. Assuming such an estimator, which we denote by $\textsc{WeightedDotProductOracle}$, our algorithm $\textsc{ApproxContractedGraph}$ (Algorithm~\ref{alg:quadratic} below) obtains an approximation $\mathcal{L}'$ to the Laplacian of the contracted graph.  \\
\vspace{-10pt}
\begin{algorithm}[H]
\caption{$\textsc{ApproxContractedGraph}(G, \xi,\mathcal{D})$ 
{\qquad \qquad \qquad  time  $m^{1/2+O(\epsilon)}\cdot \poly(1/\xi)$ }
}
\label{alg:quadratic}
\begin{algorithmic}[1]

    \STATE $s \gets O^*\left(m^{O(\e/\varphi^2)} \cdot (1/\xi)^{O(1)}\right)$ \label{ln:s-num:main} \textit{\qquad \qquad \qquad \qquad \qquad \# See Theorem~\ref{thm:contracted-graph} for the exact value}
    
    \STATE $S\gets$ (multi)set of $s$ i.i.d random vertices together with their cluster label \label{ln:randset}
    \STATE $S_i \gets S\cap C_i$, $\widetilde{w}(i) \gets \frac{|S_i|}{s} \cdot n$, for all $i\in [k]$
    \FOR{$i,j\in [k]$}
    \STATE Assign $K_{i,j} = \frac{\widetilde w(i)}{|S_i|} \cdot \frac{\widetilde w(j)}{|S_j|} \cdot \sum_{\substack{x \in S_i\\y \in S_j}} \langle f_x, \Sigma_{[k]} f_y\rangle_{\text{apx}}$ \label{ln:matk}
    \ENDFOR
    \STATE Assign $\cL' = d \cdot (I - K)$.\label{ln:matl}
    \STATE Use SDP to round $\cL'$ to a Laplacian $\widetilde \cL$ s.t $\frac{\varphi^2}{3} \cL' - \frac{\xi}{2}\cdot d n\cdot I_k \preceq \widetilde{\cL} \preceq 2 \cL' + \frac{\xi}{2} \cdot d n\cdot I_k.$ \label{ln:defmat}
    \STATE  $\widetilde{H} \gets \left([k], {[k] \choose 2}, \widetilde{W}, \widetilde{w} \right)$ \textit{\# $\widetilde{H}$ is the weighted graph with Laplacian $\widetilde{\cL}$ and vertex weights $\widetilde{w}(i)$} 
\STATE \textit{\qquad\qquad \qquad  \qquad  \qquad\qquad \qquad\qquad \qquad\qquad \quad\#  Note that $\wt{W}_e=-\wt{\cL}_e$ for every $e\in {[k] \choose 2}$}
\STATE \return $\widetilde{H}$
\end{algorithmic}
\end{algorithm}

\vspace{-5pt}
\noindent {\bf Estimating weighted dot products:} Our construction of \textsc{WeightedDotProductOracle} (Algorithm~\ref{alg:weighted-dot}) for estimating $\rdp{f_x, \Sigma_{[k]} f_y}$ proceeds along the lines of
~\cite{GluchKLMS21}. 
We run short random-walks of length $t \approx \log n/\varphi^2$ to obtain dot product access to the spectral embedding of vertices. 
Given $x \in V$, let $m_x$ denote the probability distribution of endpoints of a $t$-step random-walks started from $x$.

We first show that 
one can estimate $\rdp{m_x, m_y}$ 
in time $\approx m^{1/2 + O(\e/\varphi^2)} \cdot \text{poly}(1/\xi)$ with probability $1-n^{-100\cdot k}$. Then, we construct a Gram matrix $\mathcal{G}\in \R^{s\times s}$ such that $\mathcal{G}_{x,y} = \rdp{m_x, m_y}$ for every $x,y \in S$, where $S$ is a small set of sampled vertices with $|S|=s=m^{O(\epsilon)}$. 
Next, we apply an appropriate linear transformation to the Gram matrix $\mathcal{G}$ and use it to estimate $\rdp{f_x, \Sigma_{[k]} f_y}$ up to very tiny additive error $\frac{\xi}{n \cdot \poly(k)}$ (see Section \ref{sec:wdp:oracle}). \\

\vspace{-5pt}
\noindent {\bf Using Semidefinite Programming to round $\mathcal{L}'$:} As mentioned above, our proxy for the Laplacian $\cL_H$ is obtained via an approximation to $I - U_{[k]} \Sigma_{[k]} U_{[k]}^T$. However, this approximator might not even be a Laplacian. To allay this, we first show that using calls to weighted dot product oracle, we can approximate all the entries of $I - U_{[k]} \Sigma_{[k]} U_{[k]}^T$ to within a very good precision. Starting off from such an approximation, one can use semidefinite programming methods to round the intermediate approximator to a bonafide Laplacian $\widetilde{\cL}$. In some more detail, we show the following.

\begin{restatable}[Informal version of Theorem~\ref{thm:contracted-graph}]{theorem}{infthmcg}\label{thm:infthmcg}
	The algorithm $\textsc{ApproxContractedGraph}$ (Algorithm \ref{alg:quadratic}) when given a $(k, \Omega(1), \e)$-clusterable graph
	as input, uses a data structure $\cD$ obtained from  $\approx m^{1/2 + O(\e)}$ time preprocessing 
	routine, runs in time $\approx m^{1/2 + O(\e)}$, and finds a graph $\wt{H}$ with Laplacian $\wt{\cL}$
	such that with probability $1-n^{-100}$:
\[ \wt{\cL} \approx_{O(1/\varphi^2), \xi}  \cL_H. \]
\end{restatable} 

\noindent  {\bf Approximating the Dasgupta Cost of the contracted graph $\widetilde{H}$:} Consider the graph $\wt{H} = \left([k], {[k] \choose 2}, \widetilde{W}, \widetilde{w}\right)$ returned by the $\textsc{ApproxContractedGraph}$ procedure (Algorithm \ref{alg:quadratic}). 

Once Theorem~\ref{thm:infthmcg} has been established, our estimation primitive $\textsc{EstimatedCost}$ (Algorithm \ref{alg:tree:mn}) uses a simple vertex-weighted extension of a result of \cite{charikar2017approximate}
to find a 
tree $\wt{T}$ on $\wt{H}$.

Specifically, we need the following definitions.

\begin{restatable}[\textbf{Vertex-weighted sparsest cut problem}]{mydef}{wsparsestcut}
Let $H=(V,E,W,w)$ be a vertex and edge weighted graph. For every set $S\in V$, we define the sparsity of cut $(S, V\setminus S)$ on graph $H$ as \[\spp_H(S)=\frac{W(S, V\setminus S)}{w(S)\cdot w(V\setminus S)}\text{,}\] where $w(S)=\sum_{x\in S}w(x)$. The vertex-weighted sparsest cut of graph $G$ is the cut with the minimum sparsity, i.e., $\arg\min_{S\subseteq V} \spp_H(S)$.
\end{restatable}

\begin{restatable}[\textbf{Vertex-weighted recursive sparsest cut algorithm (\textsc{WRSC})}]{mydef}{wrsc}
\label{def:wrc}
Let $\alpha>1$ and $H=(V,E,W,w)$ be a vertex and edge weighted graph. Let $(S, V\setminus S)$ be the vertex-weighted sparsest cut of $H$.  The vertex-weighted recursive sparsest cut algorithm on graph $H$ is a recursive algorithm that first finds a cut $(T,V\setminus T)$ such that $\spp_H(T) \leq \alpha\cdot \spp_H(S)$, and then recurs on the subgraph $H[T]$ and subgraph $H[V\setminus T]$.
\end{restatable}

Next, we first state results which help bound the Dasgupta Cost incurred by the tree one gets by using the \emph{vanilla} recursive sparsest cut algorithm on any graph. Then, in Corollary~\ref{cor:vertex-extention}, we present corresponding bounds for vertex-weighted graphs.

\begin{restatable}[Theorem 2.3 from \cite{charikar2017approximate}]{theorem}{charikarthm} \label{thm:charikar}
	Let $G=(V,E)$ be a graph. Suppose the RSC algorithm uses an $\alpha$ approximation
	algorithm for uniform sparsest cut. Then the algorithm RSC achieves an $O(\alpha)$
	approximation for the Dasgupta cost of $G$.
\end{restatable}

The following corollary from \cite{charikar2017approximate},  follows using the $O(\sqrt{\log |V|})$ approximation algorithm for the uniform sparsest cut.

\begin{restatable}[\cite{charikar2017approximate}]{mycor}{charikarcor}\label{cor:charikar}
	Let $G=(V,E)$ be a graph. Then algorithm RSC achieves an $O(\sqrt{\log |V|})$ approximation
	for the Dasgupta cost of $G$.
\end{restatable}
Since the clusters of $G$ have different sizes, and since the Dasgupta Cost of a graph is a function of the size of the lowest common ancestor of the endpoints of the edges, 
we use weighted Dasgupta cost to relate the cost of $G$ and the contracted graph $H$. 

\begin{restatable}[\textbf{Weighted Dasgupta cost}]{mydef}{defwtdcost}
\label{def:weightedcost}
	Let $G=(V,E,W, w)$ denote a vertex and edge weighted graph. For a tree $T$ with $|V|$ leaves (corresponding to vertices
	of $G$), we define the weighted Dasgupta cost of $T$ on $G$ as 
	$$\wCT_G(T) = \sum_{(x,y) \in E} W(x,y) \cdot \sum_{z\in \lv(T[\lca(x,y)])} w(z).$$
	\label{def:wtd:dcost}
\end{restatable}
We get the following guarantee on the Weighted Dasgupta Cost obtained by the  \textsc{WRSC} algorithm. 
\begin{restatable}{mycor}{wcharikarcor}
\label{cor:vertex-extention}
Let $H=(V,E,W,w)$  be a vertex and edge weighted graph. Then algorithm \textsc{WRSC} achieves an $O(\sqrt{\log |V|})$ approximation for the weighted Dasgupta cost of $H$.
\end{restatable}

Letting  $\widetilde{T}=\textsc{WRSC}(\widetilde{H})$  be the tree computed by Algorithm \ref{alg:tree:mn}, using Corollary \ref{cor:vertex-extention}, we show that
the estimate 
$$\estg \eqdef \est$$ computed by Algorithm \ref{alg:tree:mn} satisfies 
\[\CT(G) \leq \estg\leq O\left(\frac{\sqrt{ \log k}}{\varphi^{O(1)}} \right)\CT(G). \]
The details are presented in Section~\ref{sec:sketch:tree}.\\

The rest of the paper is structured as follows:
In Section~\ref{sec:estimate-cost}, we prove Theorem~\ref{thm:estdcost}. In Section~\ref{sec:exp-cost} we prove the guarantees of the \textsc{TotalClustersCost} procedure. Finally, in Section \ref{sec:wdp:oracle}, we prove the correctness of the \textsc{WeightedDotProductOracle}.

\newpage

{
\small
\bibliographystyle{alpha}
\bibliography{references}
}
\newpage

\tableofcontents
\appendix

\section{Proof of main result (Theorem~\ref{thm:estdcost})} \label{sec:estimate-cost}
Throughout this section, assume that $k \geq 2$, $\varphi\in (0,1)$ and that $\frac{\epsilon}{\varphi^2}$ is smaller than a positive sufficiently small
constant. We will always assume that $G=(V,E)$ is a graph that admits a $(k,\varphi,\epsilon)$-clustering $\mathcal{C} = C_1, \ldots , C_k$ with $\eta \coloneqq \frac{\max_{i \in [k]} |C_i|}{ \min_{i \in [k]} |C_i|} = O(1).$

Given a set $S \subseteq V$, we let $G[S]$ denote the induced subgraph on $S$, and we use the notation $G\{S\}$ to denote the graph $G[S]$ with self loops added so that every vertex in $G\{S\}$ has the same degree as in $G$. 

Recall that we use $O^*$-notation to suppress $\poly (k,  1/\varphi,  1/\e)$ and $\polylog n$ -factors. 
For $i \in \mathbb{N}$ we use $[i]$ to denote the set $\{1,2, \dots, i \}$. For a vertex $x\in V$, we say that $\mathds{1}_x\in\mathbb{R}^{n}$ is the indicator of $x$, that is, the vector which is $1$ at index $x$ and $0$ elsewhere. For a (multi)set $S\subseteq V$, we say that $\mathds{1}_S\in\mathbb{R}^{n}$ is the indicator of set $S$, i.e., $\mathds{1}_S=\sum_{x\in S}\mathds{1}_x$. For a multiset $I_S=\{x_1,\ldots,x_s\}$ of vertices from $V$ we abuse notation and also denote by $S$ the $n \times s$ matrix whose $i^{\text{\tiny{th}}}$ column is $\mathds{1}_{x_i}$. 

Our algorithm and analysis use spectral techniques, and therefore, we setup the following notation. For a symmetric matrix $A$, we write $\nu_i(A)$ (resp. $\nu_{\max}(A), \nu_{\min}(A))$ to denote the $i^{\text{th}}$ largest (resp. maximum, minimum) eigenvalue of $A$.

We also denote with $A_G$ the adjacency matrix of $G$ and let $\cL_G = d \cdot I - A_G$ denote
the Laplacian of $G$. Denote with $L_G$ the \textit{normalized Laplacian} 
of $G$ where $L_G=I-\frac{A_G}{d}$. We denote the eigenvalues of $L_G$ by  
$0 \leq \lambda_1 \leq \ldots \leq \lambda_n \leq 2$ and we write $\Lambda$ to refer to the 
diagonal matrix of these eigenvalues in non-decreasing order. We also denote by 
$(u_1,\ldots,u_n)$  an orthonormal basis of eigenvectors of $L_G$ and with $U \in \R^{n\times n}$  
the matrix whose columns are the orthonormal eigenvectors of $L_G$ arranged in non-decreasing order 
of eigenvalues. Therefore the eigendecomposition of $L_G$ is $L_G=U \Lambda U^T$. For any 
$1\leq k \leq n$ we write $U_{[k]}\in \R^{n\times k}$ for the matrix whose columns are the first $k$ 
columns of $U$. Now, we introduce a central definition to this work, which is the notion of a spectral embedding. 

\defembedding*

The spectral embeddings of vertices in a graph provide rich geometric information which has been
shown to be useful in graph clustering \cite{lee2014multiway, CzumajPS15, chiplunkar2018testing, GluchKLMS21}.

In this paper, we are interested in the class of graphs that admit a $(k, \varphi, \e)$-clustering.
We would often need to refer to weighted graphs of the form $H = ([k], \binom{[k]}{2}, W, w)$
where $W$ is a weight function on the edges of $H$ and $w$ is a weight function on the vertices
of $H$. We denote the Laplacian of $H$ as $\cL_H = D_H - W$ where $D_H \in \R^{k \times k}$ is a diagonal matrix 
where for every $i \in [k]$, $D_H(i,i) = \sum_{j \in [k]} W(i,j)$.
Our algorithms often require
to estimate the inner product between spectral embeddings of vertices $x$ and $y$ denoted  $f_x$ and $f_y$. For pairs of vertices $x, y\in V$ we use the notation 
$$
\langle f_x, f_y\rangle:=({f_x})^T(f_y)
$$
to denote the dot product in the embedded domain. The following remark asserts that the inner products between $f_x$ and
$f_y$ are well-defined even though the choice for these vectors may not be basis free.

First, we need the following 
standard result on eigenvalues of $(k, \varphi, \e)$-clusterable graphs \cite{lee2014multiway, chiplunkar2018testing}.

\GLKMlemma*

Now, we are ready to state the remark from before.
\subspacewelldefine*

We denote the transition matrix of the \textit{random walk associated with} $G$ by $M=\frac{1}{2}\cdot \left(I+\frac{A_G}{d}\right)$. From any vertex $v$, this random walk takes every edge incident on $v$ with probability $\frac{1}{2d}$, and stays on $v$ with the remaining probability which is at least $\frac{1}{2}$. Note that 
$M=I-\frac{L_G}{2}$. Observe that for all $i$, $u_i$ is also an eigenvector of $M$, with eigenvalue $1-\frac{\lambda_i}{2}$.
We denote with $\Sigma$ the diagonal matrix of the eigenvalues of $M$ in descending order.  Therefore the eigendecomposition of $M$ is $M=U \Sigma U^T$. We  write $\Sigma_{[k]}\in \R^{k\times k}$ for the matrix whose columns are the first $k$ rows and columns of $\Sigma$. Furthermore, for any $t$, $M^t$ is a transition matrix of random walks of length $t$. For any vertex $x$, we denote the probability distribution of a $t$-step random walk started from $x$ by $m_x=M^t \mathds{1}_x$. For a multiset $I_S=\{x_1,\ldots,x_s\}$ of vertices from $V$, let matrix $M^t S \in \R^{n\times s}$ is a matrix whose column are probability distribution of $t$-step random walks started from vertices in $I_S$. Therefore, the $i$-th column of $M^t S$ is $m_{x_i}$.

We recall standard (partial) ordering on $n$-by-$n$ symmetric matrices.
\begin{mydef} \label{def:lowener:order}
	Given two symmetric matrices $A, B \in \R^{n \times n}$, we say that matrix $A$
	precedes $B$ in the Loewner (or the PSD) order if $B - A$ is a positive semidefinite
	matrix. This is denoted as $A \preceq B$ or as $B \succeq A$. 
\end{mydef}

In this section, we present an algorithm to estimate the Dasgupta cost of  $(k, \varphi, \e)$-clusterable graphs in sublinear time assuming oracle access to the underlying clustering, i.e., in the \linebreak $\rsm$ , which we formally define below: 
\begin{mydef}[\textbf{The $\rsm$}]
\label{def:ss:model} 
	Let $G = (V,E)$ be 
	graph that admits a $(k,\varphi,\epsilon)$-clustering $C_1,\ldots, C_k$. In the $\rsm$, we assume that there exists an oracle that gives us a (multi)set $S$ sampled independently with probability proportional to their degrees, together with their cluster label, i.e., for every vertex $x\in S$, a label $i \in [k]$ such that $C_i \ni x$.  
\end{mydef}
The following definition is central to our algorithm and analysis:

\defcontract*

In fact, our algorithm for estimating the Dasgupta cost of $G$ will first construct a low cost tree for the contracted graph $H$ and use its cost as a proxy for the Dasgupta cost of $G$.  The algorithm is reproduced in Section~\ref{sec:sketch:tree} below (see Algorithm~\ref{alg:tree}), and satisfies the following guaratees:

\thmestdcost*

We start by designing a sublinear time estimator for the quadratic form of $L_G$, where $L_G$ is the normalized Laplacian of the graph $G$. Later, we will use this estimator to obtain an approximation of the contracted graph. We can without loss of generality, assume that the graph is $d$-regular, by adding self-loops to the graph to make all the degrees $d$. 
Therefore, we prove the correctness of procedure \textsc{ApproxContractedGraph} (Algorithm \ref{alg:quadratic}) for {\em regular} graphs. Throughout this overview, in Section \ref{sec:est:quad:forms} and Section \ref{sec:quadproof} we assume that graph $G$ is $d$-regular. 

Recall that  $M$  is the random walk matrix of $G$ and let $M = U \Sigma U^T$ be the eigendecomposition of $M$. Let $U_{[k]}\Sigma_{[k]} U_{[k]}^T $ be a rank-$k$ approximation to $M$ obtained by truncating the terms corresponding to small eigenvalues. Since, $G$ is $(k, \varphi, \e)$-clusterable, the space spanned by the bottom $k$ eigenvectors of the $M$ is uniquely defined (see Remark \ref{rm:dot-products-well-defined}).
Thus, for any $z \in \R^n$,  $z^T (U_{[k]} \Sigma_{[k]} U_{[k]}^T) z$ is well defined. 
Lemma \ref{lem:lw-Mkmass} reproduced below shows that the quadratic form of $I - U_{[k]}\Sigma_{[k]} U_{[k]}^T$ approximates the quadratic form of $L_G$ multiplicatively.
\quadraticformlemma*

Let $G$ be a $d$-regular graph that admits a $(k,\varphi,\epsilon)$-clustering $\cC=C_1, \ldots , C_k$. Let $H = G/\cC$ be the contraction of $G$ with respect to the partition $\cC $ (Definition \ref{def:contracted}), let $\mathcal{L}_H$ be the Laplacian of graph $H$. We apply Lemma~\ref{lem:lw-Mkmass} to estimate the quadratic form of $\mathcal{L}_H$ on a vector $z\in \R^k$. To that effect, for $z \in \R^k$, we define $\zext \in \R^n$ as the natural extension of $z$ to $\R^n$: we let $\zext\in \R^n$ be the vector such that for every $x\in V$,  $\zext(x) = z_i$, where $C_i$ is the cluster that $x$ belongs to.  

\begin{mydef}[\textbf{Extension $\zext$ of a vector $z\in \R^k$}]\label{def:ext}
For a vector $z\in \R^k$ let $\zext\in \R^n$ be the vector such that for every $x\in V$,  $\zext(x) = z_i$, where $C_i$ is the cluster that $x$ belongs to.
\end{mydef}

Note that $z^T \mathcal{L}_H z = \zext^T \mathcal{L}_G \zext = d\cdot \zext^T L_G \zext.$ 
Thus, to estimate $z^T \mathcal{L}_H z$ it suffices to design a good estimate for $\zext^T L_G \zext$, for which we use $\zext^T(I - U_{[k]}\Sigma_{[k]} U_{[k]}^T)\zext$, as per Lemma~\ref{lem:lw-Mkmass}:
$$
\zext^T(I - U_{[k]}\Sigma_{[k]} U_{[k]}^T)\zext=||\zext||_2^2- \zext^TU_{[k]}\Sigma_{[k]} U_{[k]}^T\zext= \sum_{i\in [k]} |C_i| z_i^2- \zext^TU_{[k]}\Sigma_{[k]} U_{[k]}^T\zext.
$$
Since the first term on the RHS can be easily approximated in the random sample model (Definition \ref{def:ss:model}), we concentrate on obtaining a good estimate for the second term. We have
\begin{equation}
	\zext^T U_{[k]} \Sigma_{[k]} U_{[k]}^T \zext = \sum_{i,j \in [k]} z_i z_j \sum_{\substack{x \in C_i\\y \in C_j}} \rdp{f_x, \Sigma_{[k]} f_y},
\end{equation}
and therefore in order to estimate $\zext^T  L_G \zext$, it suffices to use few random samples to estimate the sum above, as long as one is able to compute high accuracy estimates for  $\rdp{f_x, \Sigma_{[k]} f_y}, x, y\in V,$ with high probability. We refer to such a primitive as a {\em weighted dot product oracle}, since it computes a weighted dot product between the $k$-dimensional spectral embeddings $f_x$ 
 and $f_y$ for $x,y\in V$ (as per Algorithm~\ref{alg:LearnEmbedding}). Assuming such an estimator, which we denote by $\textsc{WeightedDotProductOracle}$,  our algorithm (Algorithm \textsc{ApproxContractedGraph}) obtains an approximation $\wt{L} = \wt{L}(\wt{H})$ to the Laplacian of the contracted graph and is reproduced below.  
 
\begin{algorithm}[H]
\caption*{\textbf{Algorithm 4} \textsc{ApproxContractedGraph}($G, \xi,\mathcal{D}$){\qquad \qquad \qquad  \qquad  time  $n^{1/2+O(\epsilon)}\cdot \poly(1/\xi)$}}
\begin{algorithmic}[1]

   \STATE $s \gets O^*\left(n^{O(\e/\varphi^2)} \cdot (1/\xi)^{O(1)}\right)$ \label{ln:s-num:main} \textit{\qquad \qquad \qquad \qquad \qquad \# See Theorem~\ref{thm:contracted-graph} for the exact value}
    
    \STATE $S\gets$ (multi)set of $s$ i.i.d random vertices together with their cluster label \label{ln:randset}
    \STATE $S_i \gets S\cap C_i$, $\widetilde{w}(i) \gets \frac{|S_i|}{s} \cdot n$, for all $i\in [k]$
    \FOR{$i,j\in [k]$}
    \STATE Assign $K_{i,j} = \frac{\widetilde w(i)}{|S_i|} \cdot \frac{\widetilde w(j)}{|S_j|} \cdot \sum_{\substack{x \in S_i\\y \in S_j}} \langle f_x, \Sigma_{[k]} f_y\rangle_{\text{apx}}$ 
    \ENDFOR
    \STATE Assign $\cL' = d \cdot (I - K)$.
    \STATE Use SDP to round $\cL'$ to a Laplacian $\widetilde \cL$ s.t $\frac{\varphi^2}{3} \cL' - \frac{\xi}{2}\cdot d n\cdot I_k \preceq \widetilde{\cL} \preceq 2 \cL' + \frac{\xi}{2} \cdot d n\cdot I_k.$ 
    \STATE  $\widetilde{H} \gets \left([k], {[k] \choose 2}, \widetilde{W}, \widetilde{w} \right)$ \textit{\# $\widetilde{H}$ is the weighted graph with Laplacian $\widetilde{\cL}$ and vertex weights $\widetilde{w}(i)$} 
\STATE \textit{\qquad\qquad \qquad  \qquad  \qquad\qquad \qquad\qquad \qquad\qquad \quad\#  Note that $\wt{W}_e=-\wt{\cL}_e$ for every $e\in {[k] \choose 2}$}
\STATE \return $\widetilde{H}$
    
\end{algorithmic}
\end{algorithm}

\begin{remark} \label{rm:other-estimator}
	A yet another natural choice for obtaining the estimator $\wt{L}$ uses the rank $k$ SVD of 
	$L_G$ which is $U_{[k]} \Lambda_{[k]} U_{[k]}^T$. However, this does not provide a multiplicative 
	approximation to the quadratic form of the Laplacian. To see this, consider the following instance. 
	Let $G$ consist of $k$ disjoint expanders each with inner conductance $\varphi$. In this case,
	$\lambda_1 = \lambda_2 = \cdots = \lambda_k = 0$, and the rank $k$ approximation above results in the estimator 
	being $0$.
\end{remark}

The following theorem asserts that the procedure $\textsc{ApproxContractedGraph}$ finds a reasonably good approximation to $\cL_H$.

\begin{restatable}{theorem}{thmcontractedgraph}\label{thm:contracted-graph}
Let $\frac{1}{n^3}<\xi<1$.  Let $H=G/\cC$  be the contraction of $G$ with respect to $\cC$ (Definition \ref{def:contracted}), and let $\mathcal{L}_H$ be the Laplacian of $H$. 

With probability at least $1 - n^{-100}$ over the initialization procedure (Algorithm \ref{alg:LearnEmbedding-tt}) the following property is satisfied: 
	$\textsc{ApproxContractedGraph}(G, \xi, \cD)$ (Algorithm \ref{alg:quadratic}), runs in time 
	$\runningtimeiv$, uses $O^*\left(n^{O(\e/\varphi^2)}\cdot \left(\frac{1}{\xi}\right)^{O(1)}\right)$ seed queries, 
	and finds a graph $\wt{H}$ with Laplacian $\wt{\mathcal{L}}$ such that with probability
	at least $1 - n^{-100}$, we have
$$ \wt{\cL}\approx_{O(1/\varphi^2), \xi} \cL_H. $$ 
\end{restatable}

Once Theorem~\ref{thm:contracted-graph} has been established, our algorithm for approximating Dasgupta cost of $G$ simply uses the algorithm of~\cite{charikar2017approximate} on the approximation $\wt{H}$ to the contracted graph $H$. Since the cluster sizes in $G$ are different and the Dasgupta cost of a graph is a function of the size of the lowest common ancestor of endpoints of edges, to relate the cost of $G$ and the contracted graph $H$, we recall the definition of weighted Dasgupta cost.

\defwtdcost*

Let $G=(V,E)$ be a $d$-regular graph that admits a $(k,\varphi,\epsilon)$-clustering $\cC = C_1, \ldots , C_k$, and  let $H = G/\cC$ be the contraction of $G$ with respect to the partition $\cC $. We denote
the optimal Dasgupta tree for $H$ as $T^*_H$. With this setup,
we will now show the main result (Theorem \ref{thm:estdcost})
which is finally presented in Section \ref{sec:sketch:tree}. 
We consider the graph $\wt{H} = \left([k], {[k] \choose 2}, \widetilde{W}, \widetilde{w}\right)$ returned by the $\textsc{ApproxContractedGraph}$ procedure. 
Our estimation primitive $\textsc{EstimatedCost}$ (Algorithm \ref{alg:tree}) uses a simple vertex-weighted extension of a result of \cite{charikar2017approximate}
to find a 
tree $\wt{T}$ on $\wt{H}$. Specifically, we need the following definitions.

\wsparsestcut*
\wrsc*

Next, we first state results which help bound the Dasgupta Cost incurred by the tree one gets by using \emph{vanilla} recursive sparsest cut algorithm on any graph. In Corollary~\ref{cor:vertex-extention}, we present algorithms which present bounds on the Dasgupta Cost one obtains for vertex weighted graphs. 

\charikarthm*

The following corollary from \cite{charikar2017approximate},  follows using the $O(\sqrt{\log |V|})$ approximation algorithm for the uniform sparsest cut.
\charikarcor*
\wcharikarcor*

Let $\widetilde{T}=\textsc{WRSC}(\widetilde{H})$. We denote the cost of $\wt{T}$ as $\wCT_{\wt{H}}(\wt{T})$, we present an estimator ($\estg$) 
for the Dasgupta cost of an optimal tree in $G$. In particular, we show $\estg \leq O\left( \frac{\sqrt{\log k}}{\varphi^{O(1)}} \right) \cdot \CT_G(T^*_G)$, where, 
$$\estg \eqdef \est \text{.}$$

The proof proceeds in two steps: in the first step, we prove Lemma \ref{lem:est:tstar}
which upper bounds $\estg$ in terms of $\wCT_H(T^*_H)$, 
where $T^*_{H}$ is the optimum Dasgupta tree for $H$. 
Next, we use Lemma \ref{lem:dcostg:tstar} to relate $\wCT_H(T^*_H)$ with $\CT_G(T^*_G)$. We finally restate the $\textsc{EstimatedCost}$ procedure.
\begin{algorithm}[H]
\caption{$\textsc{EstimatedCost}(G)$ {\qquad  \qquad\qquad\qquad\qquad \qquad \qquad \qquad  time  $\approx n^{1/2+O(\epsilon)}$}}
\label{alg:tree}
\begin{algorithmic}[1]
\STATE $\xi \gets O\left(\frac{\varphi^2}{d_{\max}\cdot k^3\cdot \sqrt{\log k}}\right)$ \label{setxi}
\STATE $\mathcal{D} \gets \textsc{InitializeWeightedDotProductOracle($G,\xi$)}$ \textit{\qquad\qquad\ \qquad \quad  \# See Algorithm \ref{alg:LearnEmbedding-tt}}
\STATE $\widetilde{H} \gets\textsc{ApproxContractedGraph}(G, \xi,\mathcal{D})$ \textit{\quad \ \#The Laplacian $\widetilde{\mathcal{L}}$ of $\widetilde{H}$ satisfies Equation \eqref{eq:l-lh-spec}}
\STATE $\widetilde{T} \gets \textsc{WeightedRecursiveSparsestCut}(\widetilde H)$ \textit{\qquad  \#Weighted version of Algorithm of~\cite{charikar2017approximate}}
\STATE $\estg \gets \est$ \label{ln:est-set}
\STATE \return \estg
\end{algorithmic}
\end{algorithm}

The details are presented in Section~\ref{sec:sketch:tree}. The rest of this section is organized as follows. We first present the analysis of \textsc{ApproxContractedGraph} in Section~\ref{sec:est:quad:forms}. In turn, this analysis requires guarantees on how well the matrix $\cL'$ obtained in Line 7 of Algorithm~\ref{alg:quadratic} approximates $\cL_H$. This guarantee is presented in Section~\ref{sec:quadproof}. 
 Finally, the proof of Theorem~\ref{thm:estdcost} is given in Section~\ref{sec:sketch:tree}. In Section ~\ref{sec:hoeffding:proof} we prove a few intermediate Lemmas used in the proof.

\subsection{Correctness of \textsc{ApproxContractedGraph} (Proof of Theorem \ref{thm:contracted-graph})}  \label{sec:est:quad:forms}
Throughout this section, we will assume that $G=(V,E)$ is a $d$-regular graph. The main result of this section is the correctness of \textsc{ApproxContractedGraph} in the random sample (Theorem \ref{thm:contracted-graph}).

We first collect the ingredients we need before proceeding further with the proof. As a first step, since \textsc{ApproxContractedGraph} relies on our \textsc{WeightedDotProductOracle}, we will need correctness guarantees for the latter in our analysis. These are presented in Theorem~\ref{thm:wdp} below, and proved in Appendix \ref{sec:wdp:oracle}.

\begin{restatable}{theorem}{thmweighteddot}
\label{thm:wdp}
Let $M$ denote the random walk matrix of $G$, and  let $M=U\Sigma U^T $ denote the eigendecomposition of  $M$. 
With probability at least $1 - n^{-100}$ over the initialization procedure (Algorithm \ref{alg:LearnEmbedding}) the following holds:

With probability at least $1-3\cdot n^{-100\cdot k}$ for all $x,y\in V$ we have
\[|\adp{f_{x},\Sigma_{[k]} f_y} -\rdp{f_{x},\Sigma_{[k]} f_y}  | \leq \frac{\xi}{\mycolor{nk^2}
} \text{,}\]
where $\adp{f_{x},\Sigma_{[k]} f_y} = \textsc{WeightedDotProductOracle}(G,x,y,  \xi, \mathcal{D})$. 

Moreover, the running time of the procedures
$\textsc{InitializeWeightedDotProductOracle}$ (Algorithm \ref{alg:LearnEmbedding}) and $\textsc{WeightedDotProductOracle}$ 
(Algorithm \ref{alg:weighted-dot}) is $\runningtimeiv$.
\end{restatable}

The other ingredient in our proof comprises of the following two lemmas we state below. 

\begin{lemmma} \label{lem:kk'-close}
    Let $\frac{1}{n^4}<\xi<1$. Let $H=G/\cC$  be the contraction of $G$ with respect to $\cC$ (Definition \ref{def:contracted}), and let $\mathcal{L}_H$ be the Laplacian of $H$. Let $K' \in \R^{k \times k}$ denote the matrix where $K'_{i,j} = \mathds{1}_{C_i} U_{[k]}^T \Sigma_{[k]} U_{[k]} \mathds{1}_{C_j}$. Then with 
	probability at least $1 - n^{-50\cdot k}$ we have
 $$K' - \xi n/k \cdot I_k \preceq K \preceq K' + \xi n/k \cdot I_k.$$
\end{lemmma}
The proof of Lemma~\ref{lem:kk'-close} is deferred to Section~\ref{sec:quadproof}. 

\begin{lemmma} \label{lem:k'-to-lh}
    Let $\frac{1}{n^4}<\xi<1$. Let $H=G/\cC$  be the contraction of $G$ with respect to $\cC$ (Definition \ref{def:contracted}), and let $\mathcal{L}_H$ be the Laplacian of $H$. Let $K' \in \R^{k \times k}$ denote the matrix where $K'_{i,j} = \mathds{1}_{C_i} U_{[k]}^T \Sigma_{[k]} U_{[k]} \mathds{1}_{C_j}$. Then 
    $$\cL_H/2 \preceq d(I - K') \preceq 3/\varphi^2 \cdot \cL_H.$$
\end{lemmma}
\begin{proof} Follows immediately from Lemma~\ref{lem:lw-Mkmass}.
\end{proof}

We now prove Theorem \ref{thm:contracted-graph}.

\begin{proof} (Of Theorem \ref{thm:contracted-graph}) Let $\mathcal{L}'$ denote the matrix obtained in Line 7 of Algorithm~\ref{alg:quadratic}. 
We use the following SDP to round $\mathcal{L}'$ to a Laplacian $\widetilde{\mathcal{L}}$ in Line 8 of Algorithm~\ref{alg:quadratic}:
\begin{condenum}[label=(C\arabic*)., ref=(C\arabic*)]
	\item Minimize $0$ (Feasibility program, no objective)
	\item $\wt{\cL}$ symmetric and $\wt{\cL} \succeq 0$
    \item For every $i \in [k]$, $\sum_j \widetilde{\cL}_{i,j} = 0$, and $\widetilde{\cL}_{i,i} \geq 0$.
	\item For every $i \neq j \in [k]$, $\widetilde{\cL}_{i,j} \leq 0$
	\item $	\frac{\varphi^2}{3} \cL' - \frac{\varphi^2}{3} \cdot \frac{\xi dn}{k} \cdot I_k \preceq \wt{\cL} \preceq 2 \cL' + \frac{2 \xi dn}{k} \cdot I_k $. 
\end{condenum}	   

First, we will show that this program is feasible by showing that $\mathcal{L}_H$ is a feasible solution. Since $\mathcal{L}_H$ is a graph Laplacian, it satisfies constraints (C1), (C2), (C3) and (C4). Now, consider constraint (C5). 

Using Lemma~\ref{lem:kk'-close} and Lemma~\ref{lem:k'-to-lh}, we see that, except with probability at most $n^{-100}$, the following both hold for the matrix $K$ defined in Line 5 of Algorithm~\ref{alg:quadratic}.

\begin{itemize}
    \item $K' - \xi n/k \cdot I_k \preceq K \preceq K' + \xi n/k \cdot I_k $, and
    \item $\cL_H/2 \preceq d(I - K') \preceq 3/\varphi^2 \cdot \cL_H$.
\end{itemize}

where $K'$ is a $k$-by-$k$ matrix defined as $K'_{i,j} = \mathds{1}_{C_i} U_{[k]}^T \Sigma_{[k]} U_{[k]} \mathds{1}_{C_j}$.

Let us assume that both of these conditions hold. Conditioned on this, the following is seen to hold

\begin{equation} \label{eq:sandwich-lh}
\frac{\varphi^2}{3} \cdot d(I - K) - \frac{\varphi^2}{3} \frac{\xi dn}{k} \cdot I_k \preceq \cL_H \preceq 2 \cdot d(I-K) + 2\frac{\xi dn}{k} \cdot I_k. 
\end{equation}

This is precisely constraint (C5), which means that $\cL_H$ indeed satisfies the constraints of the SDP and thus $\cL_H$ is a feasible solution. Further, the set of feasible solutions contains an open ball: in particular all Laplacian matrices $\cL'$ with $\|\cL' - \cL\|_F^2 \leq n^{-10}$ lie inside this ball. Therefore, in additional time $k^{O(1)} \cdot \log n$, the 
	Ellipsoid Algorithm returns some feasible solution $\wt{\cL}$. $\wt{\cL}$ satisfies all these 
	constraints as well. Writing $\cL' = d(I-K)$, this means

	\begin{equation} 
	\label{eq:k-l-tilde}
		\frac{\varphi^2}{3} \cL' - \frac{\varphi^2}{3} \cdot \frac{\xi dn}{k} \cdot I_k \preceq \wt{\cL} \preceq 2 \cL' + \frac{2 \xi dn}{k} \cdot I_k  \text{.} 
	\end{equation}
Therfore, By \eqref{eq:sandwich-lh}, we have
\[
		\frac{\varphi^2}{6}\cdot \cL_H -\varphi^2\cdot  \frac{\xi d n}{k} \cdot I_k  \preceq \wt{\cL} \preceq \frac{6}{\varphi^2}\cdot \cL_H + \frac{6}{\varphi^2} \cdot \frac{\xi  d n}{k} \cdot I_k. 
\]

\paragraph{Running time:} The overall running time is dominated by the call to $\textsc{QuadraticOracle}$ which finds the matrix
	$K$ and thus it is at most 
	$\runningtimeiv$. 

    Finally, we bound the number of seed queries issued. The algorithm estimates the cluster sizes $\widetilde{w}(i)$ within a multiplicative $(1 \pm \delta)$ factor
	with $\delta = \frac{\xi}{512\cdot k^2\cdot n^{40\cdot \epsilon/\varphi^2}}$. By simple Chernoff bounds, this can be done using  

	$$s = \frac{400\log n \cdot k^2}{\delta^2}  \leq \frac{10^9 \cdot n^{80 \e/\varphi^2} \cdot k^{6}  \cdot \log n}{\xi^2} $$ seeds (see Lemma \ref{lem:semi-supervise-model}). 
    
\end{proof}

\subsection{Proof of an intermediate Lemma used in Theorem~\ref{thm:contracted-graph} (Lemma~\ref{lem:kk'-close})} \label{sec:quadproof}

To finish the proof of Theorem~\ref{thm:contracted-graph}, we now prove Lemma~\ref{lem:kk'-close}.
Throughout this section, we assume that $G=(V,E)$ is a $d$-regular graph.

First, we will need Lemma \ref{lem:var}, which proves that for any pair of large enough subsets $A, B \subseteq V$
we can estimate the means of the weighted inner product between  the spectral embedding of vertices in $A$ and $B$
 i.e., $\frac{1}{|A|\cdot |B|} \sum_{\substack{x \in A\\y \in B}} \inner{f_x}{\Sigma_{[k]} f_y}$, by taking enough random samples from $S_A \subseteq A$ and $S_B \subseteq B$ and estimating the 
(weighted) inner product empirically. We prove Lemma \ref{lem:var} in Appendix \ref{sec:hoeffding:proof}).

\begin{restatable}{lemmma}{lemhoeffding} \label{lem:var}
	 Let $A,B\subseteq V$. Let $S_A \subseteq A$ and $S_B \subseteq B$ denote (multi)sets of vertices sampled independently and uniformly at random from $A$ and $B$ respectively, where $|S_A|, |S_B| \geq  \frac{1600\cdot k^3\cdot n^{40 \epsilon/\varphi^2}\cdot \log n}{\xi^2}$. Let $M$ denote the lazy random walk matrix of $G$, and $M=U\Sigma U^T$ be the eigendecomposition of $M$.  Then, with probability at least $1 - n^{-100\cdot k}$ we have 

	\begin{equation} \label{eq:apx-dp}
		\left| \mathds{1}_A^T \cdot (U_{[k]}\Sigma_{[k]}U_{[k]}^T) \ 
		\mathds{1}_B \ 
		- \frac{|A|\cdot |B|}{|S_A|\cdot |S_B|}\cdot \mathds{1}_{S_A}^T (U_{[k]}\Sigma_{[k]}U_{[k]}^T) \mathds{1}_{S_B} \right| \leq \xi\cdot n
	\end{equation}
\end{restatable}

Next, to prove the correctness of quadratic oracle we need the following lemma that bounds the $\ell_\infty$-norm on any unit vector in the eigenspace spanned by the bottom $k$ eigenvectors of $L_G$, i.e. $U_{[k]}$ proved by \cite{GluchKLMS21}.

\begin{restatable}[\cite{GluchKLMS21}]{lemmma}{lemmalinf}\label{lem:l-inf-bnd}
Let $\varphi\in(0,1)$ and $\epsilon \leq \frac{\varphi^2}{100}$, and let $G=(V,E)$ be a $d$-regular graph that admits $(k,\varphi,\e)$-clustering $C_1, \ldots, C_k$.  Let $u$ be a normalized eigenvector of $L$ with $||u||_2=1$ and with eigenvalue at most $2\epsilon$.     Then we have 
\[||u||_\infty \leq  n^{20\cdot\epsilon /\varphi^2}\cdot \sqrt{\frac{160}{\min_{i\in [k]}|C_i|}} \text{.} \]
\end{restatable}

Finally, we need the following lemma which shows that Algorithm~\ref{alg:quadratic} receives enough samples from 
every cluster in the $\rsm$. Lemma \ref{lem:semi-supervise-model} proves this 
fact. The proof relies on a simple Chernoff bound application and we defer it to 
Appendix \ref{sec:hoeffding:proof}.

\begin{restatable}{lemmma}{lemsemisupevise} \label{lem:semi-supervise-model}
Let $S\subseteq V$ denote a set of random vertices returned by the $\rsm$ (Definition \ref{def:ss:model}) in a regular graph. For every $i\in [k]$ let $S_i=S\cap C_i$. If $|S| \geq \frac{400\cdot \log n\cdot k^2}{\delta^2}$, then with probability at least $1-n^{-100\cdot k}$ for every $i\in[k]$ we have $|S_i|\in (1\pm \delta)\cdot |S|\cdot \frac{|C_i|}{n}$.
\end{restatable}

We are now ready to prove Lemma~\ref{lem:kk'-close}
\begin{proof} (Of Lemma~\ref{lem:kk'-close})
We will first bound the entrywise difference between the matrices $K$ and $K'$. In particular, we first show that for a fixed entry indexed by $i,j \in [k]$, we have $|K_{i,j} - K'_{i,j}| \leq \frac{\xi n}{2k^2}$. We will prove this via a repeated use of triangle inequality. For the reader's simplicity, we recall

\begin{itemize}
    \item $K_{i,j} = \frac{\widetilde w(i)}{|S_i|} \cdot \frac{\widetilde w(j)}{|S_j|} \cdot \sum_{\substack{x \in S_i\\y \in S_j}} \langle f_x, \Sigma_{[k]} f_y\rangle_{\text{apx}}$.
    
    \item $K'_{i,j} = \mathds{1}_{C_i} U_{[k]}^T \Sigma_{[k]} U_{[k]} \mathds{1}_{C_j}$

\end{itemize}

Now, we setup the stage to use triangle inequality as mentioned earlier. To this end, we introduce the following auxiliary quantities

\begin{itemize}
    \item $\alpha_{i,j} = \frac{\widetilde w(i)}{|S_i|} \cdot \frac{|\widetilde w(j)}{|S_j|} \cdot\langle \sum_{x \in S_i} f_x, \Sigma_{[k]} \sum_{y \in S_j} f_y \rangle$.
    \item $\beta_{i,j} = \frac{|C_i|}{|S_i|} \cdot \frac{|C_j|}{|S_j|} \cdot\langle \sum_{x \in S_i} f_x, \Sigma_{[k]} \sum_{y \in S_j} f_y \rangle$
\end{itemize}

We know

\begin{equation} \label{eq:setup-triangle-ineq}
    |K_{i,j} - K'_{i,j}| \leq |K_{i,j} - \alpha_{i,j}| + |\alpha_{i,j} - \beta_{i,j}| + |\beta_{i,j} - K'_{i,j}|
\end{equation}

We bound all the terms above. Towards the first term note that

\begin{align*}
    |K_{i,j} - \alpha_{i,j}| &\leq \frac{\widetilde w(i)}{|S_i|} \cdot \frac{\widetilde w(j)}{|S_j|} \cdot \Biggl|\sum_{\substack{x \in S_i \\ y \in S_j}} \Big(\langle f_x, \Sigma_{[k]} f_y \rangle_{apx} - \langle f_x, \Sigma_{k]} f_y \rangle \Big)\Biggr| \\
    &\leq \frac{\widetilde w(i)}{|S_i|} \cdot \frac{\widetilde w(j)}{|S_j|} \cdot \frac{\xi}{nk^2} &&\text{By Theorem~\ref{thm:wdp}} \\
    &\leq \frac{\xi n}{16 k^4} &&\text{By choice of $s$.}
\end{align*}

The second line above holds with probability at least $1 - 3n^{-100k}$ for any fixed pair of vertices. By a union bound over all pairs of sampled vertices, it holds with probability at least $1 - n^{-50k}$.

Next, we bound the second term in Equation~\eqref{eq:setup-triangle-ineq}. 

\begin{align*}
    |\alpha_{i,j} - \beta_{i,j}| &\leq |\langle \mathds{1}_{S_i}, \mathds{1}_{S_j} \rangle| \Biggl| \Big(\frac{\widetilde w(i) \cdot \widetilde w(j)}{|S_i| |S_j|} - \frac{|C_i| |C_j|}{|S_i| |S_j|} \Big) \Biggr| \\
    &\leq \left||C_i|\cdot |C_j| - (1-\delta)^2\cdot |C_i|\cdot |C_j|\right| \cdot \max_{x\in V}||f_x||_2^2 &&\text{As $||\Sigma_{[k]}||_2=\max_{i\in [k]} \left(1-\frac{\lambda_i}{2}\right)\leq 1$, and Lemma~\ref{lem:semi-supervise-model}} \nonumber \\
&\leq 4\cdot \delta\cdot |C_i|\cdot |C_j|\cdot \sum_{i=1}^k ||u_i||^2_\infty &&\text{As $||f_x||_2^2=\sum_{i=1}^k u_i^2(x)$} \nonumber \\
&\leq 4\cdot \delta\cdot |C_i|\cdot |C_j|\cdot k\cdot \frac{160\cdot n^{40\cdot \epsilon/\varphi^2}}{\min_{t\in [k]} |C_t|} &&\text{By Lemma \ref{lem:l-inf-bnd}} \nonumber \\
&\leq 320\cdot \delta\cdot  n^{40\cdot \epsilon/\varphi^2}\cdot n \nonumber \\
&\leq \xi n/16k^2 &&\text{As $\delta=\frac{\xi}{512\cdot k^2\cdot n^{40\cdot \epsilon/\varphi^2}}$}
\end{align*}

Finally, we consider the last term in Equation~\eqref{eq:setup-triangle-ineq}. Here, we seek to bound $|\beta_{i,j} - K'_{i,j}|$. Since $|S_i| \geq \frac{1600 k^7 n^{40 \e/\varphi^2} \cdot \log n}{\xi^2}$ for every $i$, it follows from Lemma~\ref{lem:var} that $|\beta_{i,j} - K'_{i,j}| \leq \xi n/16k^2$.

Overall, by an application of triangle inequality, this means that in every single entry, we have $|K_{i,j} - K'_{i,j}| = err(i,j) \leq \xi n/5k^2$. Let $E \in \R^{k \times k}$ denote the matrix whose $(i,j)$-th entry is $err(i,j)$. Observe that the matrix $\xi n/k \cdot I_k - E$ is symmetric and diagonally dominant and therefore, in the psd order we have $E \preceq \xi n/k \cdot I_k$ which implies the lemma.

\end{proof}


\subsection{Correctness of \textsc{EstimatedCost} (Proof of Theorem \ref{thm:estdcost})} \label{sec:sketch:tree}
Let $G=(V,E)$ be a $d$-regular graph that admits a $(k,\varphi,\epsilon)$-clustering $\cC = C_1, \ldots , C_k$, and  let $H = G/\cC$ be the of contraction of $G$ with respect to the partition $\cC $. We denote the optimal Dasgupta tree for $H$ as $T^*_H$. With this setup,
we will now show the main result (Theorem \ref{thm:estdcost}). Our estimation primitive $\textsc{EstimatedCost}$ (Algorithm \ref{alg:tree}) uses a simple vertex-weighted extension of a result of \cite{charikar2017approximate}
to find a 
tree $\wt{T}$ on $\wt{H}$. 
Let $\widetilde{T}=\textsc{WRSC}(\widetilde{H})$ (See Definition \ref{def:wrc} and Corollary \ref{cor:vertex-extention}). We denote the cost of $\wt{T}$ as $\wCT_{\wt{H}}(\wt{T})$, we present an estimator ($\estg$) 
for the Dasgupta cost of an optimal tree in $G$. In particular, we show $\estg \leq O\left( \frac{\sqrt{\log k}}{\varphi^{O(1)}} \right) \cdot \CT_G(T^*_G)$, where,
$$\estg \eqdef \est $$
The proof proceeds in two steps: in the first step, we prove Lemma \ref{lem:est:tstar}
which upper bounds $\estg$ in terms of $\wCT_H(T^*_H)$, 
where $T^*_{H}$ is the optimum Dasgupta tree for $H$. 
Next, we use Lemma \ref{lem:dcostg:tstar} to relate $\wCT_H(T^*_H)$ with $\CT_G(T^*_G)$.

\subsubsection{Estimating the cost of the contracted graph}
The main result of this section is Lemma \ref{lem:est:tstar}. This lemma bounds $\wCT_{\wt{H}}(\wt{T})$ in terms of $\wCT_H(T^*_H)$ where $T^*_H$ is the 
optimum Dasgupta tree for $H$. This is useful in relating $\estg$ with $\wCT_H(T^*_H)$. 

\begin{restatable}{lemmma}{lemesttilde} \label{lem:est:tstar}
Let $\frac{1}{n^4}<\xi<1$, $\xi'=\xi/16$. Let $H = G/\cC$ be the contraction of $G$ with respect to the partition $\cC $ (Definition \ref{def:weighetdcontracted}) and  let $T^*_H$ denote an optimum weighted Dasgupta tree for $H$.  Let $\widetilde{H}=\left([k], {[k] \choose 2}, \widetilde{W}, \widetilde{w}, \right)$ be the graph obtained by 
	$\textsc{ApproxContractedGraph}(G, \xi', \mathcal{D})$  (Algorithm \ref{alg:quadratic}). Let $\widetilde{T} = \textsc{WRSC}(\widetilde{H})$ denote a hierarchical clustering tree constructed on the graph $\widetilde{H}$ using the recursive sparsest cut algorithm. Then with probability at least $ 1-2\cdot n^{-100}$ we have
\[
\Omega\left(\varphi^2\right)\cdot \wCT_{H}(T^*_{H}) - \xi  m n k^2 \leq \wCT_{\wt{H}}(\widetilde{T}) \leq O\left( \frac{\sqrt{\log k}}{\varphi^2}\cdot \wCT_{H}(T^*_{H}) +  \xi  m n k^2 \sqrt{\log k} \right)
\]
\end{restatable}

To prove Lemma \ref{lem:est:tstar} we first present Definition \ref{def:maximal:clusters} and Lemma \ref{lem:cut-opt} which is a variation of Claim 2.1. from \cite{charikar2017approximate}.

\begin{mydef}[\textbf{Maximal clusters induced by a tree}] \label{def:maximal:clusters}
	Let $H=(V,E,W_H,w_H)$ be a vertex and edge-weighted graph. Let $T$ be a  tree with $|V|$ leaves (corresponding to the vertices of $H$). For any node $u$ of the tree $T$, we define the  weight of the node $u$ as the sum of the weight of those vertices of $H$ that are leaves of the subtree $T[u]$: \[ w_T(u)=\sum_{x\in \lv(T[u])} w_H(x) \text{.}\] 
Let $T(s)$ be the maximal nodes in $T$ such that their weight is at most $s$:
\[T(s) =\{u\in T: w_T(u)\leq s\} \text{.}\]
We refer to these nodes as maximal clusters of weight at most $s$. For convenience, we also define $T(s) =\lv(T)$ for every $s<\max_{x\in V} w(x)$.  Note that $T(s)$ is a partition of $V$.

We denote by $E_T(s)$ the edges that are cut in $T(s)$, i.e. edges with end points in different clusters in $T(s)$. For convenience, we also define $E_T(s) =E$ for every $s<\max_{x\in V} w(x)$. 
Also, we let $W_{T(s)}=\sum_{(x,y)\in E_T(s)} W(x,y)$ denote the total weight of the cut edges in $T(s)$. 
\end{mydef}

\begin{lemmma} \label{lem:cut-opt}
	Let $H=(V,E,W_H,w_H)$ be a vertex and edge-weighted graph. Let $\ell=\sum_{x\in V} w_H(x)$ be the total vertex weight of graph $H$. Let $T$ be a tree with $|V|$ leaves (corresponding to the vertices of $H$).  Then we have 
\[\wCT_{H}(T) = \sum_{s=0}^{\ell} W_T(s) \text{.}\]
\end{lemmma}

\begin{proof}
	The proof is identical to the proof given in \cite{charikar2017approximate}. 
Consider any edge $(x,y) \in E$.	
Let $r = w_T(\lca(x,y))$ be the weight of the lowest common ancestor of $x,y$. Then, by Definition \ref{def:weightedcost} the contribution of the edge $(x,y)$ to the LHS
	is $r \cdot W(x,y)$. Also, note that $(x,y) \in E_T(s)$ for all $0 \leq s \leq r-1$. Hence, the contribution of the edge $(x,y)$ to RHS is also $r \cdot W(x,y)$.
\end{proof}

Next, we present Lemma \ref{lem:cut-cost} which shows that the Dasgupta cost of two graphs is close if the weight of every cut is similar in both graphs.
\begin{lemmma}
	\label{lem:cut-cost}
	Let $H=(V,E, W, w)$ and $H'=(V, E', W', w)$ be vertex and edge-weighted graphs. Let $\alpha, \beta >0$. Suppose that for every set $S\subseteq V$ we have 
	$ W'(S,V\setminus S) \leq \beta \cdot W(S,V\setminus S) + \alpha$. Let $T$ and $T'$ denote the optimum Dasgupta tree of $H$ and 
	$H'$ respectively. Then we have
	\[ \wCT_{H'}(T') \leq \beta\cdot \wCT_{H}(T) + \frac{\alpha}{2} \cdot |V|\cdot ||w||_1 \text{.}\]
\end{lemmma}

\begin{proof}
	The main idea in the proof is to use Lemma \ref{lem:cut-opt} to relate the cost of the tree to the cost of cuts. 
Let $\ell=||w||_1=\sum_{x\in V} w(x)$. 
For every $0\leq s\leq \ell$, let $B_1, B_2, \cdots, B_{t_s}$ denote the partition of maximal clusters of weight at most $s$
	induced by $T$ in graph $H$ (Definition \ref{def:maximal:clusters}). By Lemma \ref{lem:cut-opt} we have

\begin{equation}
\label{eq:wcost-cut}
\wCT_{H}(T) = \sum_{s=0}^{\ell} W_{T(s)} =  
	\frac{1}{2}\cdot \sum_{s=0}^{\ell} \sum_{i \in [t_s]} W(B_i, V\setminus \overline{B_i})\text{.}
\end{equation}
 Therefore, we have
	\begin{align*}
	&\wCT_{H}(T) \\
	&= 	\frac{1}{2}\cdot \sum_{s=0}^{\ell} \sum_{i \in [t_s]} W(B_i, V\setminus \overline{B_i}) &&\text{By \eqref{eq:wcost-cut}}\\
	&\geq \frac{1}{2} \cdot \sum_{s=0}^{\ell} \sum_{i \in [t_s]} \frac{W'(B_i, V\setminus \overline{B_i})-\alpha}{\beta}   &&\text{As for every $S\subseteq V $, $ W'(S,V\setminus S) \leq \beta \cdot W(S,V\setminus S) + \alpha$} \\
		&= \frac{1}{\beta}\cdot \sum_{s=0}^{\ell}  W'_{T(s)} - \frac{\alpha}{2\cdot \beta} \cdot \sum_{s = 0}^{\ell} t_s \\
		&\geq \frac{1}{\beta}\cdot \sum_{s=0}^{\ell}  W'_{T(s)} - \frac{\alpha}{2\cdot \beta} \cdot \ell\cdot |V| &&\text{As $t_s\leq |\lv(T')|=|V|$}\\
		& = \frac{1}{\beta} \cdot \wCT_{H'}(T) - \frac{\alpha}{2\cdot \beta} \cdot \ell\cdot |V|  &&\text{By Lemma \ref{lem:cut-opt}}\\
		& \geq \frac{1}{\beta} \cdot \wCT_{H'}(T') -  \frac{\alpha}{2\cdot \beta} \cdot \ell\cdot |V| &&\text{By optimality of  $T'$ on $H'$}
	\end{align*}
	Therefore, we get
\[ \wCT_{H'}(T') \leq \beta\cdot \wCT_{H}(T) + \frac{\alpha}{2} \cdot |V|\cdot ||w||_1 \text{.}\]
\end{proof}

Next, we prove the following lemma which is an important intermediate step towards Lemma \ref{lem:est:tstar}.

\begin{lemmma} \label{lem:general:est:tstar}
 Let $H = G/\cC$ be the contraction of $G$ with respect to the partition $\cC $ (Definition \ref{def:weighetdcontracted}) and  let $T^*_H$ denote an optimum weighted Dasgupta tree for $H$. Let  $\widetilde{H}= \left([k], {[k] \choose 2}, \widetilde{W}, \widetilde{w}\right)$ be an approximation of $H$ such that the following hold:
\begin{itemize}
\item for all $i\in[k]$, $ \frac{w(i)}{2}\leq \wt{w}(i) \leq 2 \cdot w(i)$, and 
\item for all $S \subseteq [k]$, $a \cdot W(S, \overline{S}) - b \leq \wt{W}(S, \overline{S}) \leq a' \cdot W(S, \overline{S}) + b'$.
\end{itemize}

Let $\widetilde{T} = \textsc{WRSC}(\widetilde{H})$ denote a hierarchical clustering tree constructed on the graph $\widetilde{H}$ using the recursive sparsest cut algorithm. Then
\[
		\frac{a}{2}\cdot \wCT_{H}(T^*_{H}) - b\cdot n \cdot k   \leq \wCT_{\wt{H}}(\widetilde{T}) \leq O\left(a' \sqrt{\log k} \cdot \wCT_{H}(T^*_{H}) + b' n \cdot k \cdot \sqrt{\log k}\right)
	\]
\end{lemmma}

\begin{proof}
	Let $T^*_{\wt{H}}$ be the optimum Dasgupta tree of $\wt{H}$. By Corollary \ref{cor:vertex-extention} 
	the tree $\wt{T}$ returned by the vertex-weighted recursive sparsest cut procedure satisfies 

	\begin{equation} \label{eq:est1:tstar}
		\wCT_{\wt{H}}(T^*_{\wt{H}})\leq \wCT_{\wt{H}}(\wt{T}) \leq O(\sqrt{\log k})\cdot \wCT_{\wt{H}}(T^*_{\wt{H}})
	\end{equation}
Recall $H=\left([k], {[k] \choose 2}, W, w\right)$ is the contraction of $G$ with respect to $\cC$, 
	i.e., $H=G/\cC=$  (Definition \ref{def:contracted}). 
	Recall that $\widetilde{H} =\left([k], {[k] \choose 2}, \widetilde{W}, \widetilde{w}\right)$ 
	where $\widetilde{H}$ satisfies items 1 and 2 in the premise. Now we define 
	$\widehat{H}=\left([k], {[k] \choose 2}, W, \widetilde{w}\right)$, 
	where $W$ is the matrix of edge weights in $H$, and $\widetilde{w}$ is the vector of vertex weights in $\widetilde{H}$. Let $T^*_{\widehat{H}}$ be the optimum Dasgupta tree of $\widehat{H}$.
	Note that the trees $T^*_{\widehat{H}}$ and $\wt{T}$ satisfy the pre-requisites of
	Lemma \ref{lem:cut-cost}. Therefore, we have 

	\begin{equation}
		\label{eq:qhat-qtilde}
		a \cdot \wCT_{\widehat{H}}(T^*_{\widehat{H}}) - \frac{b}{2} \cdot k \cdot ||\widetilde{w}||_1   \leq \wCT_{\wt{H}}(T^*_{\wt{H}}) \leq a' \cdot \wCT_{\widehat{H}}(T^*_{\widehat{H}}) + \frac{b'}{2} \cdot k \cdot ||\widetilde{w}||_1  
	\end{equation}

	Note that the vertex weight function $w$ of $H$ satisfies $w_i=|C_i|$ 
	for all $i\in [k]$. Also, recall from the premise, we have 
	$\frac{1}{2}\cdot w(i)\leq\widetilde{w}(i)\leq 2\cdot w(i)$. Therefore, by 
	Definition \ref{def:weightedcost} we have

	\begin{align}
		\label{eq:tt'11}
		\wCT_{\widehat{H}}(T^*_{\widehat{H}}) 
		&\leq \wCT_{\widehat{H}}(T^*_{H}) &&\text{By optimality of $T^*_{\widehat{H}}$ on $\widehat{H}$} \nonumber\\
		&\leq 2\cdot \wCT_{H}(T^*_{H}) &&\text{As $\widetilde{w}(i)\leq 2\cdot w(i)$}
	\end{align}

	Similarly, we have

	\begin{align}
		\label{eq:tt'22}
		\wCT_{\widehat{H}}(T^*_{\widehat{H}}) 
		&\geq  \frac{1}{2}\cdot \wCT_{H}(T^*_{\widehat{H}}) &&\text{As $\widetilde{w}(i)\geq \frac{1}{2}\cdot w(i)$ on $\widehat{H}$} \nonumber \\
		&\geq \frac{1}{2}\cdot \wCT_{H}(T^*_{H}) &&\text{By optimality of $T^*_{\widehat{H}}$}
	\end{align}

	Therefore, by \eqref{eq:qhat-qtilde}, \eqref{eq:tt'11}, \eqref{eq:tt'22}, and as 
	$||\widetilde{w}||_2\leq 2||w||_2=2n$, we have

	\begin{equation}
		\label{eq:qtil-q-don?}
		\frac{a}{2} \cdot \wCT_{H}(T^*_{H}) - b\cdot n \cdot k   \leq \wCT_{\wt{H}}(T^*_{\wt{H}}) \leq 2\cdot a' \cdot \wCT_{H}(T^*_{H}) + b'\cdot n \cdot k 
	\end{equation}

	Finally, by \eqref{eq:est1:tstar} and \eqref{eq:qtil-q-don?} we get
	\[
		\frac{a}{2} \cdot \wCT_{H}(T^*_{H}) - b\cdot n \cdot k \leq \wCT_{\wt{H}}(\widetilde{T}) \leq O\left( a' \sqrt{\log k} \cdot \wCT_{H}(T^*_{H}) + b'\cdot n\cdot k\sqrt{\log k} \right).
	\]
\end{proof}

Finally, we prove Lemma \ref{lem:est:tstar}.
\lemesttilde*

\begin{proof} 

	Recall that $\widetilde{H}=\textsc{ApproxContractedGraph}(G, \xi', \cD)$ and let $\wt{\mathcal{L}}$ 
	be  the Laplacian of $\widetilde{H}$. Therefore, by Theorem \ref{thm:contracted-graph} with 
	probability at least $1 - n^{-100}$, we have

	\begin{equation}
		\label{eq:l-lh-z}
		\Omega\left(\varphi^2\right)\cL_H - \xi\cdot m \cdot I_k \preceq \wt{\cL}  \preceq O\left(\frac{1}{\varphi^2}\right)\cdot \cL_H  + \xi\cdot m \cdot I_k
	\end{equation}

	Note that for every $S \subseteq  [k]$, we have $\mathds{1}_{S}^T \cL_H \mathds{1}_{S} = W(S,\overline{S})$, 
	and $\mathds{1}_{S}^T  \wt{\cL} \mathds{1}_{S} =  \wt{W}(S, \overline{S})$. Also, note that 
	$||\mathds{1}_{S}||_2^2\leq k$. Therefore, by \eqref{eq:l-lh-z} with probability at least 
	$1 - n^{-100}$ for every $S\in [k]$ we have

	\begin{equation}
		\label{eq:cut-prop}
		\Omega\left(\varphi^2\right)\cdot W(S, V\setminus S) -\xi \cdot m\cdot k \leq \wt{W}(S, V\setminus S) \leq O\left(\frac{1}{\varphi^2}\right) \cdot W(S, V\setminus S) + \xi \cdot m\cdot k
	\end{equation}
	Moreover, note that by Lemma \ref{lem:semi-supervise-model} with probability at 
	least $1 - n^{-100\cdot k}$ we have $\frac{1}{2}\cdot w(i)\leq\widetilde{w}(i)\leq 2\cdot w(i)$. 
	Now, apply Lemma \ref{lem:general:est:tstar}.  With probability at least 
	$1 - n^{-100\cdot k}-n^{-100}\geq 1-2\cdot n^{-100}$, this gives
	\[
		\Omega\left(\varphi^2\right)\cdot \wCT_{H}(T^*_{H}) - \xi  m n k^2 \leq \wCT_{\wt{H}}(\widetilde{T}) \leq O\left( \frac{\sqrt{\log k}}{\varphi^2}\cdot \wCT_{H}(T^*_{H}) +  \xi  m n k^2 \sqrt{\log k} \right)
	\]
\end{proof}

\subsubsection{Bounding the optimum cost of the graph with the contracted graph} \label{sec:opt:contracted:cost}

The main result of this section is Lemma~\ref{lem:dcostg:tstar}, that relates the cost of the optimum tree of the contracted graph with the cost of the optimum of $G$. This allows us to bound the estimator proposed in Algorithm \ref{alg:tree} with the Dasgupta cost of the optimum tree of $G$.

\begin{restatable}{lemmma}{optcost} \label{lem:dcostg:tstar}
Let $H = G/\cC$ be the contraction of $G$ with respect to the partition $\cC $ (Definition \ref{def:weighetdcontracted}). Let $T^*_H$ and $T^*_G$ be optimum weighted Dasgupta trees for $H$ and $G$ respectively. Then we have
\[ \CT_G(T^*_G) \leq \wCT_H(T^*_H) + \textsc{TotalClustersCost}(G) \leq O\left(\frac{1}{\varphi^7}\right) \cdot \CT_G(T^*_G) \text{,}\]
where $\textsc{TotalClustersCost}(G)$ is an output of Algorithm \ref{alg:clustercost} which satisfies the guarantees of Theorem \ref{thm:Gk}
\end{restatable}

\begin{restatable}{theorem}{totalclustercost}
\label{thm:Gk}
Let $G=(V,E)$ be a $d$-regular $(k,\varphi,\epsilon)$-clusterable graph. For every $i\in[k]$, let $G\{C_i\}$ denote the induced subgraph on $C_i$ with added self loops so that the degrees in $G\{C_i\}$ and $G$ are the same, and let $T^*_i$ denote the tree
	with optimum Dasgupta cost for $G\{C_i\}$. 
Then procedure \textsc{TotalClustersCost} (Algorithm \ref{alg:clustercost}) returns a value such that:
\[\sum_{i \in [k]} \CT_{G\{C_i\}}(T^*_i) \leq \textsc{TotalClustersCost}(G) \leq O\left(\frac{1}{\varphi^{5}}\right)\cdot \sum_{i \in [k]} \CT_{G\{C_i\}}(T^*_i) \text{.} \]
\end{restatable} 
We prove Theorem \ref{thm:Gk} in Section \ref{sec:clustercost}. 
To prove Lemma \ref{lem:dcostg:tstar}, we show that there exists a tree on the contracted graph whose cost is not more than $\frac{1}{\varphi^{O(1)}}$ times the optimum cost of the graph (see Lemma \ref{lem:optimum-respects-clusters}). To show this, we exploit the structure of $(k, \varphi, \e)$-clusterable graphs and prove that in these graphs some {\em cluster respecting cut} has conductance comparable to the conductance of the sparsest cut (see Lemma \ref{lem:cluster-respecting-cuts}).

\begin{mydef}[\textbf{Cluster respecting cuts}] \label{def:cluster:respecting:cut}
	Let $G$ be a graph that admits a $(k, \varphi, \epsilon)$-clustering
	$C_1, \ldots C_k$. We say that the the cut $(B,\overline{B})$ is a cluster respecting cut with 
	respect to the partition $\cC$ if both $B$ and $\overline{B}$
	are disjoint unions of the clusters. That is, there exists a subset $I \subseteq [k]$ such that 
	$B = \cup_{i \in I} C_i$, and $\overline{B} = \cup_{i \not\in I} C_i$.
\end{mydef}

\begin{mydef}[\textbf{Cluster respecting  tree}] \label{def:cluster:resp:tree}
Let $G=(V,E)$ be a graph that admits a $(k, \varphi, \epsilon)$-clustering $\cC=C_1, \ldots C_k$. We say that tree $T$ for $G$ (with $|V|$ leaves) is cluster respecting with respect to the partition $\cC $ if there exists a subtree $T_{[k]}$ of $T$ (rooted at the root of $T$) with $k$ leaves such that for every $i \in [k]$, then there exists a unique leaf $\ell_i$ (of $T_{[k]}$) such that the leaves in $T$ which are descendants of $\ell_i$ are exactly the set $C_i$. We call the tree $T_{[k]}$ the contracted subtree of $T$.
\end{mydef}

\begin{lemmma}[Some cluster respecting cut has conductance comparable to the sparsest cut]  
	\label{lem:cluster-respecting-cuts}
  Let $(S, V \setminus S)$
	denote the sparsest cut of $G$. Then there exists a cluster respecting cut
	(Definition \ref{def:cluster:respecting:cut}) $(B, V \setminus B)$ such that \[\max(\phi_{\text{out}}^G(B), \phi_{\text{out}}^G(\overline{B})) \leq \frac{4\cdot \phi_{\text{out}}^G(S)}{\varphi}  \text{.}\]
\end{lemmma}

\begin{proof}
	Let $(S,\overline{S})$ denote the cut with the smallest conductance in $G$, 
	where  $\vol(S) \leq \vol(V \setminus S)$. If $(S, \overline{S})$ is cluster respecting, 
	then we are already done. So, let us suppose it is not cluster respecting and let 
	$\beta = \phi_{\text{in}}(G) = \phi_{\text{out}}(S)$.
	For each $i \in [k]$, define $X_i = S \cap C_i$ to be the vertices of $S$ that belong to $C_i$. Similarly, for each $i \in [k]$, define $Y_i = \overline{S} \cap C_i$.
	Also, we use $M_i$ to denote for each $i \in [k]$, the set with smaller volume between $X_i$ and $Y_i$.
	That is, $M_i = \argmin(\vol(X_i), \vol(Y_i))$. Define $M = \cup_{i \in [k]} M_i$ and let 
	$$B = (S \setminus M) \cup (\overline{S} \cap M) \ and \ \overline{B} = V \setminus B.$$
	
	This makes $(B, \overline{B})$ a \emph{cluster respecting cut}. This is because
	we move vertices of $M$ from $\overline{S}$ to $S$ to get $B$ and the other way around
	to get $\overline{B}$. We will show that the sets $B$ and $\overline{B}$ are both non-empty 
	and we will upperbound 
	$\max(\phi_{\text{out}}(B), \phi_{\text{out}}(\overline{B})) = \displaystyle\frac{|E(B, \overline{B})|}{\min(\vol(B), \vol(\overline{B}))}$. 
	By definition of $B$, it holds that
	$$|E(B, \overline{B})| \leq |E(S, \overline{S})| + \vol(M) = \beta \vol(S) + \vol(M).$$
	Also $\cup_{i \in [k]} E(X_i, Y_i) \subseteq E(S, \overline{S})$. Moreover, recalling that
	each cluster has $\phi_{\text{in}}(C_i) \geq \varphi$, it follows that 
	$|E(S, \overline{S})| \geq \sum_{i \in [k]} |E(X_i, Y_i)| \geq \varphi \cdot \vol(M).$ 
	Therefore, $\varphi \cdot \vol(M) \leq |E(S, \overline{S})| = \beta \cdot \vol(S)$. This means

	\begin{equation} \label{eq:edge-bound}
		|E(B, \overline{B})| \leq \beta \cdot \vol(S) + \frac{\beta \cdot \vol(S)}{\varphi} \leq \frac{2 \beta \cdot \vol(S)}{\varphi}
	\end{equation}

	We define an index set
	$I_{\text{small}} = \{i \in [k]: \vol(X_i) \leq \vol(Y_i) \}$ which indexes clusters where $S$ contains smaller volume
	of $C_i$ than $\overline{S}$. We show that if $\sum_{i \in I_{\text{small}}}\vol(X_i) < \vol(S)/10$, 
	then $\min(\vol(B), \vol(\overline{B})) \geq 0.5 \vol(S)$. Moreover, we also show that the
	other case, with $\sum_{i \in I_{\text{small}}}\vol(X_i) < \vol(S)/10$ cannot occur. In all, this
	means $\max(\phi_{\text{out}}(B), \phi_{\text{out}}(\overline{B})) \leq \frac{4 \beta}{\varphi} \cdot \phi_{\text{out}}(S)$. 
		Let us consider the first situation above.

	\textbf{Case 1:}  Suppose $\sum_{i \in I_{\text{small}}} \vol(X_i) < \vol(S)/10$.
	Note that for any $i \not\in I_{\text{small}}$, $B \supseteq X_i$. Thus, we have
	$\vol(B) \geq 0.9\vol(S) \geq 0.5\vol(S)$ (and thus $B$ is non empty). Also, note that in this case, 
	$$\vol(\overline{B}) \geq \vol(\overline{S}) - \vol(M) \geq \vol(S) - \vol(M) \geq \vol(S) - \frac{\beta}{\varphi} \vol(S)$$
	which is at least $0.5 \vol(S)$ as well (and in particular, $\overline{B}$ is also non empty). 
	Thus, $\min(\vol(B), \vol(\overline{B})) \geq 0.5 \vol(S)$.
	This gives $\max(\phi_{\text{out}}(B), \phi_{\text{out}}(\overline{B}) \leq \frac{4 \beta}{\varphi}$.
	In either case, we note that the cut $(B, \overline{B})$ is cluster respecting and satisfies that
	$\max(\phi_{\text{out}}(B), \phi_{\text{out}}(\overline{B})) \leq \frac{4 \beta}{\varphi}$.


	\textbf{Case 2:} Now, we rule out the case $\sum_{i \in I_{\text{small}}} \vol(X_i) \geq \vol(S)/10$. 
	Note that all $X_i$'s are disjoint (as they are contained in different
	clusters). 
	Observe

	$$|E(S, \overline{S})| \geq \sum_{i \in I_{\text{small}}} |E(X_i, Y_i)| \geq \varphi \sum_{i \in I_{\text{small}}} \vol(X_i) \geq \varphi \cdot \frac{\vol(S)}{10}.$$

	The first step follows because for all $i \in I_{\text{small}}$,  $E(X_i, Y_i) \subseteq E(S, \overline{S})$. 
	Moreover, $(X_i, Y_i)$ is a cut of the cluster $C_i$ with $\phi_{\text{in}}(C_i) \geq \varphi$.
	However, this means that $\phi_{\text{out}}(S) \geq \varphi/10$. But $G$ is $(k, \varphi, \e)$-clusterable
	and $S$ is the sparset cut in $G$. Thus, it better hold that $\phi_{\text{out}}(S) \leq \e < \varphi/10$ 
	which leads to a contradiction.
\end{proof}


The following claim is an aside which proves the tightness of Lemma \ref{lem:cluster-respecting-cuts}.
This claim shows that the $O(1/\varphi)$ loss in approximation to conductance is inherent when we 
take a cluster respecting cut as opposed to the sparsest cut. The proof can be found in 
Appendix \ref{sec:hoeffding:proof}.

\begin{restatable}[Tightness of Lemma~\ref{lem:cluster-respecting-cuts}]{claim}{clmclusterresp} \label{clm:cluster:resp:lb}
	Let $d > 3$ be a constant. 
	Then, there exist a $(2, \varphi, \e)$ clusterable, $d$-regular graph $G$ such that
\[\min \left(\phi_{\text{out}}(B), \phi_{\text{out}}(V\setminus B) \right) \geq \phi_{\text{in}}(G) \text{,}\]
	where $(B, \overline{B})$ is the unique cluster respecting cut of $G$.
\end{restatable}

The following observation shows that in a cluster respecting cut $(B, V\setminus B)$, the sets $B, \overline{B}$ induce $(k', \varphi, \e)$-clusterable graphs themselves. This is later used to prove Lemma 
\ref{lem:optimum-respects-clusters}.

\begin{observation} \label{obs:inherit}
	 Let $(B, V\setminus B)$ be a cluster 
	respecting cut in $G$  with respect to the partition $\cC$ (Definition \ref{def:cluster:respecting:cut}). Suppose that $B$ contains 
	$k' < k$ clusters in $G$.  For every $S\subseteq V$, let $G\{S\}$ be a graph obtained by adding  $d_x - d^S_x$ self-loops to every 
	vertex $x\in S$, where $d^S_x$ is the degree of vertex $x$ in $S$ and $d_x$ denotes the original degree of $x$ in $G$. Then, 
	we have $G\{B\}$ is $(k', \varphi, \e)$-clusterable and $G\{V\setminus B\}$ is $(k-k', \varphi, \e)$-clusterable.
\end{observation}
\begin{proof}
	We will show this observation for $B$. A similar argument holds for $V \setminus B$.
	For every cluster $C_i \subseteq B$, we have
	\[\phi^{G\{B\}}_{\text{out}}(C_i) = \frac{|E(C_i, B \setminus C_i)|}{\vol(C_i)} \leq \frac{\e\cdot \vol(C_i)}{\vol(C_i)} = \e\text{.}\]
	Also, note that for every cluster  $C_i \subseteq B$ and every $S \subseteq C_i$ with $\vol(S) \leq \vol(C_i)/2$ we have
	\[\phi^{G\{B\}}_{\text{in}}(S) = \frac{|E(S, C_i \setminus S)|}{\vol(S)} \geq \varphi \text{.}\]
	Therefore, for every cluster $C_i \subseteq B$ we have $\phi^{G\{B\}}_{\text{out}}(C_i)\leq \e$ and 
	$\phi^{G\{B\}}_{\text{in}}(C_i)\geq \varphi$. Thus, $B$ is $(k', \varphi, \e)$-clusterable.
\end{proof}

As a corollary of Lemma \ref{lem:cluster-respecting-cuts}, we get that there exists a cluster respecting cut which is approximately sparsest, as per Definition \ref{def:wrc}. 
\begin{mycor}\label{cor:wrc}[Some cluster respecting cut is an approximate vertex weighted sparest cut] 
  Let $(S, V \setminus S)$ denote the cut minimizing the sparsity $\spp_G(S)$ (Definition \ref{def:wrc}).
     Then there exists a cluster respecting cut
	(Definition \ref{def:cluster:respecting:cut}) $(B, V \setminus B)$ such that
    $$\spp_G(B) \leq  \frac{8}{\varphi} \cdot \spp_G(S)\text{.}$$
\end{mycor}
\begin{proof}
Let $(S, \bar S)$ denote the cut with the smallest conductance in $G$, with $|S| \leq |V \setminus S|$. Since the graph $G$ is unweighted, from Definition \ref{def:wrc}, we obtain 
    \begin{equation}\label{eq:cut_S}
        \spp_G(S) = \frac{|E(S,V \setminus S)|}{|S| \cdot |V \setminus S|} \geq \frac{|E(S,V \setminus S)|}{|S| \cdot n} = \frac{d}{n} \cdot \frac{|E(S,V \setminus S)|}{d|S| } =   \frac{d}{n} \phi_{\text{out}}^G(S). 
    \end{equation}
By Lemma \ref{lem:cluster-respecting-cuts}, there exists a cluster respecting cut $(B, \bar B)$ with $|B| \leq |V \setminus B|$ 
such that
\begin{equation}\label{eq:S_and_B}
    \phi_{\text{out}}^G(B) \leq \frac{4 \phi_{\text{out}}^G(S)}{\varphi}. 
\end{equation}

By Definition \ref{def:wrc} again, we obtain
\begin{equation}\label{eq:cut_B}
    \spp_G(B) = \frac{|E(B, V \setminus B)}{|B|\cdot  |V \setminus B|} \leq \frac{|E(B, V \setminus B)}{|B|\cdot n/2} = \frac{2d}{n}\cdot \frac{|E(B, V \setminus B)}{d|B|} = \frac{2d}{n}\phi_{\text{out}}^G(B),
\end{equation}
where the inequality follows since $ |V \setminus B | \geq n/2$ by the assumption that $|B| \leq |V \setminus B |.$

Combining Equations \eqref{eq:cut_S}, \eqref{eq:S_and_B} and \eqref{eq:cut_B}, gives 
$$ \spp(B) \leq \frac{2d}{n} \cdot \phi_{\text{out}}^G(B) \leq \frac{2d}{n} \cdot \frac{4 \phi_{\text{out}}^G(S)}{\varphi} \leq \frac{8}{\varphi}\cdot \spp(S).$$
\end{proof}

Next, we present Lemma \ref{lem:optimum-respects-clusters}  which uses Corollary \ref{cor:wrc} to show that there exists a cluster respecting tree whose cost is not more than $O\left(\frac{1}{\varphi} \right)$ times the optimum cost of the graph.

\begin{lemmma} \label{lem:optimum-respects-clusters}
	Let $T^*_G$ be a tree with the optimum Dasgupta cost on $G$. Then, there exists a cluster respecting tree 
	$T$ with respect to $\cC$ on $G$ (Definition \ref{def:cluster:resp:tree}) such that 
	\[\CT_G(T) \leq O\left(\frac{1}{\varphi}\right)\cdot \CT_G(T^*_G) \text{.}\]
	Moreover the contracted tree of $G$ with respect to $\cC = (C_1, C_2, \ldots, C_k)$ 
	(Definition \ref{def:cluster:resp:tree}) is a binary tree.
\end{lemmma}
\begin{proof}
	Let $H = G/ \cC$ denote the contraction of $G$ with respect to $\cC$. Let $\cT_H$ 
	denote the set of trees supported on $k$ leaves (which correspond to the $k$ clusters).
	We will construct a cluster respecting tree $T$ for $G$ with the claimed properties.
	This is done in two phases. In the first phase, we show it is possible to
	obtain a tree $T' \in \cT_H$ by repeatedly applying Corollary \ref{cor:wrc}.
	In phase two, we extend this tree further by refining each of the $k$ leaves in $T'$ further 
	to obtain a tree $T$ for $G$. Assuming the $T'$ after phase one indeed belongs to $\cT_H$, we obtain $T$
	in the following way: For $i \in [k]$, take the leaf $C_i$ (in $T'$) and extend it $T^*_i$, where $T^*_i$ is the tree with optimum Dasgupta cost on induced graph $G\{C_i\}$.
	
	{\noindent \bf Phase 1: Producing $T'$.} We show how to generate $T'$ level by level.
	It will be convenient to attach to each node in $T'$ a set $I \subseteq [k]$ of indices.
	We begin by attaching $[k]$ to the root. An application of Corollary \ref{cor:wrc} produces 
	a cut $(B, \overline{B})$ with sparsity (as per Definition \ref{def:wrc}) within a factor $O(1/\varphi)$ of the 
	sparsest cut in $G$ which additionally is cluster respecting (see Definition 
	\ref{def:cluster:respecting:cut}). Make $B$ and $\overline{B}$ the two children of the root
	and attach to $B$ the index set $I_B \subseteq [k]$ which denotes the clusters contained in 
	$B$ (similarly define $I_{\overline{B}} \subseteq [k]$). Recurse and apply Corollary \ref{cor:wrc} to both $G\{B\}$ and $G\{\overline{B}\}$.  
	Note that we can do this because both $G\{B\}$ (resp. $G\{\overline{B}\}$) are 
	$(k', \varphi, \e)$-clusterable (resp $(k-k', \varphi, \e)$-clusterable) for some 
	$1 \leq k' < k$ (Observation \ref{obs:inherit}).
	Thus, in all, phase 1 returns a tree $T' \in \cT_H$.

	{\noindent \bf Phase 2: Extending $T'$} On termination, as seen earlier, phase 1 produces
	a tree $T'$ whose leaves correspond to the clusters $C_1, C_2, \ldots C_k$ in some order.
	Note that 

	\begin{equation} \label{eq:offset:t}
		\CT_G(T) = \wCT_H(T') + \sum_{i=1}^k \CT_{G\{C_i\}}(T^*_i). 
	\end{equation}

	We will show this cost is at most $O(1/\varphi) \cdot  \CT_G(T^*_G)$. To this end, let us consider another tree $T^{\text{RSC}}$ that is obtained by taking the tree $T'$ and 
	extending it by repeated applications of \emph{exact} recursive sparses cut procedure (with 
	approximation factor being $1$). Thus, the tree $T^{\text{RSC}}$ obtained this way
	is within a factor $O(1/\varphi)$ for the sparsest cut at each step. 
	By Theorem \ref{thm:charikar}, we get 
\begin{equation}
\label{eq:main-1-phi}
\CT_G(T^{\text{RSC}}) \leq O\left(\frac{1}{\varphi}\right)\cdot \CT_G(T^*_G)
\end{equation}	
	Let us rewrite this approximation in more detail. For $i \in [k]$, letting $T^{\text{RSC}}(i)$ 
	denote the tree obtained by applying recursive sparsest cut on each $C_i$. Thus, by \eqref{eq:main-1-phi} we get

	\begin{align}
		O\left(\frac{1}{\varphi}\right)\cdot \CT_G(T^*_G) \geq \CT_G(T^{\text{RSC}})  = \wCT_H(T') + \sum_{i \in [k]} \CT_{G\{C_i\}}(T^{\text{RSC}}(i))  
\label{eq:offset:rsc} 
	\end{align}
	Now two cases arise. First, consider the case where $\wCT_H(T') \leq \sum_{i\in [k]} \CT_{G\{C_i\}}(T^*_i)$.
	In this case, by \eqref{eq:offset:t} we have 
	$\CT_G(T^*_G) \leq \CT_G(T) \leq 2 \sum_{i\in [k]} \CT_{G\{C_i\}}(T^*_i)$. Also, 
	we have 
	$\CT_G(T^*_G) \geq  \sum_{i\in [k]} \CT_{G\{C_i\}}(T^*_i)$. Thus, we have 
	\[\CT_G(T) \leq 2\cdot \CT_G(T^*_G) \text{.}\]	
	Now consider the other case where $\wCT_H(T') >  \sum_{i\in [k]} \CT_{G\{C_i\}}(T^*_i)$.
	Recall from \eqref{eq:offset:t} that 
	$$\CT_G(T) = \wCT_H(T') + \sum_{i \in [k]} \CT_{G\{C_i\}}(T^*_i) \leq 2 \cdot \wCT_H(T') \leq O\left(\frac{1}{\varphi}\right) \cdot \CT_G(T^*_G)\text{,} $$
	where, the last inequality holds by \eqref{eq:offset:rsc}. Moreover, by construction, the tree $T'$ 
	is the a binary tree (and is the contracted tree of $T$) 
	as desired.

\end{proof}

Finally, we prove Lemma \ref{lem:dcostg:tstar}. 
\optcost*

\begin{proof}[Proof of Lemma \ref{lem:dcostg:tstar}] By Lemma \ref{lem:optimum-respects-clusters}, we know there is a cluster respecting tree $T$ (Definition \ref{def:cluster:resp:tree}) on $n$ leaves such that
	\begin{equation} \label{eq:tstar:to:tmod}
		\CT_G(T^*_G) \geq \Omega(\varphi) \cdot \CT_G(T)
	\end{equation}
Let  $T'$  be the corresponding contracted tree obtained from $T$ (Definition \ref{def:cluster:resp:tree}). Note that by construction  $T$ in Lemma \ref{lem:optimum-respects-clusters} we have $T'$ has $k$ leaves such that for every $i\in [k]$, we extend the leaf corresponding to $C_i$, by tree $T^*_i$, where, $T^*_i$ is the tree with the optimum Dasgupta cost on the induced subgraph $G\{C_i\}$.

Thus, we have
\begin{equation} \label{eq:2offset:t}
	\CT_G(T) = \wCT_H(T') + \sum_{i\in[k]} \CT_{G\{C_i\}}(T^*_i)
	\end{equation}
Therefore, we have

\begin{align}
\label{eq:lw-jam}
&\CT_G(T^*_G) \nonumber \\
&\geq \Omega(\varphi) \cdot \CT_G(T) &&\text{By \eqref{eq:tstar:to:tmod}}  \nonumber \\
&= \Omega(\varphi) \cdot \left( \wCT_H(T') + \sum_{i\in[k]} \CT_{G\{C_i\}}(T^*_i) \right) &&\text{By \eqref{eq:2offset:t}} \nonumber \\
&\geq \Omega(\varphi) \cdot \left( \wCT_H(T^*_H) + \sum_{i\in[k]} \CT_{G\{C_i\}}(T^*_i) \right) &&\text{By optimality of  $T^*_H$ on $H$}  \nonumber \\
&\geq \Omega(\varphi) \cdot \wCT_H(T^*_H) + \Omega(\varphi^7)\cdot \textsc{TotalClustersCost}(G) &&\text{By Theorem \ref{thm:Gk}} \nonumber \\
&\geq \Omega(\varphi^7) \cdot \left( \wCT_H(T^*_H) +  \textsc{TotalClustersCost}(G) \right). 
\end{align}

Also by Theorem \ref{thm:Gk} we have
\begin{align}
\CT_G(T^*_G) 
&\leq \wCT_H(T^*_H) + \sum_{i=1}^k \CT_{G\{C_i\}}(T^*_i)  \nonumber \\
&\leq  \wCT_H(T^*_H) + \textsc{TotalClustersCost}(G) \text{.} \nonumber \\
\end{align}

\end{proof}

\subsubsection{Proof of Theorem \ref{thm:estdcost}} \label{sec:proof:prepros}

Finally, to prove Theorem \ref{thm:estdcost}  we prove an intermediate step, which is Lemma \ref{lem:general:apx:cost}.

\begin{lemmma} \label{lem:general:apx:cost}
	Let $H = G/\cC$ be the contraction of $G$ with respect to the partition $\cC $ (Definition \ref{def:weighetdcontracted}) 
	and let $T^*_H$ denote an optimum weighted Dasgupta tree for $H$. Let $0 < a < 1 < a'$ and 
	$b \leq  \frac{a}{k d \sqrt{\log k}}$. Let  
	$\widetilde{H}= \left([k], {[k] \choose 2}, \widetilde{W}, \widetilde{w}\right)$ 
	be an approximation of $H$ such that the following hold: 

	\begin{itemize}[leftmargin=*]
		\item For all $i\in[k]$, $ \frac{w(i)}{2}\leq \wt{w}(i) \leq 2 \cdot w(i)$, and 
		\item $a \cdot \wCT_{H}(T^*_{H}) - b \cdot mn \leq \wCT_{\wt{H}}(\widetilde{T}) \leq O\left(a' \sqrt{\log k} \cdot \wCT_{H}(T^*_{H}) + b \cdot mn \cdot \sqrt{\log k}\right)$
	\end{itemize}
	where, $\widetilde{T} = \textsc{WRSC}(\widetilde{H})$ denote a hierarchical clustering tree constructed on the 
	graph $\widetilde{H}$ using the recursive sparsest cut algorithm.
	We set $\estg = \frac{1}{a} \cdot \wCT_{\wt{H}}(\wt{T}) + \frac{b}{a} mn + \textsc{TotalClustersCost}(G)$.
	Then we have
	\[\CT_G(T^*_G) \leq \estg \leq O\left(\frac{a'\sqrt{\log k}}{a\cdot \varphi^7}\right)\cdot \CT_G(T^*_G)\text{.}\]
\end{lemmma}

\begin{proof}
	Let $\widetilde{H}= \left([k], {[k] \choose 2}, \widetilde{W}, \widetilde{w}\right)$
	denote the graph defined in the premise. Let $\widetilde{T} = \textsc{RSC}(\widetilde{H})$ 
	denote a hierarchical clustering tree constructed on the graph $\widetilde{H}$ using the recursive 
	sparsest cut algorithm. We have that

	\begin{equation}
		\label{eq:est:tstar:gen}
		a \cdot \wCT_{H}(T^*_{H}) - b\cdot mn   \leq \wCT_{\wt{H}}(\widetilde{T}) \leq O\left(a' \cdot \sqrt{\log k} \cdot \wCT_{H}(T^*_{H}) + b\cdot mn \sqrt{\log k} \right)
	\end{equation}
	Note that
	\begin{equation}
		\label{eq:est-val:gen}
		\estg = \estgen \text{.}
	\end{equation}
Therefore, by \eqref{eq:est:tstar:gen} and \eqref{eq:est-val:gen} we have
	\begin{align}
	\label{eq:simple-next}
		&\wCT_{H}(T^*_{H})+ \textsc{TotalClustersCost}(G)  \leq \estg  \nonumber\\
		&\leq O\left(\frac{a' \cdot \sqrt{\log k}}{a} \cdot \wCT_{H}(T^*_{H}) + \frac{b\cdot mn}{a} \cdot \sqrt{\log k}  \right) + \textsc{TotalClustersCost}(G). 
	\end{align}
Let $T^*_G$ denote a Dasgupta tree with optimum cost for $G$. Then, by Lemma~\ref{lem:dcostg:tstar} we have

	\begin{equation} \label{eq:dcostg:tstar}
		\CT_G(T^*_G) \leq \wCT_{H}(T^*_{H})+ \textsc{TotalClustersCost}(G) \leq \left(\frac{1}{\varphi^{7}}\right)\cdot \CT_G(T^*_G).
	\end{equation}
	By the first part of both \eqref{eq:simple-next} and \eqref{eq:dcostg:tstar} we have
	
	\begin{equation}
		\label{eq:t*g-min}
		\CT_G(T^*_G) \leq  \estg. 
	\end{equation}

We also have
	\begin{align}
	\label{eq:cost-done:)}
	&\estg  \nonumber\\
	&\leq O\left(\frac{a' \cdot \sqrt{\log k}}{a}\cdot \wCT_{H}(T^*_{H}) + \frac{b\cdot mn \cdot \sqrt{\log k}}{a}  \right) + \textsc{TotalClustersCost}(G) &&\text{By  \eqref{eq:simple-next}} \nonumber\\
	&\leq O\left( \frac{a' \sqrt{\log k}}{a\cdot \varphi^{7}} \right) \CT_G(T^*_G)  + O\left(\frac{b\cdot d \cdot n^2 \cdot \sqrt{\log k}}{a}  \right) + \textsc{TotalClustersCost}(G) &&\text{By \eqref{eq:dcostg:tstar}} \nonumber\\
	&\leq O\left( \frac{a' \sqrt{\log k} }{a \cdot \varphi^{7}} \right) \CT_G(T^*_G) + O\left(\frac{n^2}{k}\right)+  \textsc{TotalClustersCost}(G) &&\text{As $b \leq \frac{a}{d \cdot k \sqrt{\log k}}$} \nonumber\\
	&\leq O\left( \frac{a' \sqrt{\log k}}{a \varphi^{7}} \right) \CT_G(T^*_G) + O(1)\cdot \textsc{TotalClustersCost}(G)  &&\text{By Observation \ref{lem:lw-cost-exp}}\nonumber\\
	&\leq O\left( \frac{a' \sqrt{\log k}}{a \varphi^{7}} \right) \CT_G(T^*_G)+ O\left(\frac{1}{\varphi^{6}}\right)\cdot \CT_G(T^*_G) &&\text{By Theorem \ref{thm:Gk}} \nonumber\\
	&\leq O\left( \frac{a' \sqrt{\log k}}{a \cdot \varphi^{7}} \right) \CT_G(T^*_G). 
	\end{align}

	Therefore, by \eqref{eq:t*g-min} and \eqref{eq:cost-done:)} we have
	\[\CT_G(T^*_G) \leq \estg \leq O\left( \frac{a' \sqrt{ \log k}}{a\cdot \varphi^7} \right)\cdot \CT_G(T^*_G) \text{.}\]

\end{proof}

\thmestdcost*

\begin{proof}
Let $H = G/\cC$ be the contraction of $G$ with respect to the partition $\cC $ (Definition \ref{def:weighetdcontracted}) and  let $T^*_H$ denote an optimum weighted Dasgupta tree for $H$. Let $\widetilde{H}=\left([k], {[k] \choose 2}, \widetilde{W}, \widetilde{w} \right)$ be the graph obtained by 
	$\textsc{ApproxContractedGraph}(G, \xi, \mathcal{D})$ (Algorithm \ref{alg:quadratic}). Let $\widetilde{T} = \textsc{RSC}(\widetilde{H})$ denote a hierarchical clustering tree constructed on the graph $\widetilde{H}$ using the recursive sparsest cut algorithm. Therefore, by Lemma \ref{lem:est:tstar} with probability at least $ 1-2\cdot n^{-100}$ we have
	\begin{equation}
\label{eq:est:tstar}
	\Omega(\varphi^2)\cdot \wCT_{H}(T^*_{H}) - \xi m n k^2   \leq \wCT_{\wt{H}}(\widetilde{T}) \leq O\left( \frac{\sqrt{\log k}}{\varphi^2}\cdot \wCT_{H}(T^*_{H}) + \xi  m n k^2 \sqrt{\log k} \right).
	\end{equation}
	Note that as per line \eqref{ln:est-set} of Algorithm \ref{alg:tree} we estimate the Dasgupta cost of $G$ by 
	\begin{equation}
	\label{eq:est-val}
	\estg = \est \text{.}
	\end{equation}

	Set $a = c\cdot\varphi^2$, where $c$ is the hidden constant in $\Omega(\varphi^2)$. Set $a' = 1/\varphi^2, b=\xi\cdot k^2$,  where  $\xi = \frac{c\cdot \varphi^2}{d \cdot k^3 \sqrt{\log k}}$ as per line 1 of Algorithm \ref{alg:tree}.
	Thus, $b =  \frac{c\cdot \varphi^2}{d \cdot k \sqrt{\log k}} \leq \frac{a}{d \cdot k \sqrt{\log k}}$. So, we can apply Lemma \ref{lem:general:apx:cost}, which gives
	\[\CT_G(T^*_G) \leq \estg \leq O\left( \frac{a' \sqrt{ \log k}}{a\cdot \varphi^7} \right)\cdot  \CT_G(T^*_G)  = \frac{\sqrt{\log k}}{\varphi^{11}} \cdot \CT_G(T^*_G) \text{.}\]
\textbf{Running Time:} Now, we prove the running time bound. First, the $\textsc{EstimatedCost}$ procedure calls $\textsc{WeightedDotProductOracle}$, which by Theorem \ref{thm:wdp} has running time $O^*\left(n^{1/2+O(\e/\varphi^2)}\cdot \left(\frac{1}{\xi}\right)^{O(1)}\right)$. 

Then, the $\textsc{EstimatedCost}$ procedure calls the $\textsc{ApproxContractedGraph}$ 
	procedure. 
    
By Theorem \ref{thm:contracted-graph}, this has running time $O^*\left(n^{1/2+O(\e/\varphi^2)}\cdot \left( \frac{1}{\xi}\right)^{O(1)}\right)$. Finally, recall the procedure $\text{WRSC}$ runs in time $\text{poly}(k)$. 
Therefore, the overall running time of $\textsc{EstimatedCost}$ procedure is seen to be 
	${n^{1/2+O(\e/\varphi^2)} \cdot \left(\frac{d \cdot k \cdot \log n}{\varphi \cdot \xi}\right)^{O(1)}}$. 
    Substituting in 
    \[\xi = \frac{\varphi^2}{d\cdot  k^3 \cdot \sqrt{\log k}},\] we get the required running time. 

	Finally, we bound the number of seed queries issued. First, let us
	consider the number of seed queries required by the $\textsc{ApproxContractedGraph}$ procedure to estimate the cluster sizes $\widetilde{w}(i)$.  These quantities are estimated to within a multiplicative $(1 \pm \delta)$ factor
	with $\delta = \frac{\xi}{512\cdot k^2\cdot n^{40\cdot \epsilon/\varphi^2}}$
	By simple Chernoff bounds, this can be done using  

	$$O\left( \frac{\log n \cdot k^2}{\delta^2} \right) = O\left( \frac{n^{80 \e/\varphi^2} \cdot k^{6}  \cdot \log n}{\xi^2} \right)$$ seeds. 
    Let us now bound the number of samples taken by this procedure by plugging in the value of $\xi = \frac{\varphi^2}{d \cdot  k^3 \cdot \sqrt{\log k}}$. This gives 
 $$s  = \frac{10^9\cdot \log n\cdot k^{12} \cdot n^{80\cdot \epsilon/\varphi^2} \cdot \log k \cdot d_{\text{avg}}^2}{\varphi^4}.$$

    Then, by Theorem \ref{thm:Gk} the number of seeds taken to compute $\textsc{TotalClustersCost}$ is $n^{1/3}\cdot \left(\frac{k\cdot \log n}{\varphi}\right)^{O(1)}$.
    
    Combining, we obtain that the total number of seed queries issued is $O^*\left(n^{1/3}+ n^{O(\e/\varphi^2)}\cdot (d \cdot )^{O(1)}\right)$ as claimed. 
\end{proof}

\subsection{Proof of Lemma \ref{lem:var}, Lemma \ref{lem:semi-supervise-model} and Claim \ref{clm:cluster:resp:lb}} \label{sec:hoeffding:proof}
\lemhoeffding*

\begin{proof}
	Let $a = \mathds{1}_A^T \cdot (U_{[k]}\Sigma_{[k]}^{1/2})$ and 
	$b = (\Sigma_{[k]}^{1/2}U_{[k]}^T) \mathds{1}_B$. This allows us
	to write the first term on LHS of Equation \eqref{eq:apx-dp} as $\inner{a}{b}$. 
	Analogously, define $a' = \frac{|A|}{|S_A|} \cdot \mathds{1}_{S_A}^T \cdot (U_{[k]} \Sigma_{[k]}^{1/2})$
	and $b' = \frac{|B|}{|S_B|} \cdot (\Sigma_{[k]}^{1/2}U_{[k]}^T) \mathds{1}_{S_B}$ so that
	the second term on LHS of Equation \eqref{eq:apx-dp} can be written as $\inner{a'}{b'}$. With this setup, we have the following:

	\begin{align}
		\left|\inner{a}{b} - \inner{a'}{b'} \right| = \left|\inner{a}{b-b'} + \inner{a-a'}{b'} \right| \leq \|a\|_2\|b-b'\|_2 + \|a-a'\|_2\|b'\|_2 \label{eq:upperbd-apx-dp}
	\end{align}
	where the last inequality follows by triangle inequality and Cauchy-Schwarz. Now by expanding out $a$ we get 
	\[a = \sum_{x \in A} \mathds{1}_x^T \left(U_{[k]}\Sigma_{[k]}^{1/2}\right) \text{.}\]
We show this quantity is estimated coordinate wise very well by the vector $a'\in \R^k$ defined as follows: 
\[a'= \frac{|A|}{|S_A|} \cdot \sum_{x \in S_A} \mathds{1}_x^T \left(U_{[k]}\Sigma_{[k]}^{1/2}\right) \text{.}\] 
Note that for all $i \in [k]$ we have $\E[a'(i)] = a(i)$. For any $i \in [k]$ we first show that with high probability $|a(i)-a'(i)|$ is small,  then by union bound we prove that $||a'-a||$ is small. Note that for every $i\in [k]$ we have $a(i)=\sum_{x \in A} \sqrt{\sigma_i}\cdot u_i(x)$, where $u_i$ is the $i$-th eigenvector of $M$ and $\sigma_i$ is the $i$-th eigenvalue of $M$. 
For every $x\in S_A$, let $Z_x$ be a random variable defined as $Z_x= |A|\cdot  \sqrt{\sigma_i}\cdot u_i(x)$. Thus, we have $a'(i)=\frac{1}{S_A}\cdot \sum_{x\in S_A} Z_x$, and $a(i)=\frac{1}{S_A}\cdot \sum_{x\in S_A} \E[Z_x]$. Therefore, by Hoeffding Bound 
	(Fact \ref{fact:hoeffding}), we have
	\begin{equation}
	\label{eq:hoef}
	\text{Pr} \left[|a'(i) - a(i)| \geq t \right] \leq \exp\left(\frac{-2|S_A| \cdot t^2}{\left(2 \cdot \max_{x\in S_A} |Z_x| \right)^2}\right). 
	\end{equation}
Next we need to bound $\max_{x\in S_A} |Z_x|$. Note that for every $i\in[k]$ and every $x\in V$ we have
	\begin{align}
	\label{eq:max-zx}
		|Z_x| &= |A|\cdot \sqrt {\sigma_i}\cdot |u_i(x)| \nonumber\\
		&\leq |A|\cdot ||u_i||_{\infty} &&\text{As $\sigma_i =1-\frac{\lambda_i}{2}\leq 1$} \nonumber\\
		&= |A| \cdot\frac{n^{20 \e/\varphi^2}}{\sqrt{\min_{i\in k} |C_i|}} &&\text{By Lemma \ref{lem:l-inf-bnd}} \nonumber\\
		&\leq \frac{|A|\cdot n^{20 \e/\varphi^2}}{\sqrt{\frac{n}{k}}} &&\text{As $\forall i\in k, |C_i| \approx \frac{n}{k}$}. 
	\end{align}
	Let $w_A=   \frac{|A|\cdot n^{20 \e/\varphi^2}}{\sqrt{\frac{n}{k}}}$ and $\beta = \frac{\xi}{2\cdot k \cdot n^{20 \e/\varphi^2}}$. By \eqref{eq:hoef} and \eqref{eq:max-zx} we have
	\begin{align*}
	\text{Pr} \left[|a'(i) - a(i)| \geq \beta\cdot w_A \right] \leq \exp\left(\frac{-2|S_A| \cdot (\beta\cdot w_A)^2}{4\cdot w_A^2 }\right) = \exp\left(-|S_A| \beta^2/2\right) \leq n^{-200k}\text{,}
	\end{align*}
	where the last inequality holds by choice of $|S_A| \geq  \frac{1600\cdot k^3\cdot n^{40 \epsilon/\varphi^2}\cdot \log n}{\xi^2}\geq 400\cdot k\cdot \log n\cdot \frac{1}{\beta^2}$. Thus, by a union bound over all $i \in [k]$, with  probability at least $1-k\cdot n^{-200k}$ we have 
\begin{equation}
\label{eq:aa;}
||a-a'||_2=\sqrt{\sum_{i=1}^k (a'(i) - a(i))^2 }\leq \sqrt{k}\cdot \beta\cdot w_A.
\end{equation}	
A similar analysis shows that with  probability at least $1-k\cdot n^{-200k}$ we have 
\begin{equation}
\label{eq:bb'}
||b-b'||_2\leq \sqrt{k}\cdot \beta\cdot w_B\text{,}
\end{equation}	
where, $w_B=   \frac{|B|\cdot n^{20 \e/\varphi^2}}{\sqrt{\frac{n}{k}}}$. Also note that 
\begin{align}
\label{eq:anorm}
||a||_2 &= \left|\left|\mathds{1}_A^T \left(U_{[k]}\Sigma_{[k]}^{1/2}\right) \right|\right|_2 \nonumber\\
&\leq \left|\left| \mathds{1}_A^T   \right|\right|_2 \cdot \left|\left| U_{[k]} \right|\right|_2 \cdot \left|\left| \Sigma_{[k]}^{1/2}\right|\right|_2 \nonumber\\
&\leq \sqrt{|A|} \cdot \max_{i\in [k]}\sqrt{\sigma_i} &&\text{As $||U_{[k]}||_2=1$, and $\Sigma$ is diagonal} \nonumber\\
&\leq \sqrt{n}&&\text{As $\sigma_i=1-\frac{\lambda_i}{2}\leq 1$}
\end{align}
Similarly we have
\begin{equation}
\label{eq:bnorm}
||b||_2\leq   \sqrt{n}. 
\end{equation}
Thus by \eqref{eq:upperbd-apx-dp} we get
	\begin{align*}
		\|a\|_2\|b-b'\|_2 + \|a-a'\|_2\|b'\|_2 &\leq  \sqrt{k\cdot n}\cdot \beta\cdot (w_A+w_B) &&\text{By \eqref{eq:aa;}, \eqref{eq:bb'}, \eqref{eq:anorm}, \eqref{eq:bnorm}}  \\
		&\leq \xi\cdot n &&\text{As $w_A+w_B\leq   \frac{2\cdot n\cdot n^{20 \e/\varphi^2}}{\sqrt{\frac{n}{k}}}$ and $\beta = \frac{\xi}{2\cdot k \cdot n^{20 \e/\varphi^2}}$}
	\end{align*}
Therefore, with probability at least $1 - 2\cdot k \cdot n^{-200k} \geq 1 - n^{-100k}$ we have
\[
\left| \mathds{1}_A^T \cdot (U_{[k]}\Sigma_{[k]}U_{[k]}^T) \ 
		\mathds{1}_B \ 
		- \frac{|A|\cdot |B|}{|S_A|\cdot |S_B|} \cdot \mathds{1}_{S_A}^T (U_{[k]}\Sigma_{[k]}U_{[k]}^T) \mathds{1}_{S_B} \right| \leq \xi\cdot n. 
\]
\end{proof}

\begin{fact} \label{fact:hoeffding} (\textbf{Hoeffding Bounds})
	Let $Z_1, Z_2, \cdots Z_n$ be iid random variables with $Z_i \in [a,b]$ for all $i \in [n]$
	where $-\infty \leq a \leq b \leq \infty$. Then
	\begin{align*}
		&\text{Pr} \left[\frac{1}{n} \cdot \sum |(Z_i - \E[Z_i])| \geq t \right] \leq \exp(-2nt^2/(b-a)^2), \text{ and} \\
		&\text{Pr} \left[\frac{1}{n} \cdot \sum |(Z_i - \E[Z_i])| \leq t \right] \leq \exp(-2nt^2/(b-a)^2). 
	\end{align*}
\end{fact}
 
Next, we prove Lemma \ref{lem:semi-supervise-model}.
\lemsemisupevise*
\begin{proof}
Let $s=|S|$. For $x\in V$, and $r \in [s]$, let $Y_x^r$ be a random variable which is $1$ if the $r$-th sampled vertex is $v$, and $0$ otherwise. Thus $\E[Y_x^r]=\frac{1}{n}$. Observe that $|S_i|=|S\cap C_i|$ is a random variable defined as $\sum_{r=1}^s \sum_{x\in C_i} Y_x^r$ where its expectation is given by
\[\E[|S\cap C_i|]=\sum_{r=1}^s \sum_{x\in C_i} Y_x^r \geq s\cdot \frac{|C_i|}{n}\geq s\cdot \frac{\Omega(1)}{k} \text{,}\]
where, the last inequality holds since all clusters have size $\Omega(n/k)$, since we assume $\vol(C_i)/\vol(C_j) = O(1)$ for all $i,j$ and the graph is $d$-regular.

Notice that the random variables $ Y_x^r$ are negatively associated, since for each $r$, $\sum_{x\in V} Y_x^r=1$. Therefore, by Chernoff bound,
\[
\text{Pr} \left[\left| |S\cap C_i|- \frac{|C_i|}{n}\right| > \delta \cdot s \cdot \frac{|C_i|}{n}\right] \leq 2\cdot \text{exp} \left(- \frac{\delta^2}{3}\cdot \frac{s}{k} \right) \leq n^{-120\cdot k} \text{,}
\]
where, the last inequality holds by choice of $s\geq \frac{400\cdot \log n\cdot k^2}{\delta^2}$. Therefore, by union bound,
\[
\text{Pr} \left[\exists i: \left| |S\cap C_i|- \frac{|C_i|}{n}\right| > \delta \cdot s\cdot \frac{|C_i|}{n}\right] \leq 2\cdot  k\cdot n^{-120\cdot k}\leq n^{-100\cdot k}
\text{.}\]
\end{proof}

Finally, we prove Claim \ref{clm:cluster:resp:lb}. 
\clmclusterresp*
\begin{proof}
	We present such an instance $G$ explicitly. We pick a large enough integer $m$.
	Let $X_1 = \{1, 2, \cdots, m - 2 \e m\}$ denote the set of first $m - 2 \e m$ integers
	and use standard constructions to obtain an $d-1$ regular $\varphi \cdot d$-expander on vertices
	in $X_1$. Also, let $Y_1 = \{m - 2 \e m + 1, m -2\e m + 2, \cdots m\}$ denote
another set of $2 \e m$ integers and obtain another $d-1$ regular $\varphi \cdot d$ expander on 
	$Y_1$. Put a matching of size $\e \varphi dm$ between the sets $X_1$ and $Y_1$. Also, put 
	a matching on remaining degree $d-1$ vertices in $X_1$. Notice that $|E(X_1,Y_1)| = \e \varphi dm$.
	Let $C_1 = X_1 \cup Y_1$. Now we describe another set of vertices. 
	This time we consider three sets: $X_2 = \{1, 2, \cdots, m - 4 \e m\}$, 
	$Y_2 = \{m - 4 \e m + 1, \cdots m -2 \e m\}$ and $Z_2 = \{m - 2 \e m  + 1, \cdots, m\}$.
	We again obtain a $d-1$ regular $\varphi \cdot d$ expander on all of these sets. Next, add a matching
	of size $\e \varphi dm$ between $X_2$ and $Y_2$ and between $X_2$ and $Z_2$.
	We add a matching between remaining degree $d-1$ vertices in $X_2$ and
	another matching between remanining degree $d-1$ vertices in $Z_2$. 
	Notice that $|E(X_2, Y_2)| = \e \varphi dm = |E(X_2, Z_2)|$. Next, let
	$C_2 = X_2 \cup Y_2 \cup Z_2$.

	Finally, we add a matching between the remaining degree $d-1$ vertices in $Y_1$ and $Y_2$.
	Overall this gives a $d$-regular graph on $2m$ vertices. We let $B = C_1$ and 
	thus $\overline{B} = C_2$. Notice that 
	$\varphi_{\text{out}}(C_1) = \varphi_{\text{out}}(C_2) = \e$. Also, by construction,
	note that $\varphi_{\text{in}}(C_1) = \varphi_{\text{in}}(C_2) \geq \varphi$.
	Now consider the following set $S = X_1 \cup Z_2$. We see that 
	$$|E(S, \overline{S})| = |E(X_1, Y_1)| + |E(X_2, Z_2)| = 2 \e \varphi dm.$$
	Also $|S| = |X_1| + |Z_2| = m$.  And therefore, it holds
	that $\phi(G) \leq \phi(S) = 2 \e \varphi$.
\end{proof}

\section{Sublinear estimator for cost of expanders}
\label{sec:exp-cost}
In this section, we formally prove Theorem \ref{thm:singleclustercost} from Section \ref{sec:expandercost-overview}, which demonstrates an algorithm for estimating the Dasgupta cost of a $\varphi$-expander up to a $\poly(1/\varphi)$ factor using $\approx n^{1/3}$ seed queries. 

Then, we prove Theorem \ref{thm:Gk} which demonstrates an algorithm for estimating the total contribution of the clusters to the Dasgupta cost of a $d$-regular graph that admits $(k,\varphi, \epsilon)$-clustering. 
\singleclustercost*

\totalclustercost*

For completeness, we restate Theorem \ref{thm:MSun21:thm3} from \cite{MSun21} and the algorithm for computing $\mathcal{T}_{\deg}$ from \cite{MSun21}. However, we don't explicitly construct $\mathcal{T}_{\deg}$.
\mstedg*
\begin{algorithm}[H]
	\caption*{\textbf{Algorithm 2} \textsc{HCwithDegrees($G\{V\}$)} \cite{MSun21}}
	\begin{algorithmic}[1]
		\STATE \textbf{Input}: $G=(V, E, w)$ with the ordered vertices such that $d_{v_1}\geq \ldots \geq d_{v_{|V|}}$  
		\STATE \textbf{Output}: An HC tree $\mathcal{T}_{\deg}(G)$
		\IF{$|V|=1$}
		\STATE \textbf{return} the single vertex $V$ as the tree
		\ELSE
		\STATE $i_{\max} \coloneqq \lfloor \log_2{|V|-1} \rfloor$; $r \coloneqq 2^{i_{\max}}$; $A \coloneqq \{v_1,\ldots,v_r\}$; $B \coloneqq V\textbackslash A$
		\STATE Let $\mathcal{T}_1 \coloneqq $ \textsc{HCwithDegrees($G\{A\}$)}; $\mathcal{T}_2 \coloneqq $ \textsc{HCwithDegrees($G\{B\}$)}
		\STATE \textbf{return} $\mathcal{T}_{\deg}$ with $\mathcal{T}_1$ and $\mathcal{T}_2$ as the two children
		\ENDIF
	\end{algorithmic}
\end{algorithm}

The rest of the section is structured as follows. In Section \ref{sec:exp-irr-opt} we first prove that for every $\varphi$-expander $G$, the quantity $\sum_{x \in V} \rank(x) \deg(x) $ approximates the cost of $\mathcal{T}_{\deg}$ up to $O(1/\varphi)$ factor. In Section \ref{sec:rankdeg}, we show how to estimate the quantity $\sum_{x \in V} \rank(x) \deg(x) $. Then, in Section \ref{sec:clustercost}, we put everything together and complete the proofs of Theorem \ref{thm:singleclustercost} and Theorem \ref{thm:Gk}. Finally, in Section \ref{sec:tightrunningtime} we prove the optimality of our sampling complexity for a single expander.

\subsection{Bound cost of an expander by $\sum \rank(x)\cdot \deg(x)$}\label{sec:exp-irr-opt} 
Let $G = (V,E)$ be an arbitrary expander with vertices $x_1, x_2, \ldots x_n$ ordered such that $d_1 \geq d_2 \geq \ldots \geq d_n$, where $d_i = \text{deg}(x_i)$. We denote
by $\tdeg$ the Dasgupta Tree returned by Algorithm 1 of \cite{MSun21}. Recall that
this is a binary tree which is obtained by recursive applications of a merge procedure.
The call at the root level to merge aggregates a left subtree with leaves 
$v_1, v_2, \cdots v_{n/2}$ and a right subtree which has remaining vertices as leaves.
We would like to show the following two lemmas.

\begin{lemma} \label{lem:dcost:lb}
	Let $G = (V,E)$ be a graph with degree sequence $d_1 \geq d_2 \geq \ldots \geq d_n$
	and expansion $\varphi$. We have
	$$\CT_G(\tdeg) \geq \Omega(\varphi) \sum_{i=1}^{n} i \cdot d_i \text{.}$$
\end{lemma}

\begin{lemma} \label{lem:dcost:ub}
	Let $G = (V,E)$ be a graph with degree sequence $d_1 \geq d_2 \geq \ldots \geq d_n$. We have
	$$\CT_G(\tdeg) \leq 2 \sum_{i=1}^n i \cdot d_i.$$
\end{lemma}

Note that Lemma~\ref{lem:dcost:ub} does not require the graph to be an expander. Both Lemma \ref{lem:dcost:lb} and Lemma \ref{lem:dcost:ub} hold even if the graph has self-loops.  

\subsubsection{Lower bound on Dasgupta cost of an expander (Proof of Lemma~\ref{lem:dcost:lb})}

In this section, we prove Lemma~\ref{lem:dcost:lb}. We will need some notation. Order the vertices
in decreasing order of degrees and let $H = \lfloor \log_2 n \rfloor$. 
For each $i \in \{0,1, \ldots H\}$, define the $i$-th
bucket as $$L_i = \{j \in V : 2^i \leq j \leq 2^{i+1} - 1 \}.$$ We also need another notation.
Define $L_{\leq i} = \bigcup_{j \leq i} L_j$. We will prove the following two claims.

\begin{claim} \label{clm:lb:aux}
	$\CT_G(\tdeg) \geq \frac{1}{2} \cdot \sum_{i = 0}^{H} |E(L_i, \overline{L_i})| \cdot |L_{\leq i}|$.
\end{claim}

\begin{claim} \label{clm:lb:final}
	$\sum_{i = 0}^{H} |E(L_i, \overline{L_i})| \cdot |L_{\leq i}| \geq \varphi \cdot \sum i \cdot  d_i$.
\end{claim}

Note that once these two claims are shown, Lemma~\ref{lem:dcost:lb} follows as a corollary.

\begin{proof} (Of Claim~\ref{clm:lb:aux})
	Recall that the vertices of $\tdeg$ are arranged in decreasing order. 
	Also, recall $\CT_G(T) = \sum_{\{x,y\}\in E} |\lv(\tdeg[\lca(x, y)])|\text{.}$
	
	We will lower bound $\CT_G(\tdeg)$ by considering contributions to
	Dasupta Objective from a subset of the edges. In particular, we sum 
	only over edges between ``prefix sets'' in $\tdeg$ to get
	
	\begin{equation} \label{eq:ct:lb}
		\CT_G(\tdeg) \geq \sum_{i=1}^{H} |E(L_i, L_{\leq i-1})| \cdot |L_{\leq i}|. 
	\end{equation}
	
	The above expression peels off sets $L_i$ one at a time and considers
	the contribution of edges in the set $E(L_i, L_{\leq i-1})$ which is at least 
	$|L_{\leq i-1}| + |L_i| = |L_{\leq i}|$. We will show that this is at least 
	half the target expression (i.e., half the right hand side in the claim above) 
	$ \sum_{i=0}^{H} |E(L_i, \overline{L_i})| \cdot |L_{\leq i}|$ which will finish the proof. 
	
	The two expressions differ in the contribution they charge to an edge.
	Fix some $i$ and take an edge $e \in E(L_i, \overline{L_i})$.
	Denote the contribution of edge $e$ in Equation~\eqref{eq:ct:lb}
	as $\contrib(e)$ and denote the contribution of $e$ to the target expression 
	as $\target(e)$. It suffices to show that for every 
	$e \in \bigcup_{i=0}^{H} E(L_i, \overline{L_i})$, $\target(e) \leq 2 \contrib(e)$.

	Fix $0 \leq i \leq H$ and take an edge $e \in E(L_i, \overline{L_i})$.  Let $j \neq i $ be the index such that $e \in E(L_i, L_j)$ We have $\contrib(e) \geq \max\{ |L_{\leq i}|, |L_{\leq j}|\}$ and 
	$\target(e) =  |L_{\leq j}| + |L_{\leq i}| \leq 2  \max\{ |L_{\leq i}|, |L_{\leq j}|\}.$ This holds for every edge
	$e \in \bigcup_{i=0}^{H} E(L_i, \overline{L_i})$ and this finishes the proof.
\end{proof}

Next, we prove Claim~\ref{clm:lb:final}.

\begin{proof} (Of Claim~\ref{clm:lb:final})
Since the expansion of $G$ is at least $\varphi$, we obtain
	\begin{align*}
		\sum_{i=0}^{H} |E(L_i, \overline{L_i})| \cdot |L_{\leq i}| &\geq \sum_{i=0}^{H} \varphi \cdot \vol(L_i) \cdot |L_{\leq i}| \\
		&=  \varphi \sum_{i=0}^{H} \left( \sum_{j \in L_i} d_j \right) |L_{\leq i}|.
	\end{align*}
Furthermore, note that for all $j \in L_i$, it holds that $j \leq |L_{\leq i}|$. Therefore, from the above, we get that
\begin{align*}
    \sum_{i=0}^{H} |E(L_i, \overline{L_i})| \cdot |L_{\leq i}| &\geq \varphi \sum_{i=0}^{H} \left( \sum_{j \in L_i} d_j \right) |L_{\leq i}|  \\
    & \geq \varphi \sum_{i=0}^H \sum_{j \in L_i}d_j \cdot j \\
    & = \varphi \sum_{j \in [n]} d_j \cdot j,
\end{align*}
	which finishes the proof.
\end{proof}

\subsubsection{Upper bound on Dasgupta cost of an expander (Proof of Lemma~\ref{lem:dcost:ub})}

In this section, we prove Lemma~\ref{lem:dcost:ub}. Like the previous section,
we do this by proving the following two claims.

\begin{claim} \label{clm:ub:aux}
	$\CT_G(\tdeg) \leq \sum_{i = 0}^{H} \vol(L_i) \cdot |L_{\leq i}|$.
\end{claim}

\begin{claim} \label{clm:ub:final}
	$\sum_{i = 0}^{H} \vol(L_i) \cdot |L_{\leq i}| \leq 2 \sum i \cdot d_i$.
\end{claim}

Lemma~\ref{lem:dcost:ub} follows as a corollary.

\begin{proof} (Of Lemma~\ref{lem:dcost:ub})
	Immediate from Claim~\ref{clm:ub:aux} and Claim~\ref{clm:ub:final}. 
\end{proof}

Now, we will prove Claim~\ref{clm:ub:aux} and Claim~\ref{clm:ub:final} in the rest of this
section. We begin with the first claim.

\begin{proof} (Of Claim~\ref{clm:ub:aux})
	We want to show 
	$\CT_G(\tdeg) \leq \cdot \sum_{i = 0}^{H} \vol(L_i) \cdot |L_{\leq i}|$.
	Write $Obj(e)$ to denote the contribution to Dasgupta Cost of the 
	edge $e$ in tree $\tdeg$. Recall, for an edge $e$ which is not a self-loop,  $Obj(e)$ equals the number of leaves in the
	subtree rooted at the LCA of the endpoints of the edge. 
    Denote by $\target(e)$
	the contribution of edge $e$ to the objective in the right hand side of 
	this expression. If $e = (u,u)$ is a self-loop, then $e$ does not contribute to the Dasgupta cost, and therefore $Obj(e) = 0 \leq \target(e)$.
    So consider an edge $e = (u,v)$ such that $u \neq v$ with $u \in L_i$, $v \in L_j$
	where $i \leq j$. The edge $e$ is considered in the above sum at indices $i$ and
	$j$. The contribution of $e$ to the target objective is given
	as $$\target(e) = |L_{\leq i}| + |L_{\leq j}|.$$ On the other hand, from Algrorithm \ref{alg:HCdeg} and by 
	definition of $L_i$'s, we have $Obj(e) \leq |L_{\leq i}|+ |L_{\leq j}| = \target(e)$. This
	holds for every edge and therefore 
	$$\CT_G(\tdeg) = \sum_{e \in E(G)} Obj(e) \leq \sum_{e \in E(G)} \target(e) = \sum_{i=0}^{H} \vol(L_i) \cdot |L_{\leq i}|.$$
\end{proof}

Finally, we prove Claim~\ref{clm:ub:final} to wrap up. 

\begin{proof} (Of Claim~\ref{clm:ub:final})
	Note that $|L_{\leq i}| \leq 2^{i+1}$. We have,
	\begin{align*}
		\sum_{i=0}^{H} \vol(L_i) \cdot |L_{\leq i}| &\leq \sum_{i=0}^{H} \left( \sum_{j \in L_i} d_j \right) \cdot |L_{\leq i}| \\
		&\leq \sum_{i=0}^{H} \left( \sum_{j \in L_i} d_j \right) \cdot 2^{i+1}. 
	\end{align*}
	
	We want to upper-bound the last expression above. Note that this expression is
	of the form $\sum_{j \in [n]} \alpha_j \cdot d_j$ where $\alpha_j = 2^{i+1}$ if $j \in L_i$.
	However, for any $j \in L_i$, note that $j \geq  2^{i} = \frac{1}{2} \alpha_j$. This means that
	$$\sum_{j \in [n]} \alpha_j\cdot  d_j \leq 2 \sum_{j \in [n]} j\cdot  d_j$$ as desired.
\end{proof}

\subsection{Estimating $\sum \rank(x)\cdot \deg(x)$}\label{sec:rankdeg}
In this section, we prove Lemma ~\ref{lemma:est_i_di}, which asserts that we can estimate $\sum \rank(x)\cdot \deg(x)$ using $O^*\left(n^{1/3}\right)$ samples. 
\begin{restatable}{lemma}{estidi}
	\label{lemma:est_i_di}
Let $G = (V,E)$ be a $\varphi$-expander (possibly with self-loops).
 There exists an estimator $v$ using $O^*\left(n^{1/3}\right)$  samples, such that with probability at least $1-n^{-100}$, 
 \[ \sum_{x \in V} \rank(x) \deg(x) \leq v \leq O(1)\cdot \sum_{x \in V} \rank(x) \deg(x). \]
\end{restatable}
Partition the vertices into buckets as follows: For $d = 2^0, 2^1, \cdots, 2^{\log(n /\varphi)}$, let $B_d := \{x \in V : d \leq \deg(x) <2d\}.$  We will refer to $B_d$ as the \emph{degree class} of $d$. Let $n_d:=|B_d|$ denote the size of the degree class, and let $r_d$ denote the highest rank in $B_d$. Note that $r_d$ is the number of vertices in $G$ that have degree at least $d$, so we have $r_d= \sum_{t \geq d} n_t$. 

Sometimes we will write $B_{\geq d}$ and $B_{<d}$ to denote $\cup_{t \geq d} B_d$ and $\cup_{t < d} B_d$, respectively.

Note that there are at most $\log(n /\varphi)$ different degree classes, since each vertex can have at most $n(1/\varphi -1)$ self-loops. 

The vertices in $B_d$ have ranks $r_d, r_{d}-1,\dots, r_d -n_d+1$ and degrees in $[d,2d]$, which gives the bounds

\begin{equation}\label{eqn:rdnd}
    \frac{d}{2} \cdot n_d \cdot r_d \leq \sum_{i = r_d -n_d+1}^{r_d} i \cdot d \leq \sum_{x \in B_d} \rank(x) \cdot \deg(x) \leq \sum_{i = r_d -n_d+1}^{r_d} i \cdot 2d \leq 2 d \cdot n_d \cdot r_d.
\end{equation}
Thus, our goal will be to efficiently approximate the quantities $r_d \cdot n_d$.  

We start by proving the following technical lemma, which shows that there exists a degree class $B_d$ that contains a large fraction of the degree mass, 
and that satisfies $d \leq  O\left(\frac{\log(n/\varphi)}{\varphi^2}\right) \cdot n_d$. 
\begin{lemma}\label{lemma:heavy_and_good}
    There exists a degree class $d$ such that $n_d \cdot d \geq \frac{m \cdot \varphi}{4\log(n/\varphi)}$ and  $d \leq \frac{16}{\varphi^2} \log (n/\varphi) \cdot n_d.$
\end{lemma} 
\begin{proof}
    Let $m' \geq \varphi\cdot  m$ denote the number of non-self-loop edges, and $\deg'(\cdot)$ denote the degrees discounting self-loops.
    Orient the edges from high degree to low degree (break ties arbitrarily within any degree class). 
     That way, we have $m' = \sum_{x \in V} \deg'_{in}(x)$. Say that a degree class $B_d$ is \emph{heavy} if $ n_d \cdot d \geq \frac{m'}{4\log(n /\varphi)}$, and say that it is \emph{light} otherwise. Moreover, call a degree class $B_d$ \emph{good} if $\sum_{x \in B_d} \deg'_{in}(x) \geq \frac{n_d \cdot d \cdot \varphi }{2}$, and {\em bad} otherwise. There must exist a good heavy class, since otherwise
    \begin{align*}
        m' &= \sum_{B_d: B_d\text{ is light}} \sum_{x \in B_d} \deg'_{in}(x) + \sum_{B_d: B_d\text{ is heavy}} \sum_{x \in B_d} \deg'_{in}(x) &  \\
        & \leq \sum_{B_d: B_d\text{ is light}} 2d \cdot n_d +  \sum_{B_d: B_d\text{ is heavy}} \sum_{x \in B_d} \deg'_{in}(x) \\
      & < \log(n /\varphi) \frac{m'}{2\log(n /\varphi)} +  \sum_{B_d: B_d\text{ is heavy}} \sum_{x \in B_d} \deg'_{in}(x), &\text{by definition of the light classes}\\
        & < \frac{m'}{2} + \sum_{B_d: B_d\text{ is heavy}}\frac{n_d \cdot d  \cdot \varphi }{2}, & \text{assuming that all heavy classes are bad} \\
         & < \frac{m'}{2} + \frac{m \cdot \varphi}{2}, & \\
        & \leq  m', & \text{since $m \cdot \varphi \leq m' $ in a $\varphi$-expander}
    \end{align*}
    which is a contradiction. Thus, there exists a degree class $d$ that is both heavy and good, i.e.
\begin{equation*}
\begin{split}
n_d \cdot d \geq \frac{m'}{4\log(n /\varphi)} &\geq \frac{m\cdot \varphi }{4\log(n /\varphi) } \qquad \qquad \text{(heavy)}\\
&\text{and}\\
\sum_{x \in B_d} \deg'_{in}(x) &\geq \frac{n_d \cdot d \cdot \varphi}{2 } \qquad \qquad \text{(good)}.
\end{split}
\end{equation*}
Let $d$ be such a degree class. To establish the lemma, it remains to show that $d \leq \frac{16}{\varphi^2} \log (n/\varphi) \cdot n_d$.

  First, observe that $r_d \cdot d \leq \sum_{x : \deg(x) \geq d} \deg(x) \leq 2m \leq \frac{8 \log(n /\varphi)  \cdot n_d \cdot d}{\varphi}$, where the last inequality follows from the assumption that $d$ is heavy. Therefore, we have \[r_d \leq \frac{8 \log(n /\varphi) \cdot  n_d}{\varphi}.\]

 Now, consider the number of non-self-loop edges between $B_d$ and $B_{\geq d}$. Recall that we orient the edges from high degree to low degree, so that the number of non-self-loop edges between $B_d$ and $B_{\geq d}$ is equal to $\sum_{x \in B_d} \deg'_{in}(x) \geq \frac{n_d \cdot d \cdot \varphi}{2 }.$ On the other hand, the number of non-self-loop edges between $B_d$ and $B_{\geq d}$ can be at most $|B_{\geq d}| \cdot |B_d| = \left(\sum_{t \geq d}n_t \right)\cdot n_d = r_d \cdot n_d \leq  \frac{8\log(n /\varphi)\cdot n_d^2}{\varphi}  $. 
    Combining, we obtain  
\[  \frac{n_d \cdot d \cdot \varphi }{2 } \leq \frac{8\log(n /\varphi)\cdot n_d^2}{\varphi} , \]
which gives $d \leq \frac{16}{\varphi^2} \log(n/\varphi) \cdot n_d$, as required. 
\end{proof}
We now introduce the definition of a Dasgupta cost heavy degree class, i.e. a class that contributes a significant fraction of the Dasgupta cost. 
\begin{definition}
    Say that a degree class $d$ is \emph{$\alpha$-Dasgupta Cost Heavy}, or just \emph{$\alpha$-DC-heavy}, if 
    \begin{equation*} \sum_{x\in B_d} \rank(x) \deg(x) \geq 2 \alpha \sum_{x\in V} \rank(x) \deg(x).\end{equation*} 
\end{definition}
The following claim shows that for $\alpha$-DC-heavy classes, we can use $n_d$ as a proxy for $r_d$. 
\begin{claim}\label{claim:rankbound}
    If $d$ is $\alpha$-DC-heavy, then $n_d \geq\frac{1}{2} \alpha \cdot r_d.$ 
\end{claim}
\begin{proof}
By Equation \eqref{eqn:rdnd}, we have
\[2d \cdot  n_d \cdot r_d  \geq \sum_{x \in B_d} \rank(x) \deg(x) \geq 2\alpha \sum_{x \in V} \rank(x) \deg(x) \geq 2\alpha \sum_{i = 1}^{r_d} i \cdot d \geq 2\alpha\cdot \frac{r_d^2}{2}\cdot d,\]
which rearranges to 
\[n_d \geq \frac{1}{2}\alpha \cdot r_d.\]
\end{proof}

The next lemma is the key result underlying our bound on the number of samples required. It shows that a degree class that contributes a nontrivial amount to Dasgupta cost of the graph must contain at least a $\approx n^{-1/3}$ fraction of edges of the graph:
\begin{lemma}\label{lemma:heav_deg_class}
    If a degree class $t$ is $\alpha$-DC-heavy, then  $n_t \cdot t  \cdot n^{1/3}\cdot \frac{\log^2 n}{\alpha^2 \varphi^2}  \geq \Omega(m). $
\end{lemma}

\begin{proof}
We will apply Lemma \ref{lemma:heavy_and_good}, which asserts that there exists a degree class $B_d$ that contains a large fraction of the degree mass, 
and that satisfies $d \leq  O\left(\frac{\log^2(n/\varphi)}{\varphi}\right) \cdot n_d$. We will then use the degree class $B_d$ as a reference, and show that the degree mass of $B_t$ cannot be much smaller. 

More formally,we have the following optimization problem over the variables $t,d,n_t, n_d$: 
\begin{align*}
      \max_{t,d, n_t, n_d} \frac{n_d \cdot d}{n_t \cdot t}  & \quad  \\
      \text{ such that}&\\
       n_t^2 \cdot t & \geq \alpha^2
       \cdot n_d^2\cdot  d \\
      d & \leq  \frac{16 \log (n/\varphi)}{\varphi^2} \cdot n_d \\
     &   n_t, n_d \leq n  \\
    &   n_t,d, t \geq 1  \\
      &   n_d\geq 0 . 
\end{align*}
First, we will show that the optimal value is $\approx n^{1/3}$. 
\begin{claim}\label{claim:optimization}
    The above optimization problem has a finite optimal value. Furthermore, if $t,d,n_d, n_t$ is an optimal solution, then the first two constraints are tight and $t=1$. 
\end{claim}
\begin{proof}
First, we show that the optimization problem has a finite optimal value. Observe that adding the constraint $t \leq n^4$ does not change the optimal objective value, since increasing the value of $t$ can only harm the objective, and any feasible choice of $n_d, n_t, d$ remains feasible after adding the constraint. Now, with the additional constraint, we have that the feasible region is bounded (since $1 \leq n_t \leq n, 1 \leq t \leq n^4 $, $0 \leq n_d \leq n$, and $1 \leq d \leq  \frac{16 \log(n/\varphi)}{\varphi^2} \cdot n_d$), closed, and non-empty (since taking for instance $d =1, n_d = 1$, $t = n_t = n$ is a feasible solution). So the optimization problem has a finite value and attains its maximum value. 

Next, we show that if $t,d,n_d, n_t$ is an optimal solution, then the first two constraints are tight and $t=1$. Suppose that the first constraint is loose, i.e. that $n_t^2 \cdot t > \alpha^2 \cdot n_d^2 \cdot  d$. Clearly, we can't have $n_t = t=1$ and $n_d = n$ (otherwise the first constraint would not be satisfied), so it is possible to either decrease $t$, decrease $n_t$ or increase $n_d$. Either of these options gives a higher objective value, which contradicts the optimality of the given solution.

Suppose instead that the second constraint is loose. Let $d' = \gamma^2 \cdot d, n_d' = \frac{n_d}{\gamma}$ for some sufficiently small $\gamma >1$. Then $t,n_t, d', n_d'$ is a feasible solution, but $n_d' \cdot d' = \gamma \cdot n_d \cdot d > n_d \cdot d$, so this gives a higher objective value, which is a contradiction. 

Finally, suppose that $t >1$, and that the first two constraints are tight. If $n_t < n$, then let $t' =1$, $n_t' = n_t \sqrt{t}$. Then $t', n_t', d, n_d$ is a feasible solution, but $n_t' \cdot t' =  n_t \sqrt{t} < n_t \cdot t$, so this gives a higher objective value, contradiction. If instead $n_t = n$, then we have 
$n_t^2 \cdot t > n^2$ (by the assumption that $t >1)$. On the other hand, we have $n_t^2 \cdot t = \alpha^2 \cdot n_d^2 \cdot d = \frac{16 \log(n/\varphi^2)}{\varphi^2}n_d^3$ (by the assumption that the first two constraints are tight), from which we deduce $n_d, d >1$. In particular, there exists $\gamma >1$ such that $t' = t \cdot \gamma^{-1}$, $n_d' =  n_d \cdot \gamma^{-1/3} , $ $d'  =  d\cdot \gamma^{-1/3}$ is a feasible solution. But this solution has objective value $\frac{n_d \cdot d}{n_t\cdot  t}\cdot \gamma^{1/3}>\frac{n_d \cdot d}{n_t\cdot  t},$ contradiction.  
\end{proof}
Now let $t,d,n_d, n_t$ be an optimal solution. It follows by Claim \ref{claim:optimization} that
\begin{align*}
     n^2  &\geq n_t^2\cdot t \qquad \qquad \text{since $n_t \leq n$ and $t = 1$} \\
     & = \alpha^2\cdot  n_d^2 \cdot d \qquad \text{since the first constraint is tight} \\
     &  = \frac{\alpha^2 \cdot \varphi^4}{16^2\log^2(n/\varphi)}\cdot d^3 \qquad \text{since the second constraint is tight,}
\end{align*}
which rearranges to 
\[d \leq \left( \frac{O(n \log(n/\varphi))}{\varphi^2 \cdot  \alpha} \right)^{2/3}\]
Furthermore, since $t=1$, we have
\[n_t \cdot t = n_t = \sqrt{n_t^2 \cdot t} = \alpha\cdot \sqrt{ n_d^2 \cdot d}.\]
Thus, 
\[ \frac{n_d \cdot d}{n_t \cdot t}  = \frac{n_d \cdot d}{ \alpha \cdot n_d \cdot d^{1/2}} = \frac{d^{1/2}}{\alpha} \leq \frac{1}{\alpha} \left( \frac{O(n \log(n/\varphi))}{ \varphi^2\cdot  \alpha}\right)^{1/3}  \]
This shows that the optimal solution to the optimization problem has value $O(n^{1/3} \varphi^{-2/3}\alpha^{-4/3} \log^{1/3}(n/\varphi)).$

Now let $B_d$ be a degree class such that $n_d \cdot d \geq \frac{m\cdot \varphi}{4\log n}$ and  $d \leq \frac{16}{\varphi^2} \log(n/\varphi) \cdot  n_d$ (exists by Lemma \ref{lemma:heavy_and_good}), and let $t$ be any $\alpha$-DC-heavy degree class. 
By Claim \ref{claim:rankbound}, we have 
\begin{align*}
    n_t^2 \cdot t &\geq \frac{1}{2} \cdot \alpha \cdot r_t \cdot n_t \cdot t \qquad \qquad \qquad \qquad \text{by Claim \ref{claim:rankbound},} \\
    & \geq \alpha^2 \sum_{x \in V} \deg(x) \rank(x) \qquad \text{since $t$ is $\alpha$-DC heavy} \\
    &  \geq  \alpha^2 \sum_{x \in B_d} \deg(x) \rank(x) \\
    & \geq \alpha^2 \cdot n_d \cdot r_d \cdot d \\
    & \geq \alpha^2 \cdot n_d^2 \cdot d
\end{align*}

Thus, $t, n_t, d, n_d $ is a feasible solution to the optimization problem, and in particular 
 \[n_t \cdot t \cdot O(n^{1/3} \varphi^{-2/3}\alpha^{-4/3} \log^{1/3}(n/\varphi))\geq n_d \cdot d \geq \frac{m \cdot \varphi}{4 \log (n/\varphi)}, \] which gives the result. 
\end{proof}
We can now obtain a good estimator for the size of each bucket. 
\begin{lemma}\label{lemma:nhat}
 Given $\alpha$, there exists an estimator $\hat{n}_t$ using $O\left(n^{1/3} \cdot \frac{\log^3 (n/\varphi)}{\alpha^2 \varphi^2}\right)$ samples, such that with probability at least $1-\frac{1}{2}n^{-101}$, the following holds: 
\begin{enumerate}
    \item For every degree class $t$, it holds that $\hat{n}_t \leq 6n_t$
    \item If $t$ is an $\alpha$-DC-heavy class, then $\hat{n}_t \geq  \frac{1}{2}n_t$. 
 \end{enumerate}
\end{lemma}
\begin{proof}
Let $c$ be the constant in front of $m$ in Lemma \ref{lemma:heav_deg_class}, and let $s = 16c \cdot n^{1/3} \cdot \frac{\log^3 (n/\varphi)}{\alpha^2 \varphi^2}$. Let $S$ be a set of $s$ vertices sampled independently at random with probability proportional to their degree. For each degree class $t$, let $X_t = \frac{2m}{s \cdot t} |\{x \in S: x \in B_t\}|$. Then $\mathbb{E}[X_t] = \frac{1}{t} \sum_{x \in B_t} \deg(x) \leq \frac{2t\cdot n_t}{t} = 2n_t.$ By Markov's inequality, 
$\Pr[X_t > 6n_t] \leq \frac{1}{3}$. Repeat $O(\log n)$ times and let $\hat{n}_t$ be the median, so that $\Pr[\hat{n}_t \geq 6n_t] \leq \frac{1}{4}n^{-102}.$ 

Now, suppose that $t$ is $\alpha$-DC heavy. We have that $\frac{t \cdot s}{2m} \cdot X_t$ is a sum of independent $\{0,1\}$ random variables, with $\mathbb{E}[\frac{t \cdot s}{2m}X_t] = \frac{s}{2m}\sum_{x \in B_t} \deg(x) \geq \frac{s}{2m}  t\cdot n_t \geq 16$. Here the last inequality holds by Lemma \ref{lemma:heav_deg_class} and the choice of $s$. By Chernoff bounds, we obtain that $\Pr[\frac{t \cdot s}{2m}  X_t \leq \frac{1}{2} \cdot \frac{t \cdot s}{2m}\cdot n_t] \leq \exp{(-2)} \leq \frac{1}{3}$. Since $\hat{n}_t$ is obtained from $X_t$ by repeating $O(\log n)$ times and taking the median, we get that $\Pr[\hat{n}_t < \frac{1}{2}n_t] \leq \frac{1}{4}n^{-102}.$ 

Taking the union bound over all $t$ gives the result. 
\end{proof}
Similarly, we can obtain an estimator for the highest rank in each bucket. 
\begin{lemma}\label{lemma:rhat}
 Given $\alpha$, there exists an estimator $\hat{r}_t$ using  $O\left(n^{1/3} \cdot \frac{\log^5 (n/\varphi)}{\alpha^2 \varphi^2}\right)$  samples, such that with probability at least $1-\frac{1}{2}n^{-101}$, the following holds: 
\begin{enumerate}
    \item For every degree class $t$, it holds that $\hat{r}_t \leq 6r_t$
    \item If $t$ is an $\alpha$-DC-heavy class, then $\hat{r}_t \geq  \frac{1}{4}r_t$. 
 \end{enumerate}
\end{lemma}
\begin{proof}
Let $\alpha' = \frac{\alpha}{8\log^2 (n/\varphi)}$, and let $\hat{n}$ be the estimator from Lemma \ref{lemma:nhat} with parameter $\alpha'$. For each degree class $d$, let $\hat{r}_d = \sum_{t \geq d} \hat{n}_t$. Condition on the success of $\hat{n}$ (which happens with probability at least $1-\frac{1}{2}n^{-101}$). Then Property 1 follow immediately from Lemma \ref{lemma:nhat}. It remains to prove that Property 2 holds. Fix an $\alpha$-DC-heavy class $t$. Say that a degree class $t' > t$  is \emph{heavy} if 
\[n_{t'} \geq \frac{r_t}{2 \log (n/\varphi)},\] and otherwise say that it is \emph{light}.  First, we show that if $t'$ is heavy, then $t'$ is $\frac{\alpha}{8\log^2 (n/\varphi)}$-DC-heavy. 
Indeed, if $t'$ is heavy, then
\begin{align*}
 2 \sum_{x \in B_{t'}} \rank(x) \deg(x) & \geq t'\cdot n_{t'}\cdot r_{t'}  & \text{by Equation \eqref{eqn:rdnd}} \\
& \geq \frac{2t}{4 \log^2 (n/\varphi)}  r_t ^2 & \text{because $r_t' \geq n_t'  \geq \frac{r_t}{2 \log (n/\varphi)}$ and $t'\geq 2t$}\\
&  \geq  \frac{1}{4 \log^2 (n/\varphi)}\sum_{x \in B_t} \rank(x) \deg(x) & \text{by Equation \eqref{eqn:rdnd}, since $r_t \geq n_t$} \\
& \geq \frac{2\alpha}{4 \log^2 (n/\varphi)}\sum_{x \in V} \rank(x) \deg(x)& \text{by the assumption that $t$ is $\alpha$-DC heavy}.
\end{align*}
So if $t'$ is heavy, then it is $\frac{\alpha}{8 \log^2 (n/\varphi)}$-DC heavy, and in particular, by Lemma \ref{lemma:nhat}, $\hat{n}_{t'} \geq \frac{1}{2}n_t$. 
We now have 
\begin{align*}
\hat{r}_t & =  \sum_{t' \geq t: t' \text{ is light}}\hat{n}_{t'} +  \sum_{t' \geq t: t' \text{ is heavy}}\hat{n}_{t'}  & \\
& \geq \frac{1}{2} \sum_{t' \geq t: t' \text{ is heavy}}n_{t'} +  \sum_{t' \geq t: t' \text{ is light}}\hat{n}_{t'}&  \text{by Lemma \ref{lemma:nhat}}\\
& =  \frac{1}{2}\sum_{t' \geq t}n_t + \sum_{t' \geq t: t' \text{ is light}}\left(\hat{n}_{t'} -\frac{1}{2}n_{t'}\right)&  \\
& \geq \frac{1}{2}\sum_{t' \geq t}n_t - \frac{1}{2}\sum_{t' \geq t: t' \text{ is light}}n_{t'}& \\
& \geq \frac{1}{2}\sum_{t' \geq t}n_t - \frac{\log (n/\varphi)}{4 \log (n/\varphi)}r_{t} & \text{by definition of light classes} \\
& = \frac{1}{4}r_t. & 
\end{align*}
\end{proof}
We are now ready to prove Lemma~\ref{lemma:est_i_di}.
\begin{proof}[Proof of Lemma~\ref{lemma:est_i_di}]
Let $\alpha = \frac{1}{4\log (n/\varphi)}$. Let $\hat{n}$ be the estimator from Lemma \ref{lemma:nhat}  with parameter $\alpha$ and let $\hat{r}$ be the rank estimator from Lemma \ref{lemma:rhat} with parameter $\alpha$. Let 
\[v = 32 \sum_{d} \hat{r}_d \cdot \hat{n}_d \cdot d.\]
Condition on the success of the estimators $\hat{n}$ and $\hat{r}$ (which happens with probability at least $1-n^{-101}$). Then, for each degree class $d$, we have
\begin{align*}
  d \cdot \hat{r}_d \cdot \hat{n}_d & \leq 36 d\cdot r_d \cdot n_d  & \text{by Lemma \ref{lemma:nhat} and Lemma \ref{lemma:rhat}}\\
  & \leq 72 \sum_{x \in B_d} \deg(x) \rank(x), & \text{by Equation \eqref{eqn:rdnd}.}
\end{align*}
Summing over all degree classes $d$, we have 
\begin{align*}
    v  = 32 \sum_{d} \hat{r}_d \cdot \hat{n}_d \cdot d  \leq 32 \cdot 72 \sum_{x \in V} \rank(x) \deg(x), 
\end{align*} which gives the upper bound. 
It remains to prove the lower-bound. Say that a degree class is \emph{heavy} if it is $\frac{1}{4\log (n/\varphi)}$-DC heavy, and say that it is \emph{light} otherwise. We have
\begin{align*}
    \sum_{x \in V} \rank(x) \deg(x) &= \sum_d \sum_{x \in B_d} \rank(x) \deg(x) &  \\
    & = \sum_{d: B_d \text{ is light }}\sum_{x \in B_d} \rank(x) \deg(x) +  \sum_{d: B_d \text{ is heavy } } \sum_{x \in B_d} \rank(x) \deg(x) & \\
    & \leq \frac{\log(n/\varphi)}{2 \log(n/\varphi)} \sum_{x \in V}\rank(x) \deg(x) +   \sum_{d: B_d \text{ is heavy } } \sum_{x \in B_d} \rank(x) \deg(x)& \\
    & \leq \frac{1}{2}\sum_{x \in V}\rank(x) \deg(x) + \sum_{d:B_d \text{ is heavy}} 2d \cdot n_d \cdot r_d  & \text{by Equation \eqref{eqn:rdnd}}\\
    & \leq  \frac{1 }{2} \sum_{x \in V} \rank(x) \deg(x) +  16 \sum_{d: B_d \text{ is heavy } }d \cdot \hat{n}_d \cdot \hat{r}_d &\text{by Lemmas \ref{lemma:nhat} and \ref{lemma:rhat}}\\
   &  \leq \frac{1}{2}\sum_{x \in V} \rank(x) \deg(x) +  \frac{1}{2}v. & 
\end{align*}
Rearranging, we obtain the lower-bound. 
\end{proof}

\subsection{Correctness of \textsc{ClusterCost} and \textsc{TotalClustersCost} (Proof of Theorem~\ref{thm:singleclustercost} and Theorem ~\ref{thm:Gk})} \label{sec:clustercost}
In this section we put everything together to prove Theorem \ref{thm:singleclustercost}. Furthermore, we also present the 
procedure \textsc{TotalClustersCost} for approximating the contribution of the clusters to the Dasgupta cost to a $d$-regular $(k, \varphi, \epsilon)$-clusterable graph, and we prove its guarantee (Theorem~\ref{thm:Gk}).  
\singleclustercost*

\begin{proof}[Proof of Theorem \ref{thm:singleclustercost}]
	Follows immediately from Lemma~\ref{lemma:est_i_di}, Lemma ~\ref{lem:dcost:lb}, Lemma~\ref{lem:dcost:ub} and Theorem~\ref{thm:MSun21:thm3}.  
\end{proof}

Next, we present the procedure \textsc{TotalClustersCost} for approximating the contribution of the clusters to the Dasgupta cost to a $d$-regular $(k, \varphi, \epsilon)$-clusterable graph. A natural approach is to run the \textsc{ClusterCost} procedure on each cluster and sum the output. This would use $\approx n^{1/3}$ seeds. However, since we assume that the graph is $d$-regular, we can do something much simpler. 

By virtue of Lemma \ref{lem:dcost:lb} and Lemma \ref{lem:dcost:ub}, we want to approximate the quantity $\sum \rank(v) \cdot \deg(v)$ for each cluster $C_i$. However, since the graph is $d$-regular, we get 
$$\sum_{v \in C_i} \rank_{G\{C_i\}}(v) \deg_{G\{C_i\}}(v) = \sum_{v \in C_i}\rank_{G\{C_i\}}(v) \cdot d  = \sum_{j =1}^{|C_i|} j \cdot d \approx |C_i|^2 \cdot d.$$
Furthermore, by the assumption that $\max_{i,j}|C_i|/|C_j| =  O(1)$, we have $|C_i| = O(n/k)$ for all $i$, so the total contribution from the cluster simplifies further as
$$\sum_{i \in [k]}\sum_{v \in C_i} \rank_{G\{C_i\}}(v) \deg_{G\{C_i\}}(v) \approx \sum_{i \in [k]} |C_i|^2 \cdot d   \approx k \cdot \frac{n^2}{k^2} \cdot d = \frac{d \cdot n^2}{k}.$$
Motivated by this, our procedure $\textsc{TotalClustersCost}$ below simply outputs the number $\frac{d \cdot n^2}{k}$. 

\begin{algorithm}[H]
	\caption{\textsc{TotalClustersCost}($G$)}
	\label{alg:clustercost}
	\begin{algorithmic}[1]
		\STATE\return $O \left( \frac{d \cdot n^2}{k}\right)$
	\end{algorithmic}
\end{algorithm}
\begin{remark} Algorithm \ref{alg:clustercost} uses the approximation $|C_i| \approx n/k$ for all $i$, which introduces a dependence on the parameter $\eta \coloneqq \max_{i,j} \frac{|C_i|}{|C_j|}$ in the approximation guarantee. Since we assume $\eta = O(1)$, this is enough for our purposes. However, if desired, one can easily remove this dependence by approximating each $|C_i|$ more accurately via sampling vertices with their cluster labels.
\end{remark}

\totalclustercost*
\begin{proof}
Let $v = 3 \cdot \eta^2 \cdot \frac{d \cdot n^2}{k}$ be the value output by Algorithm \ref{alg:clustercost}, where $\eta = \max_{i,j}|C_i|/|C_j| =  O(1)$. 

For every $i \in [k]$, let $\mathcal T_i$ denote the output of Algorithm \ref{alg:HCdeg} on input $G\{C_i\}$ . 
We have
\begin{align*}
\text{COST}_{G\{C_i\}}(T^*_i) & \leq \text{COST}_{G\{C_i\}}(\mathcal T_i)  && \text{by optimality of $T^*_i$} \\
& \leq 2 \sum_{v \in C_i} \rank_{G\{C_i\}}(v) \cdot \deg_{G\{C_i\}}(v) && \text{by Lemma  \ref{lem:dcost:ub}} \\
& =2 \sum_{j =1}^{|C_i|}j \cdot d && \text{since $\deg_{G\{C_i\}}(v) = \deg_G(v) = d$ for all $v \in C_i$} \\
& = 2 \cdot \frac{|C_i|(|C_i|+1)}{2}\cdot d \\
& \leq 3 \eta^2 \frac{n^2}{k^2}\cdot d && \text{since $|C_i| \leq \eta \cdot \frac{n}{k}$}. 
\end{align*}
Summing over $i \in [k]$ and recalling that $\textsc{TotalClustersCost}(G) = 3 \eta^2 \frac{n^2}{k}\cdot d$, yields the lower bound. 
Next, we prove the upper bound. We have 
\begin{align*}
\text{COST}_{G\{C_i\}}(T^*_i)&  \geq \Omega(\varphi^4)\ \text{COST}_{G\{C_i\}}(\mathcal T_i) && \text{by Theorem \ref{thm:MSun21:thm3}} \\
& \geq  \Omega(\varphi^5)\ \sum_{v \in C_i} \rank_{G\{C_i\}}(v) \cdot \deg_{G\{C_i\}}(v)  && \text{by Lemma \ref{lem:dcost:ub}} \\
& = \Omega(\varphi^5) \sum_{j = 1}^{|C_i|}j \cdot d && \text{since $\deg_{G\{C_i\}}(v) = \deg_G(v) = d$ for all $v \in C_i$}\\
& =  \Omega(\varphi^5)  \cdot \frac{|C_i|(|C_i|+1)}{2}\cdot d &&  \\
& = \Omega(\varphi^5)\cdot \frac{n^2}{k^2}\cdot d && \text{since $|C_i| \geq \frac{n}{k \cdot \eta } = \Omega\left(\frac{n}{k} \right)$}
\end{align*}
Summing over $i \in [k]$ and recalling that $\textsc{TotalClustersCost}(G) = 3 \eta^2 \frac{n^2}{k}\cdot d = O\left( \frac{n^2}{k} \cdot d \right)$, yields the upper bound. 
\end{proof}
We need the following Observation in the proof of Theorem \ref{thm:estdcost}.
\begin{observation}
    \label{lem:lw-cost-exp}
	It holds that $\textsc{TotalClustersCost}(G) \geq \Omega\left(\frac{n^2}{k}\right)$. This is because the value output by Algorithm \ref{alg:clustercost} is given by $v = 3 \cdot \eta^2 \cdot \frac{d \cdot n^2}{k} =  \Omega\left(\frac{n^2}{k}\right)$, since $d, \eta  \geq 1$ . 
\end{observation}

\subsection{Lower bound on the necessary number of seeds (Proof of Theorem \ref{thm:lw-Gk})}\label{sec:tightrunningtime}
In this subsection, we prove Theorem \ref{thm:lw-Gk} from Section \ref{sec:expandercost-overview}, which shows that the query complexity of \textsc{ClusterCost} is tight. 

\unionclusterableruntimetight*

\begin{proof}
	We will construct the two graphs on the same vertex set $V$. Pick a set $C \subseteq V$ of size $\frac{n^{2/3}}{2}$. Let $b = \frac{|V \setminus C|} {|C|} = 2n^{1/3}-1$. 
 Construct the graph $G$ as follows: 
 \begin{itemize}
 \item $C$ forms a clique
 \item Add a perfect $b$-matching between $C$ and $V \setminus C$. 
 \end{itemize}
Then every vertex in $C$ has degree $\frac{n^{2/3}}{2}-1+b$, and every vertex in $V \setminus C$ has degree $1$. Therefore, 
 \[ \sum_{i=1}^n i \cdot d_i  \leq \sum_{i=1}^{\frac{n^{2/3}}{2}} i \cdot (n^{2/3}+b) + \sum_{i = \frac{n^{2/3}}{2}+1}^n i \leq \frac{1}{2}n^{2/3} \cdot (n^{2/3})^2+n^2/2 \leq n^2. \]
Construct $G'$ as follows: 
 \begin{itemize}
 \item $C$ forms a clique
 \item Add a $(2 \alpha \cdot b, 2\alpha)$-regular bipartite graph between $C$ and $V \setminus C$. 
 \item For each vertex in $C$, delete $b(2 \alpha -1)$ of its edges internal edges in $C$. 
 \end{itemize}
Now every vertex in $C$ has degree $\frac{n^{2/3}}{2}-1+b$, but every vertex in $V \setminus C$ has degree $2 \alpha$.  Therefore, 
\[ \sum_{i=1}^n i \cdot d'_i \geq \sum_{i=1}^{\frac{n^{2/3}}{2}} i \cdot \frac{n^{2/3}}{2} + \sum_{i = \frac{n^{2/3}}{2}+1}^n i \cdot 2 \alpha   \geq \frac{n^{2/3}}{2} \cdot \left(\frac{n^{2/3}}{2} \right)^2  + 2 \alpha \cdot \frac{n^2}{2} \geq \alpha n^2. \]

Every vertex in $C$ has the same degree in $G$ and $G'$, so to distinguish between the two graphs, we need to query a vertex in $V \setminus C$. 
 
  We have $|E_G|, |E_{G'}| \geq \frac{|C|^2}{2}   = \frac{n^{4/3}}{8}.$
  In $G$, the probability that a given query returns a vertex in $V\setminus C$, is $\frac{|V \setminus C|}{2|E_G|} \leq \frac{4n}{n^{4/3}} = 4n^{-1/3}$.  Similarly, in $G'$, the probability that a given query returns a vertex in $V\setminus C$, is $\frac{|V \setminus C|\cdot 2\alpha}{2|E_G'|} \leq  \frac{8 \alpha n}{n^{4/3}}= 8 \alpha n^{-1/3}$. 
  
  Now suppose that the number of queries is at most $\frac{n^{1/3}}{8 \alpha}$. If the true input graph is $G$, then with probability at least $(1-n^{-1/3})^{\frac{n^{1/3}}{8 \alpha}} \geq \frac{1}{3}$, we fail to query any vertices in $V \setminus C$. Similarly, if the true input graph graph is $G'$, then with probability at least $(1-8\alpha n^{-1/3})^{\frac{n^{1/3}}{8 \alpha}} \geq \frac{1}{3}$, we fail to query any vertices in $V \setminus C$. So with probability at least $\frac{1}{3}$, we fail to distinguish the graphs. 
\end{proof}

\section{Correctness of \textsc{WeightedDotProductOracle} (Proof of Theorem \ref{thm:wdp})}
\label{sec:wdp:oracle}

{\bf Obtaining the weighted dot product oracle:} 
Recall that we aim 
to spectrally approximate
$\cL_H$ by $\wt{\cL}$ with probability at least $1-n^{-100}$ (as in Equation \eqref{eq:l-lh-spec}) and this amounts to approximating all quadratic
forms on $\cL_H$ with quadratic forms on $\wt{\cL}$. We achieve this by getting  estimates for $\rdp{f_x, \Sigma_{[k]} f_y}$ to be accurate with probability at
least $1 - n^{-100k}$. The details are presented below.

\thmweighteddot*

\begin{remark} \label{rem:eta:in:runningtime}
	This result is similar to Theorem 2 in \cite{GluchKLMS21}. The difference being Theorem 2 in 
	\cite{GluchKLMS21} approximates dot product between the embedding vectors $\rdp{f_x, f_y}$. 
	Here, we instead want to approximate the weighted dot product $\rdp{f_x, \Sigma_{[k]} f_y}$.
\end{remark}

We also need to set up some notation which is used in this section. Let $m\leq n$ be integers.  
For any matrix $A\in \R^{n\times m}$ with singular value decomposition (SVD) 
$A=Y\Gamma Z^T$ we assume $Y\in \R^{n\times n}$ and $Z\in \R^{m\times n}$ are orthogonal 
matrices and  $\Gamma \in \R^{n\times n}$ is a diagonal matrix of singular values. Since 
$Y$ and $Z$ are orthogonal matrices, their columns form an orthonormal basis. For any 
integer $q\in[m]$ we denote $Y_{[q]}\in \R^{n \times q}$ as the first $q$ columns of $Y$ 
and $Y_{-[q]}$ to denote the matrix of the remaining columns of $Y$.  We also denote 
$Z^T_{[q]}\in \R^{q \times n}$ as the first $q$ rows of $Z^T$ and $Z^T_{-[q]}$ to denote 
the matrix of the remaining rows of $Z$. Finally we denote $\Gamma^T_{[q]}\in \R^{q \times q}$ 
as the first $q$ rows and columns of $\Gamma$ and we use $\Gamma_{-[q]}$ as the last $n-q$ 
rows and columns of $\Gamma$. So for any $q\in[m]$ the span of $Y_{-[q]}$ is the orthogonal 
complement of the span of $Y_{[q]}$, also the span of $Z_{-[q]}$ is the orthogonal 
complement of the span of $Z_{[q]}$. Thus we can write 
$A=Y_{[q]}\Gamma_{[q]}Z^T_{[q]} + Y_{-[q]}\Gamma_{-[q]}Z^T_{-[q]}$.

\begin{algorithm}[H]
	\caption{\textsc{InitializeWeightedDotProductOracle($G,\xi$)}  
 } \label{alg:LearnEmbedding-tt}
\label{alg:LearnEmbedding}
\begin{algorithmic}[1]
	\STATE $t\gets \frac{20\cdot \log n}{\varphi^2}$	
	\STATE $R_{\text{init}}\gets n^{1/2 + O(\epsilon / \varphi^2)} \cdot  \left(\frac{k\cdot \log n}{\xi}\right)^{O(1)}$	
	
	\STATE $s\gets n^{O(\epsilon / \varphi^2)} \cdot \left(\frac{k\cdot \log n}{\xi}\right)^{O(1)}$ \label{ln:sets} 
	
	\STATE  $I_S\gets $ $s$ indices chosen independently and uniformly at random with replacement from $\{1,\ldots,n\}$ \label{ln:sample}
	
	\STATE ${\Q} \gets \textsc{EstimateTransitionMatrix}(G,I_S,R_{\text{init}},t)$  \textit{\qquad   \# $\Q$ has at most $R_{\text{init}}\cdot s$ non-zeros \label{ln:setQi}}
	\STATE $\G\gets$\textsc{EstimateCollisionProbabilities}$(G,I_S,R_{\text{init}},t)$ \label{ln:setG}

	\STATE Compute  eigendecomposition $\frac{n}{s}\cdot\G:=\widehat{W}\widehat{\Sigma} \widehat{W}^T$  of $\frac{n}{s}\cdot\G$ \textit{\qquad\qquad \qquad \qquad\qquad \# $\G\in \R^{s\times s}$ \label{ln:svdG}}
	\STATE $\Psi\gets \frac{n}{s}\cdot\widehat{W}_{[k]}\widehat{\Sigma}_{[k]}^{-2} \widehat{W}_{[k]}^T$ \textit{\qquad \qquad\qquad \qquad\qquad \qquad\qquad \qquad\qquad \qquad \qquad \qquad \# $\Psi\in\R^{s\times s}$} 
	\STATE \return   $\cD_w:=\{\Psi,\Q\}$
\end{algorithmic}
\end{algorithm}

\begin{algorithm}[H]
	\caption{$\textsc{WeightedDotProductOracle}(G,x,y,  \xi, \mathcal{D})$  \textit{\qquad \qquad  \# $\mathcal{D}_w:=\{\Psi,\Q\}$}
}
\label{alg:weighted-dot}
\begin{algorithmic}[1]
	\STATE $R_{\text{query}}\gets n^{1/2 + O(\epsilon / \varphi^2)} \cdot  \left(\frac{k\cdot \log n}{\xi}\right)^{O(1)}$
				
		\STATE ${\m_x\gets \textsc{RunRandomWalks}(G,R_{\text{query}},t+1,x)}$
		 \textit{\quad \# unlike in~\cite{GluchKLMS21},  walk length  $=t+1$}
		\STATE ${\m_y\gets \textsc{RunRandomWalks}(G,R_{\text{query}},t,y)}$ 
	\STATE  \return $\adp{f_{x},\Sigma_{[k]} f_y}:=(\m_x^T \Q) \Psi (\Q^T\m_y)$  
\end{algorithmic}
\end{algorithm}

We build up on a collection of tools from \cite{GluchKLMS21}. First, we use Lemma \ref{lem:v-close} 
which shows that $(k,\varphi,\epsilon)$-clusterable graphs, the outer products of the columns of 
the $t$-step random walk transition matrix has small spectral norm. This holds because 
the matrix power dominates by the first $k$ eigenvectors and each of them has bounded infinity norm.

\begin{restatable}[A higher success probability version of Lemma 23 from \cite{GluchKLMS21} with improved estimation error]{lemmma}{lemvclose}\label{lem:v-close}
Let $k \geq 2$ be an integer, $\varphi \in (0,1)$ and $\epsilon\in (0,1)$. Let $G=(V,E)$ be a $d$-regular 
graph that admits a $(k,\varphi,\epsilon)$-clustering $C_1, \ldots C_k$. 
Let $M$ be the random walk transition matrix of $G$. Let $1> \xi>1/n^5$,
$t\geq \frac{20 \log n}{\varphi^2}$. Let $c>1$ be a large enough constant and let $s\ge  c\cdot \mycolor{k^8} \cdot n^{(400 \epsilon / \varphi^2)}\cdot \log n /\xi^2$.  Let $I_S=\{i_1,\ldots, i_s\}$ be a multiset of $s$ indices chosen independently and uniformly at random from
$\{1,\dots,n\}$. Let $S$ be the $n\times s$ matrix whose $j$-th column equals $\mathds{1}_{i_j}$. Suppose that $M^t=U\Sigma^tU^T$ is the eigendecomposition of $M^t$ and $\sqrt{\frac{n}{s}} \cdot M^tS=\widetilde{U}\widetilde{\Sigma}\widetilde{W}^T$ is  the SVD of $\sqrt{\frac{n}{s}} \cdot M^tS$. 
If $\frac{\epsilon}{\varphi^2}\leq \frac{1}{10^5}$ then with probability at least $1-n^{-100\cdot k}$ we have
$$
\left|\left|U_{[k]} {\Sigma}_{[k]}^{-2t} U_{[k]}^T - \widetilde{U}_{[k]} \widetilde{\Sigma}_{[k]}^{-2} \widetilde{U}_{[k]}^T \right|\right|_2 < \xi
$$
\end{restatable}

The following lemma from \cite{GluchKLMS21} is instrumental in analyzing collision probabilities
of random walks from every vertex $x \in V$ in a $(k, \varphi, \e)$-clusterable graph.

\begin{restatable}[\cite{GluchKLMS21}]{lemmma}{lemMtbnd}\label{lem:Mt-bnd}
Let $k \geq 2$ be an integer, $\varphi \in (0,1)$ and $\epsilon\in (0,1)$. Let $G=(V,E)$ be a 
$d$-regular and that admits a $(k,\varphi,\epsilon)$-clustering  $C_1, \ldots , C_k$.
Let $M$ be the random walk transition matrix of $G$. For any $t\geq \frac{20\log n}{\varphi^2}$ 
and any $x\in V$ we have 
\[\|M^{t}\mathds{1}_{x}\|_2 \leq O\left(k \cdot n^{-1/2+20\epsilon /\varphi^2}\right) \text{.}\]
\end{restatable}

To prove the correctness of weighted dot product of spectral embedding of vertices, we use a similar proof strategy to \cite{GluchKLMS21},  
which, albeit, develops an 
estimator for the \emph{unweighted} dot product  between spectral embeddings i.e., $\rdp{f_x, f_y}$. In Lemma \ref{lem:bnd-e1}, we show that the weighted dot product of
spectral embeddings i.e.,  $\rdp{f_x, \Sigma_{[k]}f_y}$ can be estimated by the appropriate linear transformation of the random walk transition matrix. Since we seek weighted dot products unlike \cite{GluchKLMS21},  we run a $t$-step random walk from $x$, and a $t+1$-step walk from $y$. The one-step longer walk helps us to inject the matrix of eigenvalues in between the dot product of spectral embedding of vertex $x$ and $y$.

\begin{restatable}{lemmma}{lembndeone}
\label{lem:bnd-e1}
Let $k \geq 2$ be an integer, $\varphi \in (0,1)$ and $\epsilon\in (0,1)$. Let $G=(V,E)$ be a $d$-regular 
graph that admits a $(k,\varphi,\epsilon)$-clustering $C_1, \ldots C_k$. Let $M$ be  the random walk 
transition matrix of $G$. Let  $1/n^5 < \xi < 1$, $t\geq  \frac{20\log n}{\varphi^2}$. Let $c>1$ be a 
large enough constant and let 
$s\geq c
\cdot n^{480\epsilon / \varphi^2}\cdot \log n \cdot \mycolor{k^{13}}/({\xi}^2) $. 
Let $I_S=\{i_1,\ldots, i_s\}$ be a multiset of $s$ indices chosen independently and uniformly at random from
$\{1,\dots,n\}$. Let $S$ be the $n\times s$ matrix whose $j$-th column equals $\mathds{1}_{i_j}$.  
Let $M^t=U\Sigma^tU^T$ be the eigendecomposition of $M^t$ and 
$\sqrt{\frac{n}{s}} \cdot M^tS=\widetilde{U}\widetilde{\Sigma}\widetilde{W}^T$  be the SVD of 
$\sqrt{\frac{n}{s}} \cdot M^tS$. If $\frac{\epsilon}{\varphi^2}\leq \frac{1}{10^5}$ then with probability 
at least $1-n^{-100\cdot k}$ we have
\[\left|\mathds{1}_x^T U_{[k]}\Sigma_{[k]}{U}_{[k]}^T  \mathds{1}_y - (M^{t+1}\mathds{1}_{x})^T  (M^tS)\left(\frac{n}{s}\cdot\widetilde{W}_{[k]} \widetilde{\Sigma}^{-4}_{[k]} \widetilde{W}^T_{[k]}\right) (M^tS)^T (M^{t}\mathds{1}_{y}) \right| \leq  \frac{\xi}{\mycolor{n k^2}}\text{.}
\]
\end{restatable}

\begin{proof}
Let $m_x=M^{t+1}\mathds{1}_{x}$ and $m_y=M^{t}\mathds{1}_{y}$.  Let $c'>1$ be a large enough constant we will set later. Let $\xi' = \frac{\xi}{c'\cdot \mycolor{k^4} \cdot n^{40 \e/\varphi^2}}$. Let $c_1$ be the constant in front of $s$ in Lemma  \ref{lem:v-close}. Thus for large enough $c$ we have $s\geq c
\cdot n^{480\epsilon / \varphi^2}\cdot \log n \cdot \mycolor{k^{13}}/({\xi}^2) \geq c_1\cdot n^{400\epsilon / \varphi^2}\cdot \log n \cdot \mycolor{k^8}/({\xi'}^2)$, and therefore by Lemma \ref{lem:v-close} applied with 
$\xi'$,  with probability at least $1-n^{-100\cdot k}$ we have
\[\left|\left|U_{[k]} {\Sigma}_{[k]}^{-2t} U_{[k]}^T - \widetilde{U}_{[k]} \widetilde{\Sigma}_{[k]}^{-2} \widetilde{U}_{[k]}^T \right|\right|_2\leq \xi' \]
By Cauchy-Schwarz and submultiplicativity of the spectral norm we have
\begin{align}
\left|m_x^T U_{[k]} {\Sigma}_{[k]}^{-2t} U_{[k]}^Tm_y - m_x^T\widetilde{U}_{[k]} \widetilde{\Sigma}_{[k]}^{-2} \widetilde{U}_{[k]}^T m_y \right| 
&\leq  \left|\left|U_{[k]} {\Sigma}_{[k]}^{-2t} U_{[k]}^T - \widetilde{U}_{[k]} \widetilde{\Sigma}_{[k]}^{-2} \widetilde{U}_{[k]}^T \right| \right|_2\|m_x\|_2\|m_y\|_2 \nonumber\\
&\leq \xi' \|m_x\|_2\|m_y\|_2\label{eq:pxpyv-til-close}
\end{align}
In the rest of the proof we will show $m_{x}^T (U_{[k]} {\Sigma}_{[k]}^{-2t} U_{[k]}^T)m_{y}=\mathds{1}_x^T U_{[k]}\Sigma_{[k]} U_{[k]}^T \mathds{1}_y$ ({\bf Step 1}) and $m_x^T  \widetilde{U}_{[k]}\widetilde{\Sigma}^{-2}_{[k]} \widetilde{U}_{[k]}^T m_y = m_x^T  (M^tS)(\widetilde{W}_{[k]} \widetilde{\Sigma}^{-4}_{[k]} \widetilde{W}^T_{[k]}) (M^tS)^T m_y $ ({\bf Step 2}), and finally obtain the result by combining these facts with~\eqref{eq:pxpyv-til-close} and the upper bound on $\|m_x\|_2$ provided by Lemma~\ref{lem:Mt-bnd}.

\paragraph{Step $1$:} Note that $M^t=U\Sigma^tU^T$. Therefore we get $M^{t+1}\mathds{1}_x= U\Sigma^{t+1}U^T  \mathds{1}_x$, and $M^t\mathds{1}_y= U\Sigma^tU^T  \mathds{1}_y$. Thus we have
\begin{equation}
\label{eq:pxto1x}
m_x^T U_{[k]} {\Sigma}_{[k]}^{-2t} U_{[k]}^Tm_y = \mathds{1}_x^T \left(\left(U\Sigma^{t+1}U^T \right) \left(U_{[k]} {\Sigma}_{[k]}^{-2t} U_{[k]}^T\right) \left(U\Sigma^tU^T \right) \right)\mathds{1}_y
\end{equation}
Note that $U^T U_{[k]}$ is a $n\times k$ matrix such that the top $k\times k$ matrix is $I_{k\times k}$ and the rest is zero. Also $U_{[k]}^T U$ is a $k\times n$ matrix such that the left $k\times k$ matrix is $I_{k\times k}$ and the rest is zero. Therefore we have 
\[U{\Sigma}^{t+1}\left(U^T U_{[k]}\right) {\Sigma}^{-2t}_{[k]} \left(U^T_{[k]} U\right){\Sigma}^{t}U^T = UHU^T \text{,}\]
where $H$ is a $n\times n$ matrix such that the top left $k\times k$ matrix is $\Sigma_{k\times k}$ and the rest is zero. Hence, we have
\[ U H U^T =U_{[k]} \Sigma_{[k]} U_{[k]}^T\text{.}\]
Thus we have
\begin{equation}
\label{eq:p-vk}
m_{x}^T (U_{[k]} {\Sigma}_{[k]}^{-2t} U_{[k]}^T)m_{y} = \mathds{1}_x^T U_{[k]}\Sigma_{[k]} U_{[k]}^T \mathds{1}_y
\end{equation}

\paragraph{Step $2$:} We have $\sqrt{\frac{n}{s}} \cdot M^tS=\widetilde{U}\widetilde{\Sigma}\widetilde{W}^T$ where $\widetilde{U}\in \R^{n\times n}$, $\widetilde{\Sigma}\in \R^{n\times n}$ and $\widetilde{W}\in \R^{s\times n}$. Therefore, 
\begin{align}
\label{eq:pxpy}
 &(m_x)^T  (M^tS)\left(\frac{n}{s}\cdot\widetilde{W}_{[k]} \widetilde{\Sigma}^{-4}_{[k]} \widetilde{W}^T_{[k]}\right) (M^tS)^T (m_{y}) \nonumber\\
 &= m_x^T  \left(\sqrt{\frac{s}{n}}\cdot \widetilde{U}\widetilde{\Sigma}\widetilde{W}^T\right)\left(\frac{n}{s}\cdot\widetilde{W}_{[k]} \widetilde{\Sigma}^{-4}_{[k]} \widetilde{W}^T_{[k]}\right) \left(\sqrt{\frac{s}{n}}\cdot\widetilde{W}\widetilde{\Sigma}\widetilde{U}^T\right) m_y \nonumber\\
 &= m_x^T  \left( \widetilde{U}\widetilde{\Sigma}\widetilde{W}^T\right)\left(\widetilde{W}_{[k]} \widetilde{\Sigma}^{-4}_{[k]} \widetilde{W}^T_{[k]}\right) \left(\widetilde{W}\widetilde{\Sigma}\widetilde{U}^T\right) m_y
\end{align}
Note that $\widetilde{W}^T\widetilde{W}_{[k]}$ is a $n\times k$ matrix such that the top $k\times k$ matrix is $I_{k\times k}$ and the rest is zero. Also $\widetilde{W}_{[k]}^T\widetilde{W}$ is a $k\times n$ matrix such that the left $k\times k$ matrix is $I_{k\times k}$ and the rest is zero. Therefore we have 
\[\widetilde{\Sigma}\left(\widetilde{W}^T\widetilde{W}_{[k]}\right) \widetilde{\Sigma}^{-4}_{[k]} \left(\widetilde{W}^T_{[k]} \widetilde{W}\right)\widetilde{\Sigma} = \widetilde{H} \text{,}\]
where $\widetilde{H}$ is a $n\times n$ matrix such that the top left $k\times k$ matrix is $\widetilde{\Sigma}^{-2}_{[k]}$ and the rest is zero. Hence, we have
\begin{equation}
\label{eq:vsigv}
(\widetilde{U}\widetilde{\Sigma}\widetilde{W}^T)\left(\frac{n}{s}\cdot\widetilde{W}_{[k]} \widetilde{\Sigma}^{-4}_{[k]} \widetilde{W}^T_{[k]}\right) (\widetilde{W}\widetilde{\Sigma}\widetilde{U}^T) = \widetilde{U}\widetilde{H}\widetilde{U}^T = \widetilde{U}_{[k]}\widetilde{\Sigma}^{-2}_{[k]} \widetilde{U}_{[k]}^T
\end{equation}
Putting \eqref{eq:vsigv} and \eqref{eq:pxpy} together we get
\begin{equation}
\label{eq:vtild}
 m_x^T  (M^tS)(\widetilde{W}_{[k]} \widetilde{\Sigma}^{-4}_{[k]} \widetilde{W}^T_{[k]}) (M^tS)^T m_y = m_x^T  \widetilde{U}_{[k]}\widetilde{\Sigma}^{-2}_{[k]} \widetilde{U}_{[k]}^T m_y
\end{equation}
\paragraph{Putting it together.} By \eqref{eq:pxpyv-til-close}, \eqref{eq:p-vk} and \eqref{eq:vtild} we have
\begin{align}
&\left|m_x^T  (M^tS)\left(\frac{n}{s}\cdot\widetilde{W}_{[k]} \widetilde{\Sigma}^{-4}_{[k]} \widetilde{W}^T_{[k]}\right) (M^tS)^T m_y - \mathds{1}_x^T U_{[k]}{U}_{[k]}^T  \mathds{1}_y \right| \nonumber \\
&= \left| m_x^T\widetilde{U}_{[k]} \widetilde{\Sigma}_{[k]}^{-2} \widetilde{U}_{[k]}^T m_y - m_x^T U_{[k]} {\Sigma}_{[k]}^{-2t} U_{[k]}^Tm_y\right| && \text{By \eqref{eq:p-vk} and \eqref{eq:vtild}}\nonumber \\
& \leq  \xi' \cdot\|m_x\|_2\|m_y\|_2 && \text{By \eqref{eq:pxpyv-til-close}} 
\end{align}
By Lemma \ref{lem:Mt-bnd} for any vertex $x\in V$ we have 
\begin{equation}
\|m_x\|^2_2=\|M^{t}\mathds{1}_{x}\|^2_2 \leq O\left(k^2 \cdot n^{-1+40\epsilon /\varphi^2}\right)  \text{.}
\end{equation}
Therefore by choice of $c'$ as a large enough constant and choosing $\xi' = \frac{\xi}{c'\cdot \mycolor{k^4} \cdot n^{40 \e/\varphi^2}}$ we have
\begin{equation}
\label{eq:e2-done}
\left|m_x^T  (M^tS)\left(\frac{n}{s}\cdot\widetilde{W}_{[k]} \widetilde{\Sigma}^{-4}_{[k]} \widetilde{W}^T_{[k]}\right) (M^tS)^T m_y - \mathds{1}_x^T U_{[k]}\Sigma_{[k]}{U}_{[k]}^T  \mathds{1}_y \right|\leq O\left(\xi'\cdot k^2 \cdot n^{-1+40\epsilon /\varphi^2}\right) \leq  \frac{\xi}{\mycolor{nk^2}} \text{.}
\end{equation}
\end{proof}

Finally, Lemma \ref{lem:u-abs-close} bounds the absolute deviation between $\adp{f_x, \Sigma_{[k]} f_y}$
and our estimator. We put the two together using triangle inequality to prove Theorem 
\ref{thm:wdp} 

\begin{restatable}[A higher success probability version of Lemma 29 from \cite{us} with improved estimation error.]{lemmma}{lemuabsclose}\label{lem:u-abs-close}
	Let $G=(V,E)$ be a $d$-regular that admits a $(k,\varphi,\epsilon)$-clustering $C_1, \ldots C_k$. Let
	$1/n^5 < \xi < 1$. Let $\mathcal{D}$ denote the  data structure constructed by the procedure \textsc{InitializeOracle($G,\delta,\xi$)} (Algorithm \ref{alg:LearnEmbedding}). 
	Let $x,y\in V$. Let $\adp{f_x, \Sigma_{[k]}f_y}\in \R$ denote the value returned by the procedure 
	$\textsc{WeightedDotProductOracle}(G,x,y,  \xi, \mathcal{D})$ (Algorithm \ref{alg:weighted-dot}).
	Let  $t\geq  \frac{20\log n}{\varphi^2}$. Let $c>1$ be a large enough constant and let 
	$s\geq c\cdot n^{240\cdot\epsilon / \varphi^2}\cdot \log n \cdot k^{4}$. Let 
	$I_S=\{i_1,\ldots, i_s\}$ be a multiset of $s$ indices chosen independently and uniformly at random from
	$\{1,\dots,n\}$. Let $S$ be the $n\times s$ matrix whose $j$-th column equals $\mathds{1}_{i_j}$. 
	Let $M$ be the random walk transition matrix of $G$. Let 
	$\sqrt{\frac{n}{s}} \cdot M^tS=\widetilde{U}\widetilde{\Sigma}\widetilde{W}^T$  be  the SVD of 
	$\sqrt{\frac{n}{s}} \cdot M^tS$. If $\frac{\epsilon}{\varphi^2}\leq \frac{1}{10^5}$, and Algorithm 
	\ref{alg:LearnEmbedding} succeeds,  then with probability at least $1-n^{-100k}$ we have
	\[\left|\adp{f_{x},\Sigma_{[k]} f_y} - (M^{t+1}\mathds{1}_{x})^T  (M^tS)\left(\frac{n}{s}\cdot\widetilde{W}_{[k]} \widetilde{\Sigma}^{-4}_{[k]} \widetilde{W}^T_{[k]}\right) (M^tS)^T (M^{t}\mathds{1}_{y}) \right|<
		\frac{\xi }{\mycolor{nk^2}} \text{.}
	\]
\end{restatable}

\begin{remark}
	The result in Gluch et al. above, obtains a success probability of at least $1 - n^{-100}$. It
	can be improved to $1 - n^{-100k}$ with an overhead of $poly(k)$ times as many samples.
\end{remark}

We now prove Theorem~\ref{thm:wdp}

\begin{proof}[Proof of Theorem \ref{thm:wdp}]
\textbf{Correctness:}
Let $s= \varTheta{(n^{480\epsilon / \varphi^2}\cdot \log n \cdot \mycolor{k^{13}}/(\xi^2))}$. Recall that $I_S=\{i_1,\ldots, i_s\}$ is the multiset of $s$ vertices each sampled uniformly at random (see line 3 of Algorithm \ref{alg:LearnEmbedding}).  Let $S$ be the $n\times s$ matrix whose $j$-th column equals $\mathds{1}_{i_j}$. Recall that $M$ is the random walk transition matrix of $G$. Let $\sqrt{\frac{n}{s}} \cdot M^tS=\widetilde{U}\widetilde{\Sigma}\widetilde{W}^T$  be  the eigendecomposition of $\sqrt{\frac{n}{s}} \cdot M^tS$. We define
\[e_1=\left|(M^{t+1}\mathds{1}_{x})^T  (M^tS)\left(\frac{n}{s}\cdot\widetilde{W}_{[k]} \widetilde{\Sigma}^{-4}_{[k]} \widetilde{W}^T_{[k]}\right) (M^tS)^T (M^{t}\mathds{1}_{y}) -  \mathds{1}_x^T U_{[k]}\Sigma_{[k]}{U}_{[k]}^T  \mathds{1}_y \right| \]
and 
\[e_2=\left| \adp{f_{x},\Sigma_{[k]} f_y} - (M^{t+1}\mathds{1}_{x})^T  (M^tS)\left(\frac{n}{s}\cdot\widetilde{W}_{[k]} \widetilde{\Sigma}^{-4}_{[k]} \widetilde{W}^T_{[k]}\right) (M^tS)^T (M^{t}\mathds{1}_{y}) \right| \]
By triangle inequality we have 
\[\left|\adp{f_{x},\Sigma_{[k]}f_y}  - \langle f_x, \Sigma_{[k]}f_y \rangle \right|  = \left|\adp{f_{x},\Sigma_{[k]}f_y} - \mathds{1}_x^T U_{[k]}\Sigma_{[k]}{U}_{[k]}^T  \mathds{1}_y \right| \leq e_1 + e_2 \text{.}\]
Let $\xi'=\xi/2$.  Let $c$ be a constant in front of $s$ in Lemma \ref{lem:bnd-e1} and  $c'$ be a 
constant in front of $s$ in Lemma \ref{lem:u-abs-close}. Recall, line 3 of Algorithm 
\ref{alg:LearnEmbedding} sets $s= \varTheta{(n^{480\epsilon / \varphi^2}\cdot \log n \cdot \mycolor{k^{13}}/(\xi^2))}$. 

Since $\frac{\epsilon}{\varphi^2}\leq \frac{1}{10^5}$ and $s\geq c\cdot n^{480\epsilon / \varphi^2}\cdot \log n \cdot \mycolor{k^{13}}/(\xi'^2)$, by Lemma \ref{lem:bnd-e1} with probability at least $1-n^{-100\cdot k}$ we have $e_1\leq \frac{\xi'}{\mycolor{nk^2}}=\frac{\xi}{\mycolor{2\cdot nk^2}} \text{.}$ 
Since  
$s \geq c'\cdot n^{240\epsilon / \varphi^2}\cdot \log n \cdot k^{4}$, by 
Lemma \ref{lem:u-abs-close}, with probability at least $1-2\cdot n^{-100\cdot k}$ we have 
$e_2 \leq \frac{\xi}{\mycolor{2\cdot nk^2}}\text{.}$ Thus with probability at least $1-3\cdot n^{-100\cdot k}$ we 
have 
\[\left| \adp{f_{x},\Sigma_{[k]}f_y}  - \langle f_x, \Sigma_{[k]}f_y \rangle \right|  \leq e_1 + e_2 \leq \frac{\xi}{\mycolor{2\cdot n k^2}}+\frac{\xi}{\mycolor{2\cdot nk^2}} \leq \frac{\xi}{\mycolor{nk^2}}\text{.}\]

\textbf{Running time of $\textsc{InitializeOracle}$:}
The algorithm first samples a 
set $I_S$. Then, as per line \ref{ln:setQi} of Algorithm \ref{alg:LearnEmbedding}, it 
estimates the empirical probability distribution of $t$-step random walks starting from 
any vertex $x\in I_S$. 
The $\textsc{EstimateTransitionMatrix}$ procedure runs $R_{\text{init}}$ random walks of length  $t$ from each vertex $x\in I_S$. So it takes $O(\log n\cdot s\cdot R_{\text{init}} \cdot t)$ time 
and requires $O(\log n\cdot s\cdot R_{\text{init}})$ space to store endpoints of random walks. Then as per line 6 of Algorithm \ref{alg:LearnEmbedding} it estimates matrix $\G$ such that 
the entry corresponding to the $x^\text{th}$ row and $y^{\text{th}}$ column of $\G$ is an 
estimation of pairwise collision probability of random walks starting from $x,y \in I_S$. To 
compute $\G$ we call Algorithm \textsc{EstimateCollisionProbabilities}($G,I_S,R_{\text{init}},t$) 
(from \cite{GluchKLMS21}) for $O(\log n)$ times. This procedure takes 
$O(s\cdot R_{\text{init}}\cdot t \cdot \log n)$ time and it requires $O(s^2\cdot \log n)$ space to 
store matrix $\G$. Computing the SVD of $\G$ (done in line 7 of Algorithm \ref{alg:LearnEmbedding}) 
takes time $O(s^3)$. Thus overall Algorithm \ref{alg:LearnEmbedding} runs in 
time $O \left(\log n\cdot s\cdot R_{\text{init}} \cdot t+ s^3\right)$. Thus, by choice of 
$t= \varTheta\left(\frac{\log n}{\varphi^2}\right)$, 
$R_{\text{init}}=\varTheta{(n^{1/2+ O(\epsilon / \varphi^2)}  \cdot \log^{O(1)} n \cdot k^{O(1)}/(\xi)^{O(1)})}$ 
and $s= \varTheta(n^{O(\epsilon / \varphi^2)}\cdot (\log n)^{O(1)} \cdot k^{O(1)}/{\xi}^{O(1)})$  as in 
Algorithm \ref{alg:LearnEmbedding} we get that Algorithm \ref{alg:LearnEmbedding} runs in time 
$\log n\cdot s\cdot R_{\text{init}} \cdot t+ s^3=(\frac{k\cdot \log n}{\xi\cdot \varphi})^{O(1)}\cdot n^{1/2+O(\epsilon/\varphi^2)}$ 
and returns a data structure of size 
$O \left(s^2+\log n\cdot s \cdot R_{\text{init}}\right)=(\frac{k\cdot \log n}{\xi})^{O(1)}\cdot n^{1/2+O(\epsilon/\varphi^2)}\text{.}$

\textbf{Running time of $\textsc{WeightedDotProductOracle}$:} 
Algorithm  $\textsc{WeightedDotProductOracle}$ runs $R_{\text{query}}$ random walks of length $t, t+1$ from vertex $x$ and vertex $y$, then it computes $(\m_x^T  \widehat{Q})$ and $(\widehat{Q}^T \m_y)$. Since $\widehat{Q}\in \R^{n\times s}$ has $s$ columns and since $\m_x$ has at most $R_{\text{query}}$ non-zero entries, thus one can compute $\m_x^T \cdot \widehat{Q}$ in time $R_{\text{query}}\cdot s$. Finally Algorithm \ref{alg:weighted-dot}  returns value $(\m_x^T \widehat{Q}) \Psi (\widehat{Q}^T \m_y)$.  Since $(\m_x^T \widehat{Q}), (\widehat{Q}^T \m_y)\in \R^{s}$ and $\Psi\in \R^{s\times s}$ one can compute  $(\m_x^T \widehat{Q}) \Psi (\widehat{Q}^T \m_y)$ in time $O(s^2)$. Thus overall Algorithm \ref{alg:weighted-dot} takes $O\left(t\cdot R_{\text{query}} + s\cdot R_{\text{query}}  + s^2\right)$ time. Thus, by choice of $t= O\left(\frac{\log n}{\varphi^2}\right)$	, $R_{\text{query}}=n^{1/2+O(\epsilon / \varphi^2)}  \cdot \left(\frac{k}{\xi} \right)^{O(1)}$ and $s= n^{O(\epsilon / \varphi^2)}\cdot (\frac{k\cdot \log n}{\xi})^{O(1)}$  we get that the  Algorithm  \ref{alg:weighted-dot} runs in  time $(\frac{k\cdot \log n}{\xi\cdot \varphi})^{O(1)}\cdot n^{1/2+O(\epsilon/\varphi^2)}$.
\end{proof}

\end{document}